\documentclass[12pt,letterpaper,notitlepage]{article}
\usepackage{eurosym}
\usepackage{amsmath}
\usepackage{float}
\usepackage[left=1in,
            right=1in,
            top=1in,
            bottom=1in]{geometry}
\usepackage[footnotesize]{caption}
\usepackage{subcaption}
\usepackage{color}
\usepackage[pdftex]{graphicx}
\usepackage{graphicx}
\usepackage{psfrag}
\usepackage{Here}
\usepackage{multirow}
\usepackage{pgf}
\usepackage{color}
\usepackage{pgf,pgfarrows,pgfnodes,pgfautomata,pgfheaps,pgfshade}
\usepackage[T1]{fontenc}
\usepackage{color}
\usepackage{amsfonts}
\usepackage{amssymb, amsmath, amsthm,bm}
\usepackage{lscape}
\usepackage{natbib}
\usepackage{enumerate}
\usepackage{bbm}
\usepackage{lscape}
\usepackage{diagbox}
\usepackage{comment}
\usepackage{makecell}
\usepackage{titling}
\usepackage{setspace}
\usepackage{rotating}
\usepackage{mathtools}
\usepackage{xr-hyper}
\usepackage{afterpage}
\usepackage{lipsum}
\usepackage{enumitem}
 \usepackage[colorlinks,linkcolor={red},citecolor={blue},urlcolor={red},filecolor={red}]{hyperref}
\allowdisplaybreaks

\numberwithin{equation}{section}

\newtheorem{theorem}{Theorem}[section]

\newtheorem{assumption}{Assumption}

\newtheorem{lemma}{Lemma}[section]

\newtheorem{proposition}{Proposition}[section]
\newtheorem{remark}{Remark}

\makeatletter
\newcommand{\specialcell}[1]{\ifmeasuring@#1\else\omit$\displaystyle#1$\ignorespaces\fi}
\makeatother

\title{A Two-Way Cluster-Robust Variance Estimator for Quantile Regression\thanks{We are grateful for the helpful comments provided by Antonio F. Galvao, Carlos Lamarche, Harold Chiang, and Yuya Sasaki. All the remaining errors are ours.} }
\author{ {\large \textbf{Ulrich Hounyo}\thanks{%
Department of Economics, University at Albany -- State University of New
York, Albany, NY 12222, United States.}} \and {\large \textbf{Jiahao Lin}%
\thanks{%
Department of Economics, University at Albany -- State University of New
York, Albany, NY 12222, United States.} }}

\makeatother

\begin{document}
\maketitle 
\begin{abstract}
 This paper studies inference for linear quantile regression with two-way clustered data. Using a separately exchangeable array framework and a projection decomposition of the quantile score, we characterize regime-dependent convergence rates and establish a self-normalized Gaussian approximation. We propose a two-way cluster-robust sandwich variance estimator with a kernel-based density ``bread'' and a projection-matched ``meat'', and prove consistency and validity of inference in Gaussian regimes. We also show an impossibility result for uniform Gaussian inference in a non-Gaussian interaction regime, clarifying the limits of two-way cluster-robust inference for quantile regression. Monte Carlo simulations show that the proposed procedures provide more reliable size control than conventional alternatives. An empirical application to teacher-licensing restrictions suggests that licensing stringency is negatively associated with right-tail teacher quality at lower conditional quantiles, while the association is weak and statistically insignificant at higher quantiles.


\end{abstract}
%\date{February, 2023}

\vspace*{-0.5cm}

\vfill

\thispagestyle{empty} \pagebreak

\section{Introduction}
\label{sec:introduction}

Understanding how covariates shape the entire distribution of outcomes, rather than just the mean, is central to modern empirical economics. Quantile regression (QR), introduced by \citet{koenker1978regression}, provides a leading framework for this purpose. In many applications, however, observations are indexed by multiple clustering dimensions and may exhibit dependence along each of them. Canonical examples include worker--firm matched data, exporter--destination panels, and teacher--student interactions, where latent shocks induce correlation within rows and columns of a two-way array. While inference for QR under one-way clustering is now well understood, extending valid inference to multiway clustered environments remains an open and substantively important problem.

This paper develops asymptotic theory and feasible inference for linear QR under two-way clustering. Using a separately exchangeable array framework, we show that inference is governed by a projection decomposition of the score into row, column, and interaction components, yielding regime-dependent convergence rates. We propose a two-way cluster-robust variance estimator that adapts to these regimes and establish its validity under Gaussian limits, while also showing that uniform inference is impossible when limits are non-Gaussian.

The difficulty stems from several features that arise only when non-smooth estimation is combined with multiway dependence. First, unlike mean regression, the non-smooth quantile score complicates
uniform control of score fluctuations in neighborhoods of $\beta_{0}(\tau)$. Second, the Jacobian depends on the conditional density at the target quantile and must be estimated nonparametrically, introducing additional bias and sampling variability. Third, two-way clustering fundamentally alters the structure of the empirical process: the score cannot be reduced to a sum of independent or weakly dependent terms along a single dimension. Finally, the strength of dependence across clustering dimensions may vary, so that the rate of convergence of the estimator is not fixed but depends on the underlying dependence regime. These four challenges are not additive. In our setting, non-smooth
scores and kernel Jacobian estimation must be handled \emph{simultaneously} with
regime-dependent rates and genuinely two-way dependence, requiring a uniform
analysis of both the score and the Jacobian that remains valid across
dependence regimes.

These features place our setting outside the scope of existing multiway inference methods. In particular, although \citet{davezies2018asymptotic} develop variance estimation procedures for multiway clustered GMM estimators, their approach relies on a plug-in Jacobian that presumes knowledge of the population derivative. In quantile regression, by contrast, the Jacobian depends on an unknown conditional density and must be estimated nonparametrically, rendering their approach not directly applicable. Moreover, unlike the setting emphasized in \citet{davezies2018asymptotic}, we do not impose nondegeneracy of the asymptotic variance. Instead, the rate of convergence is allowed to vary with the strength of clustering, and inference must remain valid across such regimes.

We formalize these ideas within a separately exchangeable array framework based on the Aldous--Hoover--Kallenberg representation \citep{aldous1981representations,hoover1979relations,kallenberg1989representation}, which has become a standard device for analyzing multiway dependent data \citep[e.g.,][]{davezies2021empirical,menzel2021bootstrap,chiang2023standard,graham2024sparse}. The resulting projection decomposition reveals that the asymptotic behavior of the estimator is \emph{regime-dependent}: the convergence rate and limiting distribution are determined by the dominant projection component.

Building on this insight, we propose a feasible two-way cluster-robust variance estimator (CRVE) of sandwich form,
\[
\widehat{\Sigma}(\tau)
= \widehat{D}(\tau)^{-1}\,\widehat{\Omega}(\tau)\,\widehat{D}(\tau)^{-1}.
\]
The ``bread'' $\widehat{D}(\tau)$ is a kernel-based estimator of the conditional density at the target quantile, adapted to accommodate two-way clustering, while the ``meat'' $\widehat{\Omega}(\tau)$ aggregates row, column, and interaction contributions in a manner that mirrors the projection structure of the score. The proposed estimator is designed to adapt to the underlying dependence regime and is shown to be valid whenever a Gaussian approximation holds.

A further implication of our framework is that Gaussian approximations need not hold uniformly. When the interaction component of the score dominates and the row and column components are weak, the limiting distribution of the estimator can be non-Gaussian. In this case, we show that uniform consistency of inference is impossible over a natural class of data-generating processes. This impossibility result highlights a fundamental limitation of inference under two-way clustering and clarifies the conditions under which standard asymptotic methods can be relied upon.

Our analysis relates to and complements several strands of the literature. Recent work has emphasized that multiway dependence can generate nonstandard asymptotic behavior even in simple settings: \citet{menzel2021bootstrap} show that sample means may exhibit non-Gaussian limits under two-way clustering, while \citet{chiang2024extremal} study extremal quantiles under such dependence, focusing on rare-event behavior. Their results demonstrate that extremal quantiles can remain robust even in degenerate regimes. Our paper complements this line of work by focusing on \emph{interior quantiles}: whereas \citet{chiang2024extremal} analyze $\widehat{\beta}(\tau)$ as $\tau \to 0$, we consider fixed $\tau \in (0,1)$, where the non-smooth score and the interaction of clustering dimensions generate fundamentally different asymptotic behavior.

More broadly, our results contribute to the literature on quantile regression under dependence \citep{kato2012asymptotic,parente2016quantile,hagemann2017cluster} and to recent advances in dependence-robust covariance estimation \citep{galvao2024hac}. Compared with these studies, our framework is tailored to two-way clustering dependence, while also encompassing the i.i.d.\ and one-way clustering settings as special cases. In particular, by accommodating degeneracy and regime-dependent convergence rates, our analysis goes beyond standard settings and provides inference procedures that remain valid across a broad range of two-way clustering regimes.

We complement the theoretical results with Monte Carlo evidence demonstrating that conventional QR standard errors can severely understate uncertainty under two-way clustering, whereas the proposed CRVE delivers reliable coverage across a wide range of dependence configurations.\footnote{The MATLAB and Stata codes to implement the proposed method are available at \url{https://jiahaoecon.github.io/webpage/research/}.} An empirical application revisits the relationship between teacher-licensing stringency and the supply of high-quality teachers, uncovering substantial heterogeneity across the outcome distribution that is masked by mean-based analysis.

The remainder of the paper is organized as follows. Section~\ref{sec:model} introduces the model and develops the asymptotic theory. Section~\ref{sec:variance} presents the proposed variance estimator and establishes its validity. Section~\ref{sec:mc} reports Monte Carlo evidence, Section~\ref{sec:empirical} presents the empirical application, and Section~\ref{sec:conclusion} concludes. Mathematical derivations are provided
in the Appendix.

\section{Two-Way Clustering in Quantile Regression}

\label{sec:model} 

\subsection{Model Setting}
Let $\{(y_{ghi},X_{ghi}^{\top}):g=1,\dots,G,\;h=1,\dots,H,\;i=1,\dots,N_{gh}\}$
denote a two-way clustered array of observations, where $y_{ghi}\in\mathbb{R}$
is the scalar outcome and $X_{ghi}\in\mathbb{R}^{d}$ is a vector of regressors. The indices $g$ and $h$ label clusters along two dimensions, so that $(g,h)$ identifies a cell (e.g., unit $\times$ time). Let $N_{gh}$ denote the number of observations in cell $(g,h)$.  
%The first index $g$ identifies the cluster in the first dimension
%(the $g$-cluster), and the second index $h$ identifies the cluster
%in the second dimension (the $h$-cluster). The index pair $(g,h)$
%therefore labels a cell formed by the intersection of a $g$-cluster
%and an $h$-cluster (e.g., unit $\times$ time). Let $N_{gh}\in\mathbb{N}$
%denote the number of observations within cell $(g,h)$.

Fix a quantile index $\tau\in(0,1)$. We consider the quantile regression
model 
\begin{equation}
Q_{y_{ghi}}(\tau\vert X_{ghi})=X_{ghi}^{\top}\beta_{0}(\tau),\qquad g=1,\dots,G,\;h=1,\dots,H,\;i=1,\dots,N_{gh},\label{eq:qr-model}
\end{equation}
where $Q_{y_{ghi}}(\tau\vert X_{ghi})$ denotes the conditional $\tau$-quantile
of $y_{ghi}$ given $X_{ghi}$. The quantile error $e_{ghi}(\tau)$
is defined as $e_{ghi}(\tau):=y_{ghi}-X_{ghi}^\top\beta_{0}(\tau)$. Let
$\rho_{\tau}(u):=u\bigl(\tau-\mathbf{1}\{u\le0\}\bigr)$ denote the
check loss. For two-way clustered data, the QR estimator solves the following convex objective function
\[
\hat{\beta}(\tau):=\arg\min_{\beta\in\Theta}\frac{1}{\sum_{g,h}N_{gh}}\sum_{g=1}^{G}\sum_{h=1}^{H}\sum_{i=1}^{N_{gh}}\rho_{\tau}\!\bigl(y_{ghi}-X_{ghi}^{\top}\beta\bigr),
\]
with respect to $\beta\in\Theta\subset\mathbb{R}^{d}$, where $\Theta$
is compact.

For later use, define the quantile score 
\begin{equation}
\psi_{ghi}(\beta,\tau):=X_{ghi}\Bigl(\tau-\mathbf{1}\{y_{ghi}\le X_{ghi}^{\top}\beta\}\Bigr),\qquad\Psi_{ghi}(\tau):=\psi_{ghi}\bigl(\beta_{0}(\tau),\tau\bigr).\label{eq:qr-score}
\end{equation}
The function $\Psi_{ghi}(\tau)$ is nonlinear due to the indicator
function, which plays a central role in the asymptotic analysis. For each cell $(g,h)$, let $X_{gh}$ be the $N_{gh}\times d$ matrix
with $i^{\text{th}}$ row $X_{ghi}$, and let $y_{gh}$ and
$e_{gh}(\tau)$ be the corresponding $N_{gh}\times 1$ vectors with
$i^{\text{th}}$ elements $y_{ghi}$ and $e_{ghi}(\tau)$.  We impose the conditional quantile restriction $Q_{e_{ghi}}(\tau\vert X_{ghi})=0$, i.e., the conditional $\tau$-quantile of $e_{ghi}(\tau)$ given $X_{ghi}$ equals zero.  For simplicity, we focus on the case where each cell contains exactly
one observation, that is, $N_{gh}=1$ for all $g,h$, and suppress
the replicate index $i$. Extensions to heterogeneous $N_{gh}$ are
provided in Internet Appendix IB.

To model two-way dependence, we adopt the Aldous--Hoover--Kallenberg (AHK) representation, which provides a canonical framework for separately exchangeable arrays; see, for example, \citet{davezies2021empirical,mackinnon2021wild,chiang2023standard}.

\begin{assumption}[Two-way clustered data with the AHK representation]
\label{ass:ahk} There exist measurable functions $\Gamma$ such that
\[
(y_{gh},X_{gh})=\Gamma(U_{g},V_{h},W_{gh}),
\]
where $\{U_{g}\}_{g\ge1}$, $\{V_{h}\}_{h\ge1}$, and $\{W_{gh}\}_{g,h\ge1}$
are mutually independent sequences of i.i.d.\ random variables. Without
loss of generality, each latent variable is uniformly distributed
on $[0,1]$. The function $\Gamma$ may vary with $(G,H)$, allowing for triangular-array sequences of DGPs.
 \end{assumption}

Under Assumption~\ref{ass:ahk}, the array $(y_{gh},X_{gh})$ is separately exchangeable across $(g,h)$, and hence identically distributed, although generally dependent. The quantile index $\tau$ affects the model only through the conditional quantile restriction and does not enter the regressor process. There
exists a measurable function $\Psi(U,V,W;\tau)$ such that $\Psi_{gh}(\tau)=\Psi(U_{g},V_{h},W_{gh};\tau).$
The score then admits the Hoeffding type decomposition 
\begin{equation}
\Psi_{gh}(\tau)=E[\Psi_{gh}(\tau)]+\Psi^{(\mathrm{I})}(U_{g},\tau)+\Psi^{(\mathrm{II})}(V_{h},\tau)+\Psi^{(\mathrm{III})}(U_{g},V_{h},\tau)+\Psi^{(\mathrm{IV})}(U_{g},V_{h},W_{gh},\tau),\label{eq:anova_decomposition_main}
\end{equation}
where 
\begin{align*}
\Psi^{(\text{I})}(U_{g},\tau) & :=E[\Psi_{gh}(\tau)\vert U_{g}]-E[\Psi_{gh}(\tau)],\\
\Psi^{(\text{II})}(V_{h},\tau) & :=E[\Psi_{gh}(\tau)\vert V_{h}]-E[\Psi_{gh}(\tau)],\\
\Psi^{(\text{III})}(U_{g},V_{h},\tau) & :=E[\Psi_{gh}(\tau)\vert U_{g},V_{h}]-E[\Psi_{gh}(\tau)]-\Psi^{(\text{I})}(U_{g},\tau)-\Psi^{(\text{II})}(V_{h},\tau),\\
\Psi^{(\text{IV})}(U_{g},V_{h},W_{gh},\tau) & :=\Psi_{gh}(\tau)-E[\Psi_{gh}(\tau)\vert U_{g},V_{h}].
\end{align*}

This decomposition follows from $L^{2}$ projection theory for separately
exchangeable arrays and is unique in $L^{2}$. It isolates the distinct sources of dependence—row, column, and interaction components—which play a central role in determining the asymptotic behavior of the estimator. Closely
related decompositions for nonlinear statistics under AHK dependence
have been developed recently for U-statistics on bipartite and row--column
exchangeable arrays; see \citet{le2025hoeffding}. When convenient,
we write $\Psi_\bullet^{(j)}$ for $\Psi^{(j)}(\cdot;\tau)$, $j=\text{I},\ldots,\text{IV}$.
We suppress the dependence on $\tau$ to conserve space.

By construction, $E[\Psi^{(j)}_\bullet]=0$ for each $j$ and $E[\Psi_{\bullet}^{(j)}\Psi_{\bullet}^{(j')\top}]=0$
for $j\neq j'$. Although $(U_{g},V_{h},W_{gh})$ are independent,
the components $\Psi_{g}^{(\text{I})}$, $\Psi_{h}^{(\text{II})},$
$\Psi_{gh}^{(\text{III})}$, and $\Psi_{gh}^{(\text{IV})}$ need not
be. These components are, however, pairwise orthogonal in $L^{2}$,
which suffices to characterize asymptotic variances and limit distributions.

\subsection{Asymptotic Distribution}

\label{sec:asymptotics} For $j\in\{\text{I},\text{II},\text{III},\text{IV}\}$,
define the component variances 
\[
\sigma_{j,\Gamma}^{2}:=E\!\left[\Psi^{(j)}_{\bullet}\Psi^{(j)\top}_\bullet\right].
\]
The subscript $\Gamma$ emphasizes that these quantities depend on
the underlying DGP and may vary with $(G,H)$. To simplify notation, we suppress the explicit $(G,H)$ dependence.

A standard argument yields the Bahadur representation 
\begin{align}\label{eq:bahadur}
\hat{\beta}-\beta_{0}(\tau)
= D(\tau)^{-1}\,\bar{\Psi}_{GH} + o_{P}\!\left(\bigl\|\hat{\beta}-\beta_{0}(\tau)\bigr\|\right),
\end{align}
where 
\[
D(\tau):=E\!\left[f_{e\vert X}(0\vert X_{gh})\,X_{gh}X_{gh}^{\top}\right],\qquad\bar{\Psi}_{GH}:=\frac{1}{GH}\sum_{g=1}^{G}\sum_{h=1}^{H}\Psi_{gh}.
\]
Using \eqref{eq:anova_decomposition_main}, we decompose $\bar{\Psi}_{GH}$
as 
\begin{align*}
\bar{\Psi}_{GH} & =\frac{1}{G}\sum_{g=1}^{G}\Psi_{g}^{(\text{I})}+\frac{1}{H}\sum_{h=1}^{H}\Psi_{h}^{(\text{II})}+\frac{1}{GH}\sum_{g=1}^{G}\sum_{h=1}^{H}\Bigl(\Psi_{gh}^{(\text{III})}+\Psi_{gh}^{(\text{IV})}\Bigr)\\
 & :=\bar{\Psi}^{(\text{I})}+\bar{\Psi}^{(\text{II})}+\bar{\Psi}^{(\text{III})}+\bar{\Psi}^{(\text{IV})}.
\end{align*}
Observe that the arrays
$\{\Psi_g^{(\mathrm{I})}\}_{g=1}^G$ and
$\{\Psi_h^{(\mathrm{II})}\}_{h=1}^H$
are i.i.d. across clusters,
and, conditional on $\{U_g,V_h\}$,
$\{\Psi_{gh}^{(\mathrm{IV})}\}_{g,h}$ are independent across cells. Consequently, after appropriate normalization,
the sums associated with $\bar{\Psi}^{(\text{I})}$, $\bar{\Psi}^{(\text{II})}$,
and $\bar{\Psi}^{(\text{IV})}$ are asymptotically Gaussian, whereas
$\bar{\Psi}^{(\text{III})}$ may admit a non-Gaussian limit.

For each $j$, we impose a common-order restriction on the eigenvalues of $\sigma_{j,\Gamma}^{2}$ and use $\sigma_{j,1\Gamma}^{2}$ only as a convenient shorthand for the representative magnitude. Let $\lambda_{\max}(\cdot)$ and $\lambda_{\min}(\cdot)$ denote the maximum and minimum eigenvalues of the input matrix.

\begin{assumption}[Homogeneous order]\label{as:homogeneous}
For each nonzero component \(j\in\{\mathrm{I},\mathrm{II},\mathrm{III},\mathrm{IV}\}\),
\(
\lambda_{\max}(\sigma_{j,\Gamma}^{2})/\lambda_{\min}(\sigma_{j,\Gamma}^{2})=O(1).
\)
\end{assumption}

Assumption \ref{as:homogeneous} can be relaxed (but not removed) to
allow the diagonal elements of the variance components $\sigma_{j,\Gamma}^2$ to have heterogeneous orders across coordinates, as we demonstrate in Internet Appendix IC. We nevertheless maintain the current form in the main text for expositional simplicity. This restriction is standard and has also been imposed, either implicitly or explicitly, in the recent two-way clustering literature; see, for example, \citet{mackinnon2021wild}, \citet{chiang2023standard}, and Assumption 5 of \citet{davezies2025analytic}. In particular, an innovative Example 2 in \citet{davezies2025analytic} further demonstrates why such a restriction is needed: in its absence, the usual least-squares approximation can fail under two-way clustered dependence.

The same consideration arises in our quantile setting. To see this, note that the rate of a given coordinate of ${D}(\tau)^{-1}\bar\Psi_{GH}$ need not coincide with the rate of the corresponding coordinate of $\widehat{D}^{-1}\bar\Psi_{GH}$, because slow components may cancel only after multiplication by $D(\tau)^{-1}$. Moreover, when coordinates operate on different scales, this cancellation need not be reproduced by the feasible analogue based on $\widehat D^{-1}$.

For a simple illustration, consider median regression with $X_{gh}=(1,\alpha_g)^\top$ and $u_{gh}=\xi_h$, where $E[\alpha_g]=\mu$ and $\mathrm{Median}(\xi_h)=0$. Let
\[
\psi_h:=\frac12-1\{\xi_h\le 0\},\qquad
\bar\psi_H:=\frac1H\sum_{h=1}^H\psi_h,\qquad
\bar\alpha_G:=\frac1G\sum_{g=1}^G\alpha_g.
\]
Then
\(
\bar\Psi_{GH}
=\frac1{GH}\sum_{g=1}^G\sum_{h=1}^H X_{gh}\psi_h
=
(
\bar\psi_H, 
\bar\alpha_G\,\bar\psi_H
)^\top.
\)
Moreover, the population Jacobian is
\(
D(\tau)=f_\xi(0)
\begin{pmatrix}
1 & \mu\\
\mu & E[\alpha_g^2]
\end{pmatrix},
\)
which implies
\(
\bigl[D(\tau)^{-1}\bar\Psi_{GH}\bigr]_2
=
\frac{(\bar\alpha_G-\mu)\bar\psi_H}{f_\xi(0)Var(\alpha_g)}.
\)
By contrast, the feasible Jacobian is
\[
\widehat D
=
\widehat f_\xi(0)
\begin{pmatrix}
1 & \bar\alpha_G\\
\bar\alpha_G & \overline{\alpha_G^2}
\end{pmatrix},
\qquad
\overline{\alpha_G^2}:=\frac1G\sum_{g=1}^G\alpha_g^2,
\]
with $\widehat f_\xi(0)\overset{P}{\to}f_\xi(0)$, and therefore
\(
\bigl[\widehat D^{-1}\bar\Psi_{GH}\bigr]_2=0.
\)
Hence, the second coordinate of $D(\tau)^{-1}\bar\Psi_{GH}$ is proportional to $(\bar\alpha_G-\mu)\bar\psi_H$, whereas the corresponding coordinate of $\widehat D^{-1}\bar\Psi_{GH}$ vanishes exactly. Thus, once heterogeneous componentwise rates are allowed, $D(\tau)^{-1}\bar\Psi_{GH}$ and $\widehat D^{-1}\bar\Psi_{GH}$ need not share the same first-order behavior.

 Let $f_{e\vert X}(e\vert x)$ denote the conditional density of $e_{gh}$
given $X_{gh}=x$. $f_{e\vert X}^{(1)}(e\vert x)$ and $f_{e\vert X}^{(2)}(e\vert x)$
denote the corresponding first and second derivatives, respectively.
We impose a natural two-way array analogue of the standard moment,
smoothness, and nonsingularity conditions used in i.i.d. quantile
regression.

\begin{assumption}[Moments, smoothness, and nonsingularity]\label{as:moment}
(i) $E[\Psi_{gh}]=0$, $E\|X_{gh}\|^{q}<\infty$ for some $q\ge5$, $\sup_{U_g,V_h}E\bigl(\Vert X_{gh}\Vert^{2}\vert U_{g},V_{h}\bigr)<\infty$, and $E(X_{gh}X_{gh}^{\top})$
is nonsingular. (ii) There exists $\varrho>0$ such that whenever $\sigma_{j,\Gamma}^{2}>0$,  $E\|\sigma_{j,\Gamma}^{-1} \Psi^{(j)}\|^{4+\varrho}<\infty$ for each $j\in\{\mathrm{I},\mathrm{II},\mathrm{III},\mathrm{IV}\}$. (iii) The map $e\mapsto f_{e\vert X}(e\vert x)$ is
twice continuously differentiable for every $x$, and $\sup_{e,x}\big|f_{e\vert X}(e\vert x)\big|<\infty$
as well as $\sup_{e,x}\big|f_{e\vert X}^{(1)}(e\vert x)\big|<\infty$.
(iv) The conditional density at zero is uniformly bounded away from
zero: $\inf_{x}f_{e\vert X}(0\vert x)>0$.  (v) $\beta_{0}(\tau)$ lies in the interior
of a compact parameter space $\Theta$. \end{assumption} 

Let
\(
r_{GH}:=\min\left\{
\frac{G}{\sigma_{\mathrm{I},1\Gamma}^{2}},
\frac{H}{\sigma_{\mathrm{II},1\Gamma}^{2}},
GH
\right\},
\)
with ratios involving zero denominators interpreted as \(+\infty\). Then \(r_{GH}\) is the effective sample size for \(\bar{\Psi}_{GH}\), in the sense that \(\mathrm{Var}(\bar{\Psi}_{GH})\) is of order \(r_{GH}^{-1}\). Let
\(
N_{0,GH}:=
\sum_{g=1}^{G}\sum_{h=1}^{H}
\mathbf 1\{y_{gh}=X_{gh}^{\top}\widehat\beta\}
\)
denote the number of observations lying exactly on the fitted
quantile-regression hyperplane. 

\begin{assumption}[Active hyperplane condition]\label{as:hyperplane}
For the same $q$ in Assumption \ref{as:moment}(i), \(
N_{0,GH}
=
o_p\left((GH)^{1-1/q}r_{GH}^{-1/2}\right).
\)
\end{assumption}
This condition is automatically satisfied under the usual general-position
condition for quantile regression. In particular, when the relevant distributions are continuous,
a basic quantile-regression solution has at most \(d\) observations lying
exactly on the fitted hyperplane with probability 1; see, e.g.,
\citet{gutenbrunner1992regression}. The condition is much weaker than this
standard requirement. It also allows cluster-induced multiplicities of exact
fits. For example, suppose $q=5$ and there is perfect dependence along the \(G\)-dimension,
so that a fitted hyperplane may contain \(O(G)\)
observations. This case is still allowed by Assumption \ref{as:hyperplane}
provided
\(
G=o\left((GH)^{4/5}r_{GH}^{-1/2}\right).
\)
In this case the effective sample size is \(r_{GH}=H\). Hence the above
condition becomes
\(
G=o(H^{3/2}).
\)
This becomes a mild growth restriction. The case of perfect dependence along the
\(H\)-dimension is analogous.

Let the asymptotic variance of $\hat{\beta}$ be 
\[
\Sigma_{GH}:=D(\tau)^{-1}\,\Omega_{GH}(\tau)\,D(\tau)^{-1},\qquad\Omega_{GH}(\tau):=Var\!\left(\bar{\Psi}_{GH}\right).
\]
By orthogonality of the ANOVA components, 
\begin{equation}
\Omega_{GH}(\tau)=\frac{1}{GH}\Bigl(H\sigma_{\text{I},\Gamma}^{2}+G\sigma_{\text{II},\Gamma}^{2}+\sigma_{\text{III},\Gamma}^{2}+\sigma_{\text{IV},\Gamma}^{2}\Bigr).\label{eq: variance of score}
\end{equation}

\begin{assumption}[Orders of variance components]\label{as:order of variance}
Along any subsequence indexed by $(G_n,H_n)$ for which
$\bigl(H_n\sigma_{\mathrm{I},1\Gamma}^{2},G_n\sigma_{\mathrm{II},1\Gamma}^{2},\sigma_{\mathrm{III},1\Gamma}^{2},\sigma_{\mathrm{IV},1\Gamma}^{2}\bigr)$
converges in $[0,\infty]^{4}$, at least one of the following holds:
(i) $H_n\sigma_{\mathrm{I},1\Gamma}^{2}+G_n\sigma_{\mathrm{II},1\Gamma}^{2}\to\infty$, or  \[(ii)\quad \sigma_{\mathrm{III},1\Gamma}^{2}\to 0 \quad\text{ and }\quad
\lim_{n\to\infty}\Bigl(H_n\sigma_{\mathrm{I},1\Gamma}^{2}+G_n\sigma_{\mathrm{II},1\Gamma}^{2}+\sigma_{\mathrm{IV},1\Gamma}^{2}\Bigr)\;>\;0.
\]

\end{assumption}

Assumption \ref{as:order of variance} imposes one of two sufficient routes to asymptotic normality. First, clustering along at least one dimension may be sufficiently strong so that the Gaussian components
\(\bar{\Psi}^{(\mathrm{I})}+\bar{\Psi}^{(\mathrm{II})}\)
dominate. Second, if these clustering components do not dominate, the interaction variance
\(\sigma_{\mathrm{III},1\Gamma}^{2}\)
is required to be asymptotically negligible, thereby eliminating the potentially non-Gaussian contribution of
\(\bar{\Psi}^{(\mathrm{III})}\). The assumption also imposes that the asymptotic variance of $\hat{\beta}$ is not identically zero, while still allowing some components in the variance decomposition to be absent. In either case, the normalized score admits a Gaussian limit.

Assumption \ref{as:order of variance}  is imposed along every convergent subsequence because the original sequence of variance components need not itself converge. This subsequence formulation allows us to establish validity without requiring a unique limiting configuration of the variance decomposition. Assumption \ref{as:order of variance} can be viewed as a relaxation of standard sufficient conditions in the two-way clustering literature. Case (i) corresponds to the familiar Gaussian-dominance regime in which the row and/or column clustering components dominate, as in equation (16) of \citet{mackinnon2021wild} and Assumption 3(iv)(1) of \citet{chiang2023standard}. Case (ii) weakens the independence-type requirement imposed in equation (17) of \citet{mackinnon2021wild} and Assumption 3(iv)(2) of \citet{chiang2023standard}. Unlike those conditions, it does not require the interaction component to be i.i.d. or absent; it only requires its variance contribution, \(\sigma_{\mathrm{III},1\Gamma}^{2}\), to be asymptotically negligible. Nor does it require the row and column clustering components to be absent. These components may still be present, as long as they are not strong enough to generate the dominating Gaussian regime in case (i), and the overall asymptotic variance remains bounded away from zero. Therefore, the assumption accommodates variance configurations that are excluded by these standard conditions.

\begin{theorem}\label{thm:1} Let $\mathcal{B}_{0}$ denote the collection
of DGPs $\Gamma$ that satisfy Assumptions \ref{ass:ahk}--\ref{as:order of variance}.
Then 
\[
\Sigma_{GH}^{-1/2}\bigl(\hat{\beta}-\beta_{0}(\tau)\bigr)\overset{d}{\to}\mathcal{N}\!\left(0,\mathbf{I}_{d}\right)
\]
uniformly over $\Gamma\in\mathcal{B}_{0}$, as $G,H\to\infty$. \end{theorem}

Theorem \ref{thm:1} establishes asymptotic normality under self-normalization.
This normalization accommodates the possibility that the convergence
rate of $\hat{\beta}$ varies with the clustering structure.  In particular, Theorem \ref{thm:1} and equation \eqref{eq: variance of score} together imply that the (infeasible) convergence rate of $\widehat{\beta}(\tau)$ is $r_{GH}^{1/2}$. 
Thus, the convergence rate is determined by the projection component that dominates the variance decomposition in \eqref{eq: variance of score}. In particular, under standard one-way clustering (e.g., along the first
dimension), where $\sigma_{\mathrm{I},1\Gamma}^{2}$ is fixed and positive
definite and $\sigma_{\mathrm{II},1\Gamma}^{2}=0$, the rate reduces to $G^{1/2}$.
Under i.i.d.\ sampling, where $\sigma_{\mathrm{I},1\Gamma}^{2}
=\sigma_{\mathrm{II},1\Gamma}^{2}=0$, it reduces to $(GH)^{1/2}$.

\section{Cluster-Robust Variance Estimator (CRVE)}

\label{sec:variance} The two-way cluster-robust variance estimator for quantile regression takes the usual sandwich form 
\[
\widehat{\Sigma}=\widehat{D}^{-1}\widehat{\Omega}\,\widehat{D}^{-1},
\]
where $\widehat{D}$ is a consistent estimator of $D(\tau)$, and
$\widehat{\Omega}$ is consistent for the deterministic target ${\Omega}_{GH}$.

\subsection{Estimating $D(\tau)$.}

The matrix $D(\tau)=E\!\left[f_{e\vert X}(0\vert X_{gh})X_{gh}X_{gh}^{\top}\right]$
captures the impact of conditional heteroskedasticity through the
conditional density at the target quantile. We estimate $D(\tau)$
using Powell's (nonparametric) kernel estimator, 
\[
\widehat{D}=\frac{1}{GH\,\ell}\sum_{g=1}^{G}\sum_{h=1}^{H}K\!\left(\frac{y_{gh}-X_{gh}^{\top}\widehat{\beta}}{\ell}\right)X_{gh}X_{gh}^{\top},
\]
where $\ell>0$ is a bandwidth and $K(u)=\tfrac{1}{2}\,\mathbf{1}\{|u|\le1\}$
is the uniform kernel. Notably, the form of $\widehat{D}$ is identical
to that used under i.i.d. sampling; the difference lies entirely in
the dependence structure that governs its asymptotic behavior.

\citet{kato2012asymptotic} establishes consistency of Powell's estimator
under weak dependence. Extending this result to two-way clustered
arrays is non-trivial for three reasons. First, the convergence rate
of $\widehat{\beta}(\tau)$, denoted $r_{GH}$, can vary across dependence
regimes, and this rate enters $\widehat{D}$ in an essential way.
Second, $\widehat{D}$ itself may converge at a different regime-dependent
rate, say $r_{GH,D}$, and its leading asymptotic component may change
with the regime. The rates $r_{GH}$ and $r_{GH,D}$ need not coincide.
If $r_{GH,D}$ is relatively small, the nominal leading term
in the expansion of $\widehat{D}$ may be dominated by
remainder terms driven by the estimation error of $\widehat{\beta}$ (with a rate of $r_{GH}$). Third, dependence
arises along both cluster dimensions, so the analysis must disentangle
the row- and column-cluster components.

Let $Q_{gh}:=vech\!\left(X_{gh}X_{gh}^{\top}\right)\in\mathbb{R}^{d(d+1)/{2}},$
and denote the conditional density of $e_{gh}=e$ given subvectors of $(X_{gh}^\top,U_g,V_h)$ by $f_{e\vert X,U}(e\vert X_{gh},U_g)$, $f_{e\vert X,V}(e\vert X_{gh},V_h)$, and $f_{e\vert X,U,V}(e\vert X_{gh},U_g,V_h)$. We now impose the density, stronger moment, and bandwidth conditions
that ensure consistency of $\widehat{D}$. Let $R:=\min\{G,H\}$.

\begin{assumption}[Density and bandwidth]\label{as:bandwidth}
(i) There exist $\varepsilon_{0}>0$ 
such that, uniformly over $(x,u,v)$ and all $|e|\le\varepsilon_{0}$,
$ f_{e\vert X,U,V}(e\vert x,u,v)\le C<\infty.$ (ii) $E\!\left(\|X_{gh}\|^{4}\vert U_{g},V_{h}\right)<\infty$
uniformly over $(U_{g},V_{h})$. (iii) $\sup_{e,x}\bigl|f_{e\vert X}^{(2)}(e\vert x)\bigr|<\infty$
(iv) As $R\to\infty$, $\ell\to0$ and $R^{1/2}\ell/\log R\to\infty$.
(v) $E\!\big[Q_{gh}Q_{gh}^{\top}f_{e\vert X}(0\vert X_{gh})\big]$ is positive definite. \end{assumption}
Assumptions \ref{as:bandwidth}(i)-(ii) require uniform boundedness
of the conditional density around $e=0$ and the conditional fourth
moments of the regressors. Assumption \ref{as:bandwidth}(iii) imposes
bounded second derivative to ensure the dominated convergence. Assumption
\ref{as:bandwidth}(iv) is a standard bandwidth restriction; it is
the two-way clustered analogue of Assumption~3 in \citet{kato2012asymptotic}.
Finally, Assumption \ref{as:bandwidth}(v) imposes a nonsingularity
condition to ensure that $E\!\big[Q_{gh}Q_{gh}^{\top}f_{e\vert X}(0\vert X_{gh})\big]$
is positive definite, and hence the limiting variance is not identically
zero in the worst case.

To establish the asymptotic normality of $\widehat{D}$, we first derive a Hoeffding-type decomposition, which naturally leads to the following definitions:
\begin{align*}
\sigma_{\mathrm{I},Q}^{2} 
&:= Var\!\Big( E\!\big[\,Q_{gh}f_{e\vert X,U}(0\vert X_{gh},U_g)\,\big|\,U_g\big] \Big),\\
\sigma_{\mathrm{II},Q}^{2} 
&:= Var\!\Big( E\!\big[\,Q_{gh}f_{e\vert X,V}(0\vert X_{gh},V_h)\,\big|\,V_h\big] \Big).
\end{align*}
For $j\in\{\mathrm{I},\mathrm{II}\}$, the matrix $\sigma_{j,Q}^{2}$ may depend on $G$ and $H$. 
As before, we use $\sigma_{j,1Q}^{2}$ to denote the order of its first diagonal element.

\begin{theorem}\label{thm:Jacobian}
Let $\mathcal{B}_{1}$ denote the collection of DGPs $\Gamma$ satisfying Assumptions \ref{ass:ahk}--\ref{as:bandwidth}. Then the following statements hold uniformly over $\Gamma\in\mathcal{B}_{1}$.

\begin{enumerate}[label=(\arabic*)]
\item \textbf{Consistency and rate.}
\begin{align}
\widehat{D}-D(\tau)=o_{P}\!\left(r_{GH}^{-1/2}\ell^{-1/2}\right)+O_P\left(r_{GH}^{-1/2}+R^{-1/2}\ell^{-1/2}+\ell^{2}\right)=o_{P}(1).
\label{eq: consistency D}
\end{align}

\item \textbf{Asymptotic normality.} Suppose, in addition, that 
(i) 
$\lambda_{\max}(\sigma_{j,Q}^{2})/\lambda_{\min}(\sigma_{j,Q}^{2})=O(1),$ for each $j\in\{\mathrm{I},\mathrm{II}\}$;
(ii) At least one of the following two conditions holds: 
\begin{align}\label{eq: variance cp}
\sigma_{\mathrm{I},1\Gamma}^{2}/(\ell\sigma_{\mathrm{I},1Q}^{2})=O(1)
\quad \text{and} \quad
\sigma_{\mathrm{II},1\Gamma}^{2}/(\ell\sigma_{\mathrm{II},1Q}^{2})=O(1),\end{align}or \begin{align}\label{eq: variance cp1}H\sigma_{\mathrm{I},1\Gamma}^{2}+G\sigma_{\mathrm{II},1\Gamma}^{2}=O(1).\end{align}
Then, as $G,H\to\infty$,
\begin{align*}
V_D^{-1/2}\Bigg(
vech(\widehat{D})-vech\!\big(D(\tau)\big)
-\frac{\ell^{2}}{6}E\!\big[f_{e\vert X}^{(2)}(0\vert X_{gh})\,Q_{gh}\big]
+o(\ell^{2})
\Bigg)
\overset{d}{\to}
\mathcal{N}\!\left(
\bm{0}_{\frac{d(d+1)}{2}\times 1},
\,\mathbf{I}_{\frac{d(d+1)}{2}}
\right),
\end{align*}
where
\begin{align}\label{eq: V_D}
V_D
=
\frac{\sigma_{\mathrm{I},Q}^{2}}{G}
+
\frac{\sigma_{\mathrm{II},Q}^{2}}{H}
+
\frac{1}{2GH\ell}\,
E\!\big[Q_{gh}Q_{gh}^{\top}f_{e\vert X}(0\vert X_{gh})\big].
\end{align}
\end{enumerate}
\end{theorem}
Theorem \ref{thm:Jacobian}(1) shows that $\widehat{D}$ is a consistent estimator of $D(\tau)$. 
Theorem \ref{thm:Jacobian}(2) further establishes an asymptotic linear expansion and a central limit theorem for $vech(\widehat{D})$, normalized by $V_D^{-1/2}$. 
To clarify the stochastic order, define the following infeasible convergence rate for $\widehat{D}$:
\[
r_{GH,D}:=\min\left\{ \frac{G}{\sigma_{\mathrm{I},1Q}^{2}},\frac{H}{\sigma_{\mathrm{II},1Q}^{2}},GH\ell\right\}.
\]
The rate $r_{GH,D}$ is reminiscent of $r_{GH}$, since $\widehat{D}$ also admits a three-way decomposition into row, column, and interaction components. 
However, the two rates may behave quite differently. 
The reason is that there is no direct relationship between $\sigma_{\mathrm{I},1\Gamma}^{2}$ and $\sigma_{\mathrm{I},1Q}^{2}$, nor between $\sigma_{\mathrm{II},1\Gamma}^{2}$ and $\sigma_{\mathrm{II},1Q}^{2}$. 
Moreover, the variance components enter the two rates differently. 
For example, when the observations are i.i.d. across intersections $(g,h)$, one may have $r_{GH}=GH$ while $r_{GH,D}=GH\ell$, so the two rates differ by a factor of $\ell$. 
By contrast, the first two components of both $r_{GH}$ and $r_{GH,D}$ do not involve $\ell$.

For asymptotic normality in Theorem \ref{thm:Jacobian}(2), we impose two additional conditions. 
First, condition (2)(i) requires that the eigenvalues of $\sigma_{j,Q}^{2}$ are of the same order, as in the earlier assumptions. 
Second, condition (2)(ii) ensures that the leading stochastic term is not dominated by remainder terms arising from the estimation error in $\widehat{\beta}$. 
Intuitively, \eqref{eq: consistency D} suggests that, for the asymptotic normality result to hold, we need
\begin{align}\label{eq: rate cp}
r_{GH}^{-1}\ell^{-1}r_{GH,D}=O(1).
\end{align}
We consider two cases. 
When clustering dependence is sufficiently strong along at least one dimension, the additional condition \eqref{eq: variance cp} guarantees \eqref{eq: rate cp}. 
By contrast, when condition \eqref{eq: variance cp1} holds, clustering dependence is weak or absent along both dimensions (for example, under independence across intersections), and in this case \eqref{eq: rate cp} holds automatically. Note that these additional conditions are not needed for consistency of
\(\widehat{D}\) and do not enter the first-order analysis of
\(\widehat{\beta}\). They are used only to describe the leading
bias-variance tradeoff of the kernel estimator of the nuisance Jacobian.
\begin{remark}
In contrast to the score-based limit theory for $\widehat{\beta}$, no additional condition on the potentially non-Gaussian interaction component is required for the kernel-based estimator. This is because, in $\widehat{D}$, the interaction component is of smaller order than the leading terms and hence does not contribute to the first-order asymptotic distribution. Therefore, the limit theory is driven entirely by the dominant Gaussian components, yielding a Gaussian limit automatically.
\end{remark}

Theorem \ref{thm:Jacobian} also provides an AMSE-guided bandwidth rule
for estimating the nuisance Jacobian \(D(\tau)\). Specifically, the
leading approximation to the MSE of \(vech(\widehat D)\) is
\[
\text{AMSE}_{D}\left(\ell\right)
=
\frac{\ell^{4}}{36}
\left\Vert
E\!\big[f_{e\vert X}^{(2)}(0\vert X_{gh})\,Q_{gh}\big]
\right\Vert ^{2}
+
tr\left\{ Var\left(vech(\widehat{D})\right)\right\}.
\]
The bandwidth that minimizes this AMSE criterion is
\[
\ell_{\text{AMSE}}
=
\left(GH\right)^{-1/5}
\left(
\frac{
4.5\cdot tr\left(E\!\big[Q_{gh}Q_{gh}^{\top}
f_{e\vert X}(0\vert X_{gh})\big]\right)
}{
E\!\big[f_{e\vert X}^{(2)}(0\vert X_{gh})\,Q_{gh}\big]^{\top}
E\!\big[f_{e\vert X}^{(2)}(0\vert X_{gh})\,Q_{gh}\big]
}
\right)^{1/5}.
\]
This criterion is aimed at estimating the nuisance matrix \(D(\tau)\).\footnote{We do not claim that it is optimal for coverage or for the final
studentized quantile regression statistic. An inference-optimal bandwidth
would require a higher-order expansion of the full quantile regression
estimator and of the resulting variance estimator, including the effect
of estimating \(D(\tau)\) inside the sandwich formula. Such an analysis is
beyond the scope of the present paper.} We  follow the standard
plug-in approach used in kernel estimation of quantile-regression
Jacobian matrices and use the Gaussian location rule
\[
\widehat{\ell}_{\text{AMSE}}
=
\widehat{\sigma}\left(GH\right)^{-1/5}
\left(
\frac{
4.5\cdot\frac{1}{GH}\sum_{g,h}\left\Vert Q_{gh}\right\Vert ^{2}
}{
\alpha\left(\tau\right)
\left\Vert \frac{1}{GH}\sum_{g,h}Q_{gh}\right\Vert ^{2}
}
\right)^{1/5},
\]
with
\(\widehat{\sigma}=\text{MAD}\left(\left\{ \widehat{e}_{gh}\right\} \right)/0.6745\)
and
\(\alpha\left(\tau\right)=\left(1-\Phi^{-1}\left(\tau\right)\right)^{2}
\phi\left(\Phi^{-1}\left(\tau\right)\right)\).
Here, \(\text{MAD}(\cdot)\) is the median absolute deviation, and
\(\Phi\) and \(\phi\) are the distribution function and density function
of the standard normal distribution. The simulation results below show
that inference based on this plug-in rule is stable across the dependence
regimes considered.\footnote{In unreported sensitivity checks, we also considered several fixed
multiples of the plug-in bandwidth, such as smaller and larger fractions
of \(\widehat\ell_{\mathrm{AMSE}}\). The proposed bandwidth delivered
the best, or nearly the best, finite-sample inference performance across
the designs considered.}

\subsection{Consistency of Quantile Regression CRVE}
In contrast to $\widehat{D}$, the construction of $\widehat{\Omega}$
must account explicitly for two-way clustering and therefore differs
from the i.i.d.\ case. Recall the (estimated) quantile score 
\[
\widehat{\Psi}_{gh}=X_{gh}\Bigl(\tau-\mathbf{1}\{y_{gh}\le X_{gh}^{\top}\widehat{\beta}\}\Bigr).
\]
We estimate $\Omega_{GH}(\tau)$ by aggregating row-, column-, and idiosyncratic
components: 
\[
\widehat{\Omega}:=\widehat{\Omega}_{\mathrm{I}}+\widehat{\Omega}_{\mathrm{II}}+\widehat{\Omega}_{\mathrm{III,IV}},
\]
where 
\begin{align*}
\widehat{\Omega}_{\mathrm{I}} & :=\frac{1}{G^{2}H^{2}}\sum_{g=1}^{G}\sum_{h=1}^{H}\sum_{\substack{h'=1\\
h'\neq h
}
}^{H}\widehat{\Psi}_{gh}\widehat{\Psi}_{gh'}^{\top},\\
\widehat{\Omega}_{\mathrm{II}} & :=\frac{1}{G^{2}H^{2}}\sum_{h=1}^{H}\sum_{g=1}^{G}\sum_{\substack{g'=1\\
g'\neq g
}
}^{G}\widehat{\Psi}_{gh}\widehat{\Psi}_{g'h}^{\top},\\
\widehat{\Omega}_{\mathrm{III,IV}} & :=\frac{1}{G^{2}H^{2}}\sum_{g=1}^{G}\sum_{h=1}^{H}\widehat{\Psi}_{gh}\widehat{\Psi}_{gh}^{\top}.
\end{align*}
This estimator is the quantile-regression analogue of the two-way
CRVE for simple OLS estimator proposed by \citet{cameron2011robust}. In practice, we follow \citet{cameron2011robust} and apply the operator
$\text{EVC}(\cdot)$ to $\widehat{\Omega}$, which denotes the eigenvalue correction (e.g., projection
onto the cone of positive semidefinite matrices) applied to ensure
a positive semidefinite estimate.\footnote{An inference based on this adjustment is generally not invariant under affine transformations.}

Let $f(e_{gh},e_{gh'}\vert X_{gh},X_{gh'},U_{g},V_{h},V_{h'})$ and $f(e_{gh},e_{g'h}\vert X_{gh},X_{g'h},U_{g},U_{g'},V_{h})$
denote the conditional joint densities of $(e_{gh},e_{gh'})$ and
$(e_{gh},e_{g'h})$, respectively. Define $f(e_{gh}\vert X_{gh},X_{gh'})$ and $f(e_{gh}\vert X_{gh},X_{g'h})$ as the conditional marginal densities of $e_{gh}$, respectively. For integers $l,m\ge0$, define
the mixed partial derivatives 
\begin{align*}
f^{(l,m)}(e_{gh},e_{gh'}\vert X_{gh},X_{gh'}) & :=\frac{\partial^{\,l+m}}{\partial e_{gh}^{\,l}\,\partial e_{gh'}^{\,m}}\,f(e_{gh},e_{gh'}\vert X_{gh},X_{gh'}),\\
f^{(l,m)}(e_{gh},e_{g'h}\vert X_{gh},X_{g'h}) & :=\frac{\partial^{\,l+m}}{\partial e_{gh}^{\,l}\,\partial e_{g'h}^{\,m}}\,f(e_{gh},e_{g'h}\vert X_{gh},X_{g'h}).
\end{align*}

We impose the following conditions for validity of $\widehat{\Omega}$.

\begin{assumption}[Strong moments and smoothness]\label{as:moment-strong}
There exist a constant $C_{1}>0$ and integrable envelope
functions $D_{1}(\cdot)$ and $D_{2}(\cdot)$ such that: 

\begin{enumerate}[label=(\roman*)]
\item \label{as:ms-i}  $\max_{g\le G}\max_{h\le H}\|X_{gh}\|\le C_{1}R^{1/8}$ and $\sup_{U_g,V_h}E(\|X_{gh}\|^6\vert U_g,V_h)<\infty$.
\item \label{as:ms-iii} The conditional marginal densities are uniformly bounded:
\(
\sup_{e_{gh},X_{gh},X_{g'h}}\bigl|f(e_{gh}\vert X_{gh},X_{g'h})\bigr|<\infty\)
 and \(
\sup_{e_{gh},X_{gh},X_{gh'}}\bigl|f(e_{gh}\vert X_{gh},X_{gh'})\bigr|<\infty
\).  The conditional joint densities are uniformly bounded:
\[
\sup_{e_{1},e_{2},x_{1},x_{2},U_{g},V_{h},U_{g'},V_{h'}}\bigl|f(e_{1},e_{2}\vert x_{1},x_{2},U_{g},V_{h},U_{g'},V_{h'})\bigr|<\infty,
\]
where $(e_{1},e_{2},x_{1},x_{2})$ denotes either $(e_{gh},e_{gh'},X_{gh},X_{gh'})$
or $(e_{gh},e_{g'h},X_{gh},X_{g'h})$. 
\item \label{as:ms-iv} For $l,m\in\{1,2\}$, 
\[
\sup_{e_{2},x_{1},x_{2}}\bigl|f^{(l,0)}(e_{1},e_{2}\vert x_{1},x_{2})\bigr|\le D_{1}(e_{1}),\qquad\sup_{e_{2},x_{1},x_{2}}\bigl|f^{(0,m)}(e_{1},e_{2}\vert x_{1},x_{2})\bigr|\le D_{2}(e_{1}),
\]
for both pairs $(e_{1},e_{2},x_{1},x_{2})=(e_{gh},e_{gh'},X_{gh},X_{gh'})$
and $(e_{gh},e_{g'h},X_{gh},X_{g'h})$. 
\end{enumerate}
\end{assumption} Assumption \ref{as:moment-strong}(i) imposes standard boundedness
conditions on the regressors. Assumptions \ref{as:moment-strong}(ii)--(iii)
impose smoothness and boundedness conditions on the relevant conditional
joint densities around the target quantile. These conditions are standard
in quantile regression because uniform expansions of the nonsmooth
indicator function require local density regularity. Related smoothness conditions are also
imposed in recent work on HAC estimation for quantile regression, such as
\citet{galvao2024hac}. These restrictions should be viewed as maintained regularity conditions,
rather than conditions that can be verified exactly in finite samples. In
applications, their plausibility can be assessed by inspecting the
estimated residual density near zero and checking for influential
observations or leverage points. 

\begin{theorem}\label{thm:2-1} Let $\mathcal{B}_{2}$ denote the
collection of DGPs $\Gamma$ that satisfy Assumptions \ref{ass:ahk}-\ref{as:moment-strong}.
Then, 
\[\widehat{\Sigma}^{-1/2}\bigl(\hat{\beta}-\beta_{0}(\tau)\bigr)\ \overset{d}{\to}\ \mathcal{N}\!\left(0,\mathbf{I}_{d}\right),
\]
uniformly over $\Gamma\in\mathcal{B}_{2}$, as $G,H\to\infty$.
\end{theorem} Theorem \ref{thm:2-1} establishes the uniform validity
of the proposed two-way CRVE. Consequently, standard large-sample
inference procedures can be implemented using the quantile regression
estimator $\widehat{\beta}$ together with the variance estimator
$\widehat{\Sigma}$.

Note that if Assumption \ref{as:order of variance}(ii) fails, the limiting distribution may be non-Gaussian, with a convergence rate of $\sqrt{GH}$. This case is substantially more delicate and has only recently begun to be analyzed in a systematic way; see, for example, \citet{menzel2021bootstrap}, \citet{hounyo2025projection}, and \citet{davezies2025analytic}. In particular, \citet{menzel2021bootstrap} (cf. Proposition~4.1) provides a sharp and highly influential characterization of the asymptotic distribution for sample means. Building on this insight, we show that a closely related impossibility phenomenon extends beyond sample means to uniform inference in two-way clustered quantile regression.

For vectors $a,t\in\mathbb R^d$, the notation $a\le t$ is understood componentwise.
\begin{proposition}[Impossibility of uniform consistency]\label{prop:impossibility_uniform}
For each data-generating process $\Gamma$, let $P_\Gamma$ denote the probability measure induced by $\Gamma$ on the underlying sample space. Let $\mathcal{B}_{3}$ be the class of DGPs $\Gamma$ satisfying Assumptions \ref{ass:ahk}-\ref{as:hyperplane}, \ref{as:bandwidth}, \ref{as:moment-strong}, and $\lim\inf_{G,H\to\infty}\Bigl(H\sigma_{\mathrm{I},\Gamma}^{2}+G\sigma_{\mathrm{II},\Gamma}^{2}+\sigma_{\mathrm{III},\Gamma}^{2}+\sigma_{\mathrm{IV},\Gamma}^{2}\Bigr)>0$.
Let $\mathcal{E}$ be the collection of all measurable maps of the observed sample
\(
\{y_{gh}^{(\Gamma)},X_{gh}^{(\Gamma)}\}_{g\le G,\;h\le H}.
\)
Then there exist $\varepsilon>0$ and $\delta>0$ such that
\[
\liminf_{G,H\to\infty}\inf_{\widehat E\in\mathcal E}\sup_{\Gamma\in\mathcal B_3}
P_\Gamma\!\left(
\sup_{t\in\mathbb R^d}
\left|
P_\Gamma\!\left(\sqrt{GH}\,(\widehat\beta-\beta_0(\tau))\le t\right)
-
\widehat E\!\left(\{y_{gh}^{(\Gamma)},X_{gh}^{(\Gamma)}\}_{g\le G,\;h\le H};t\right)
\right|
>\varepsilon
\right)\ge\delta.
\]
\end{proposition}

Proposition \ref{prop:impossibility_uniform} establishes an impossibility result where no procedure can deliver uniformly consistent inference. Consequently, without Assumption \ref{as:order of variance}, the difficulty is not merely that the limiting distribution may be non-Gaussian; in some cases, a fundamental failure may arise, namely that uniformly consistent inference may no longer be attainable.

\section{Monte Carlo simulation}

\label{sec:mc}

In this simulation section, we assess the robustness of the proposed two-way clustered quantile regression inference procedure across a range of clustering configurations.
We evaluate the finite-sample performance of the proposed
two-way CRVE and compare
it with alternatives that only account for dependence
along the $g$-dimension, the $h$-dimension, or the $(g,h)$ intersection,
respectively. 

For each replication, we generate a two-way array $\{(y_{gh},X_{gh})\}_{g\le G,\,h\le H}$
from 
\begin{align}
y_{gh} & =\beta_{1}+\sum_{j=2}^{d}\beta_{j}X_{gh,j}+e_{gh},\label{eq: dgp}\\
X_{gh,j} & =\omega_{U}^{X}U_{g}^{X,j}+\omega_{V}^{X}V_{h}^{X,j}+\omega_{W}^{X}W_{gh}^{X,j},\\
e_{gh}&
=
\omega_U^e U_g^e+\omega_V^e V_h^e+\omega_W^e W_{gh}^e
-
\sigma_e \Phi^{-1}(\tau),\label{eq: e_gh}\\
\sigma_e&
=
\sqrt{(\omega_U^e)^2+(\omega_V^e)^2+(\omega_W^e)^2},
\end{align}
The latent components are mutually independent and i.i.d.\ standard
normal. Hence both the regressor and the regression error exhibit
additive two-way dependence through $(U_{g},V_{h})$ plus an idiosyncratic
component. In the baseline design, we set $\beta_j(\tau)=1$ for all $j=1,\ldots,d$ and conduct inference on the null hypothesis $\mathcal{H}_0:\beta_d(\tau)=1$ at $\tau=0.50$.

We also consider the case $\tau=0.25$. In addition, we study specifications in which $\beta_d(\tau)$ varies with $\tau\in(0,1)$ and test the corresponding quantile-specific null hypotheses. We further analyze a heteroskedastic design. Across all these alternative specifications, the results are qualitatively similar to those for the benchmark design reported in the main text. We therefore present them in  Internet Appendix IA.

We compute the quantile regression estimator $\widehat{\beta}_{d}(\tau)$
and the associated two-way clustered variance estimator. All results
are based on $10,000$ Monte Carlo replications.  By default, we set $d=10$, $G=H=50$, and 
$\omega_{\bullet}^{X}=\omega_{\bullet}^{e}=1$.  Nominal level is
$5\%$.

We compare the proposed two-way procedure (denoted \textbf{CTW}) with
four alternatives, described in detail in  Internet Appendix IA:
\begin{itemize}
\item \textbf{CG (cluster-$g$ only).} A one-way clustered inference method
that treats $g$ as the only clustering dimension and ignores dependence
across $h$. 
\item \textbf{CH (cluster-$h$ only).} A one-way clustered inference method
that treats $h$ as the only clustering dimension and ignores dependence
across $g$. 
\item \textbf{CI (intersection-only).} An i.i.d.-style inference method
that effectively uses only the $(g,h)$ intersection component and
ignores both two-way additive components. 
\item \textbf{CTW$_{\text{II}}$ (two-way cluster without intersection correction).
}A two-way clustered inference procedure that enforces positive semidefiniteness
without using EVC, but does not correct for the ``double-counting''
of the intersection component.
\end{itemize}
The one-way clustered quantile bootstrap of \citet{hagemann2017cluster}
exhibits qualitatively similar behavior to CG and CH in our simulations.

In this DGP, both $X_{gh,j}$ and $e_{gh}$ contain additive $g$-
and $h$-level components. Consequently, the score contributions relevant
for inference inherit dependence in \emph{both} dimensions. The proposed
estimator targets this structure by combining the $g$-level, $h$-level,
and $(g,h)$ components. In contrast, CG, CH, and CI omit at least
one of these components. Under the present scaling, the omitted component
does not vanish as $G,H$ increase and may become relatively more important
as the array grows, which leads to progressively more distorted standard
errors and hence worsening size (typically over-rejection) as $G,H$
increases. Rejection is based on the usual two-sided $t$-test.

\afterpage{ \begin{landscape} 
\begin{figure}[h!]
\centering \begin{subfigure}[t]{0.49\hsize} \subcaption{Two-way Clustering}\includegraphics[width=1\textwidth]{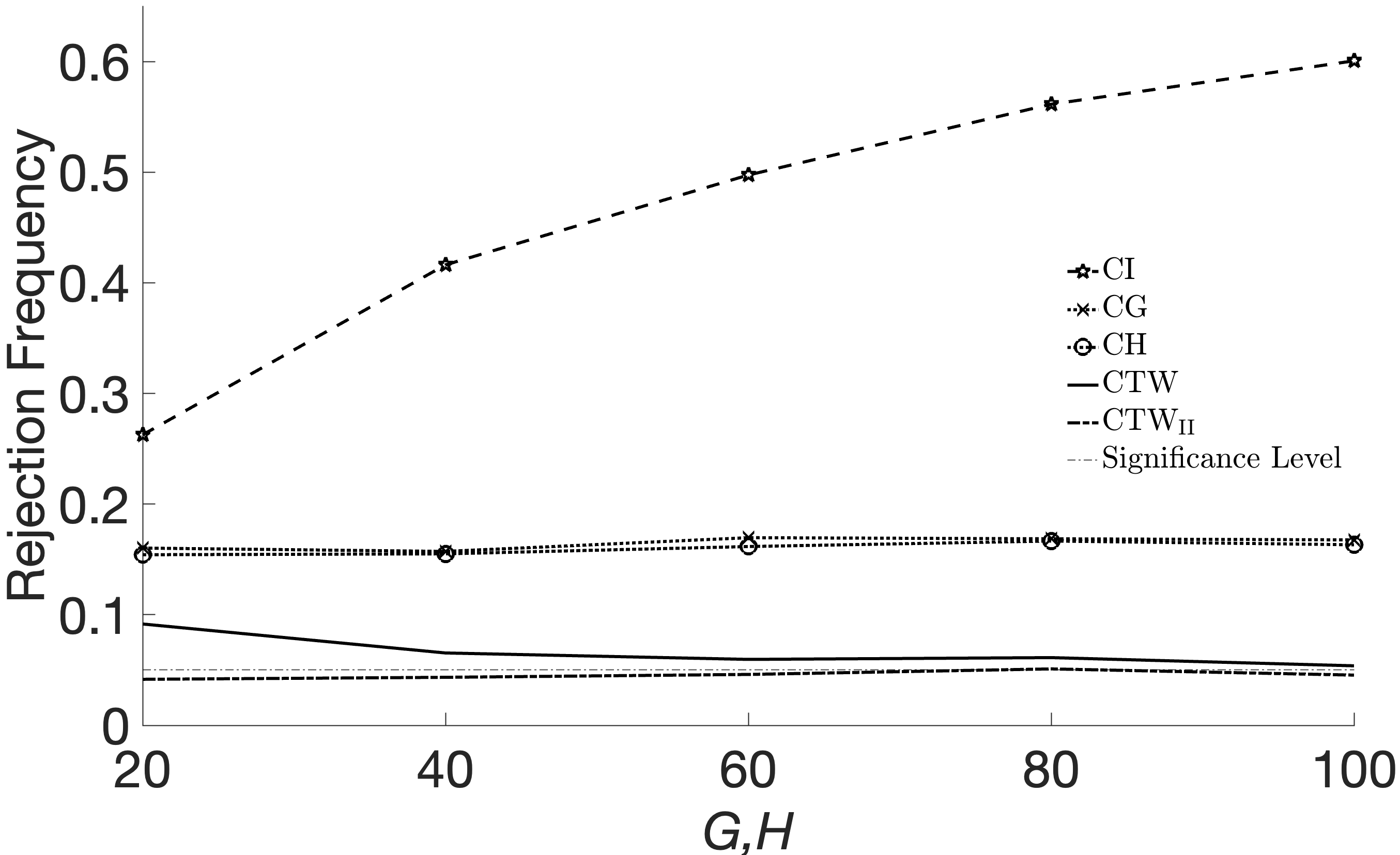}
\end{subfigure} \begin{subfigure}[t]{0.49\hsize} \subcaption{One-way Clustering along the first ($G$) dimension, $\omega_V^X=\omega_V^e=0$} \includegraphics[width=1\textwidth]{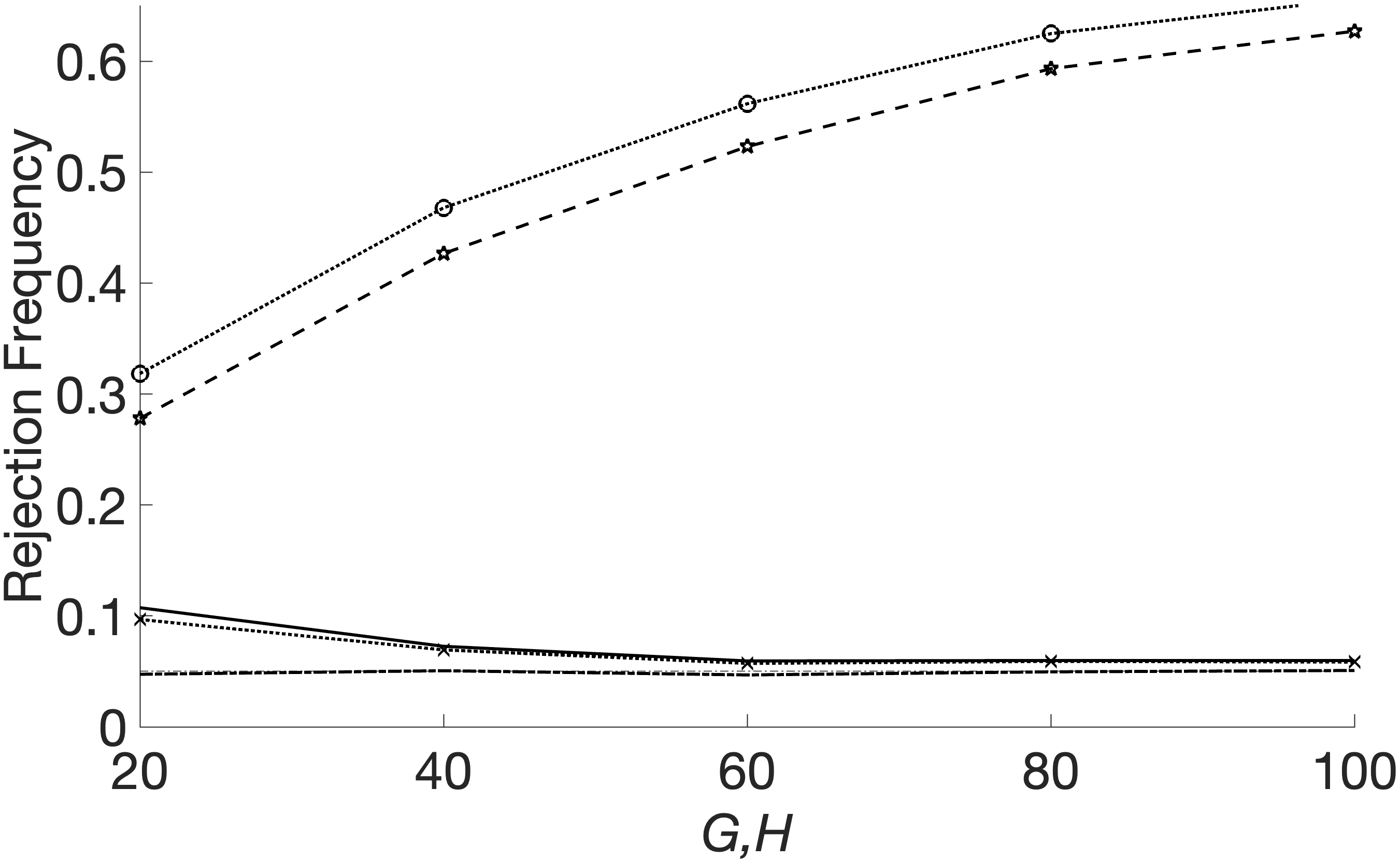}
\end{subfigure}

\begin{subfigure}[t]{0.49\hsize} \subcaption{ Independence, $\omega_U^X=\omega_U^e=\omega_V^X=\omega_V^e=0$}\includegraphics[width=1\textwidth]{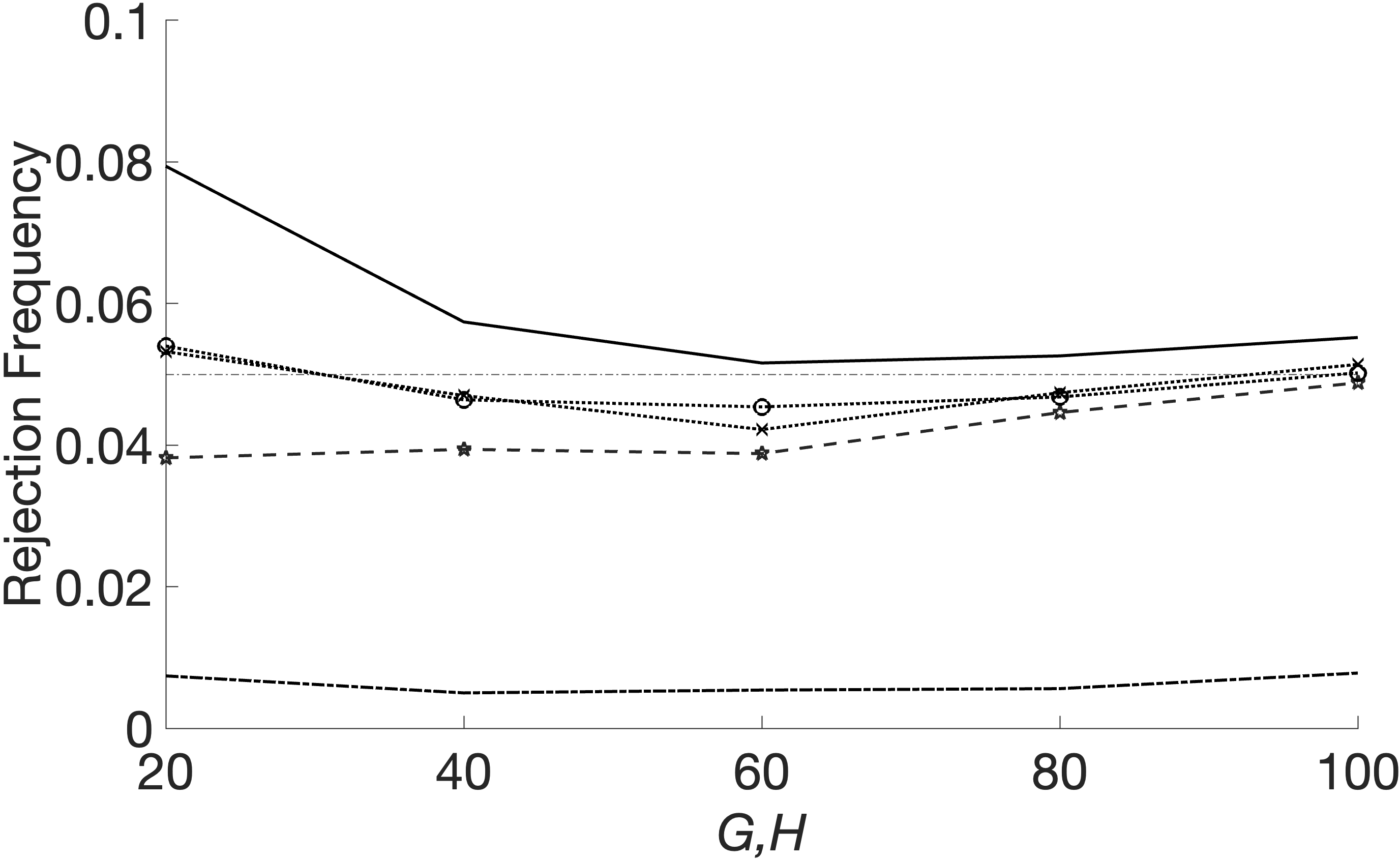}
\end{subfigure} \begin{subfigure}[t]{0.49\hsize} \subcaption{Varying Clustering Dependence along the second ($H$) dimension $\omega_V^X$ and $\omega_V^e$} \includegraphics[width=1\textwidth]{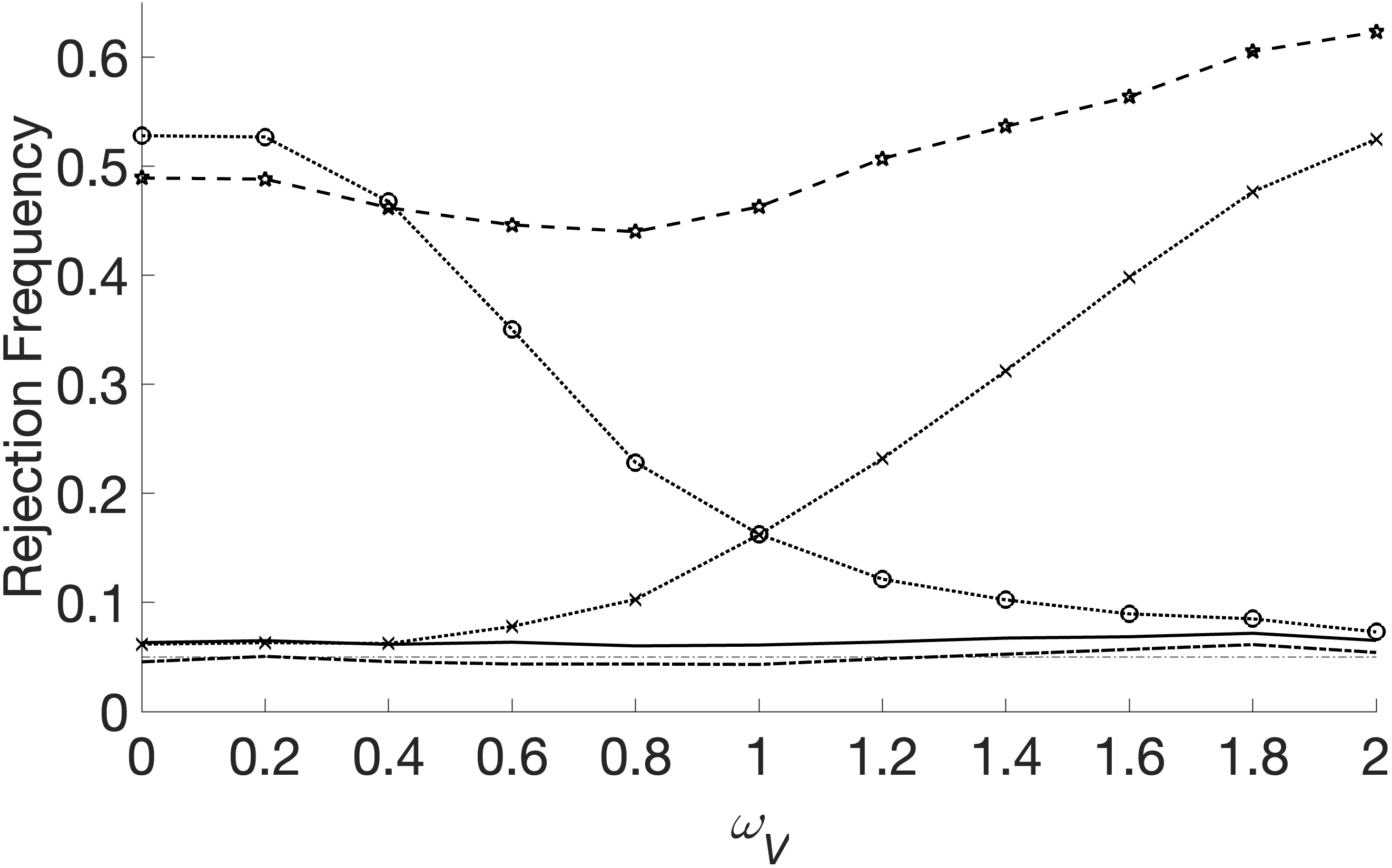}
\end{subfigure}

\caption{\textbf{Rejection frequency under varying levels of clustering dependence}.
The default setting is $\tau=0.50$, $d=10$, $G=H=50$, and $\omega_{\bullet}^{X}=\omega_{\bullet}^{e}=1$. Results are based on 10,000 Monte Carlo replicates.
The predetermined significance level is 5\%.}
\label{fig: rej frequency 1} 
\end{figure}

\end{landscape} }

Figure~\ref{fig: rej frequency 1} reports rejection frequencies under varying clustering structures. 
In Panel~(a), the data exhibit two-way clustering. The two-way CRVEs, CTW and CTW$_{\mathrm{II}}$, deliver stable and accurate size control as $G$ and $H$ increase, whereas the one-way CRVEs, CG and CH, substantially overreject, with rejection frequencies around $0.15$. Ignoring clustering altogether leads to the worst performance: CI overrejects increasingly as $G$ and $H$ grow. Between the two two-way procedures, CTW$_{\mathrm{II}}$ yields slightly lower rejection frequencies because it does not correct for the double-counting term, which inflates the estimated variance and therefore makes rejection harder.

Panel~(b) considers one-way clustering along the first ($G$) dimension only. In this case, CG, CTW, and CTW$_{\mathrm{II}}$ perform well, as each accounts for dependence in the $G$ dimension.

Panel~(c) considers the cluster-independent design. For readability, we rescale the vertical axis because all methods yield rejection frequencies below $0.10$. Here, all procedures except CTW$_{\mathrm{II}}$ provide satisfactory size control. This indicates that, while CTW$_{\mathrm{II}}$ works well under dependence, the resulting variance inflation renders it invalid (overly conservative) when clustering is absent.

Panel~(d) varies the strength of clustering dependence in the second dimension. When dependence in the second dimension is weak (small $\omega^X_V,\omega_V^e$), accounting for dependence in the first dimension is more important, and CG performs well. As dependence in the second dimension strengthens (large $\omega^X_V,\omega_V^e$), CH becomes more appropriate. In both settings, CI fails, whereas both CTW and CTW$_{\mathrm{II}}$ remain reliable across the full range of dependence strengths.

Figure~\ref{fig: rej frequency 2}, Panel~(a), further reports results for an unbalanced design in which we fix $G=50$ and vary $H$ from $20$ to $100$. We find that CH performs slightly better than CG when $H$ is small, whereas CG performs better when $H$ is large. The intuition is that when $H$ is small, each $h$-cluster contains a larger number of observations (i.e., a larger cluster size along the second dimension), so a substantial portion of the dependence is concentrated within the $H$ dimension and must be controlled; consequently, CH is more appropriate. As $H$ increases, clusters along the second dimension become smaller and less dominant, making it relatively more important to account for dependence along the first dimension, so CG improves. Panel~(b) varies the number of regressors, $d$. The qualitative patterns remain essentially unchanged, indicating that the results are not sensitive to the dimension of the covariate vector.

These patterns highlight that accounting for \emph{both} clustering
dimensions is essential in two-way array settings. Procedures that ignore any one dimension systematically
under-estimate sampling variability and over-reject. CI performs worst because it effectively treats observations as
independent across $(g,h)$ and therefore misses the dominant row/column
correlation. The one-way cluster methods (CG and CH) partially correct the problem
by capturing dependence in a single direction, which explains why
they perform better than CI, but they remain
misspecified because the neglected dimension contributes non-negligibly
to the score covariance. By construction, CTW targets the full two-way
covariance structure, which yields stable size and
a clear improvement toward the nominal level as $G$ and $H$ increase. CTW$_\mathrm{II}$ is robust to two-way clustering dependence as well, but is overly conservative when clustering is absent.

Overall, the evidence points to CTW as the preferred procedure because of its robustness across a wide range of settings. The additional results reported in the Internet Appendix IA, including those for $\tau=0.25$, designs in which $\beta_d(\tau)$ varies with $\tau$, and heteroskedastic specifications, display a similar pattern and further support the use of CTW in practice.

\begin{figure}[t!]
\centering \begin{subfigure}[t]{0.49\hsize} \subcaption{$G=50$, Varying $H$}\includegraphics[width=1\textwidth]{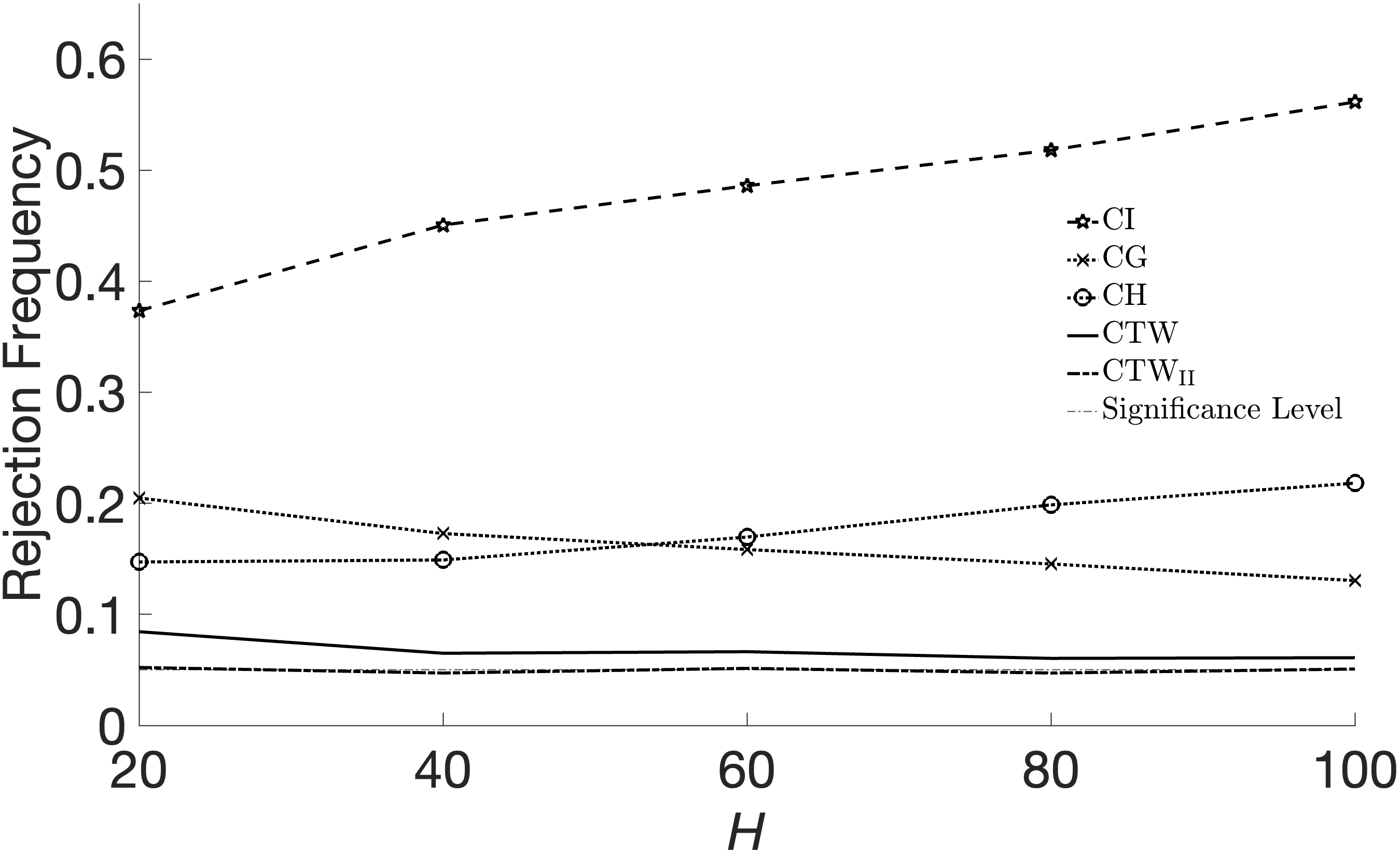}
\end{subfigure} \begin{subfigure}[t]{0.49\hsize} \subcaption{Varying $d$} \includegraphics[width=1\textwidth]{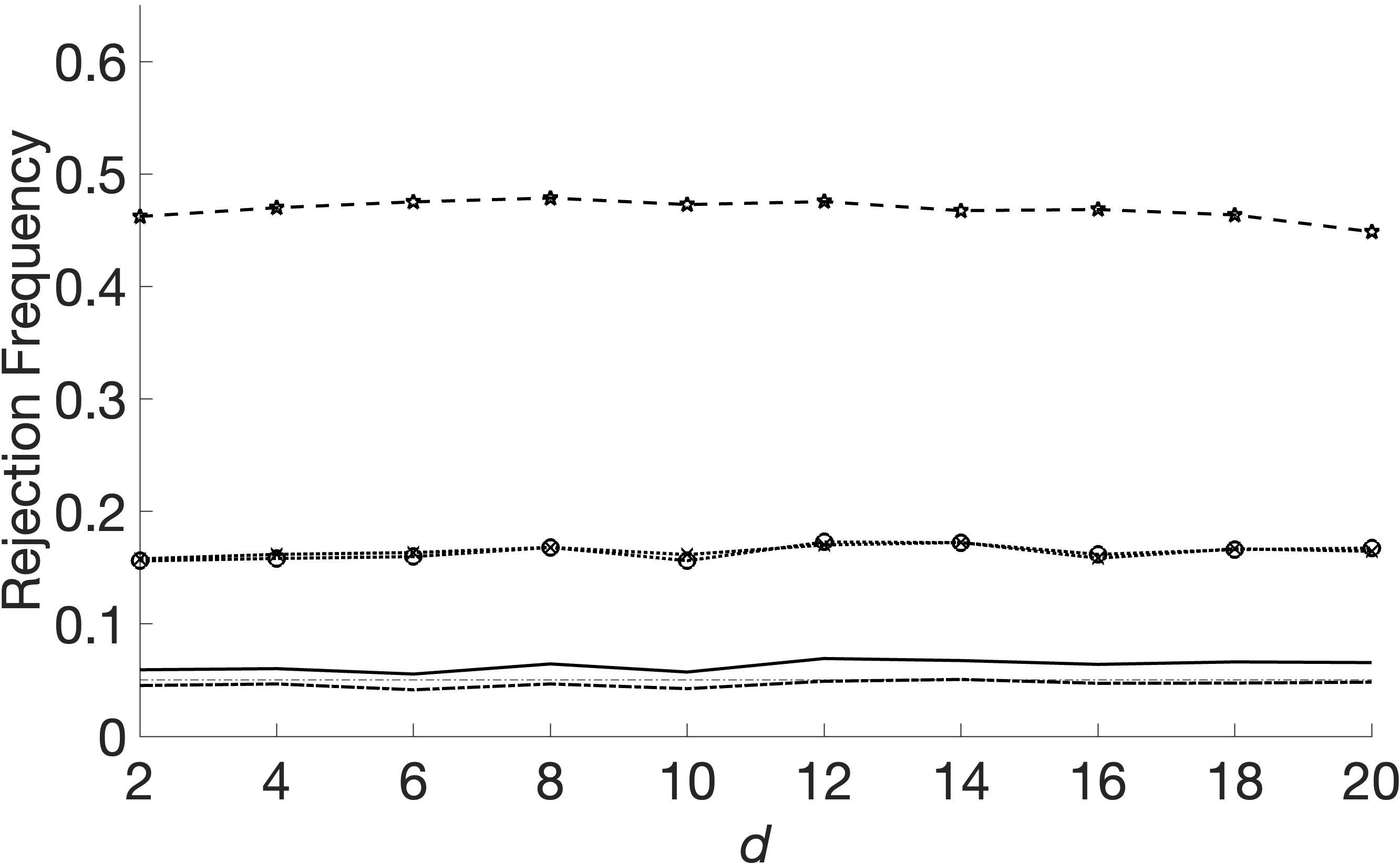}
\end{subfigure}

\caption{\textbf{Rejection frequency under different structures}.
The default setting is $\tau=0.50$, $d=10$, $G=H=50$, and $\omega_{\bullet}^{X}=\omega_{\bullet}^{e}=1$. Results are based on 10,000 Monte Carlo replicates.
The predetermined significance level is 5\%.}
\label{fig: rej frequency 2} 
\end{figure}

\section{Empirical Studies}
\label{sec:empirical}
This section uses a QR framework to study how
teacher-licensing restrictions affect teacher quality. Policy views on
licensing are mixed. Some states have increased licensing stringency,
motivated by the idea that tighter requirements can screen out
lower-ability candidates and raise the left tail of the quality
distribution (e.g., \citealt{kraft2020teacher}). Other states decreased licensing stringency, a policy choice that speaks
directly to our focus on the right tail.
 One argument for reducing stringency is that it may
attract more competitive candidates who would otherwise choose other
professions (e.g., \citealt{hanushek1995chooses,ballou1998case}).
By contrast, other work suggests that licensing requirements may have
little effect on high-quality candidates (e.g.,
\citealt{angrist2004teacher,larsen2020effect}).

Let $s$ index states and $t$ index years. For each state--year cell, let
$y_{st}$ denote the 90th percentile of college SAT scores among teachers in that cell,
which we interpret as a measure of the right-tail (high-quality)
teacher workforce. We consider the QR model
\begin{equation}\label{eq:qr_top_tail}
Q_{y_{st}\mid X_{st},W_{st}}(\tau)
=
\alpha(\tau)
+ X_{st}\beta(\tau)
+ W_{st}^{\top}\gamma(\tau),
\qquad \tau\in(0,1),
\end{equation}
where $X_{st}$ is a measure of licensing stringency and $W_{st}$ collects
controls, including school characteristics, teacher-market conditions,
non-teacher labor-market conditions, education-policy controls, and
political conditions. The parameter of interest is $\beta(\tau)$: a
negative value, $\beta(\tau)<0$, indicates that greater stringency is
associated with a lower right-tail outcome at quantile $\tau$. 

We consider a range of quantiles $\tau\in\{0.10,\ldots,0.90\}$ to allow the effect of licensing stringency to vary across the conditional distribution of this right-tail outcome. Smaller values of $\tau$ correspond to state-year cells in the lower part of the conditional distribution of right-tail teacher quality. In these state-year cells, teachers who are near the top of the quality distribution within that particular state and year may still be relatively less competitive in the broader market and therefore closer to the margin of entering or remaining in teaching. As a result, tighter licensing requirements may have a stronger effect at lower quantiles. By contrast, larger values of $\tau$ correspond to state-year cells in the upper part of the conditional distribution, where right-tail teacher quality is stronger and teachers in that upper tail may be less marginal and more competitive, so the effect of additional licensing barriers may be weaker. We use the publicly available data from \citet{larsen2020effect}, a state-level panel covering 51 state units over 17 years. Based on OLS estimates, \citet{larsen2020effect} report that licensing stringency does not, on average, significantly affect high-quality teacher candidates.

\begin{table}[t]
\centering
\caption{Effects of licensing stringency and $p$-values under different CRVEs.}
\label{tab:qr_pvalues}
\begin{tabular}{lccccccccc}
\hline\hline
$\tau$ & $0.10$ & $0.20$ & $0.30$ & $0.40$ & $0.50$ & $0.60$ & $0.70$ & $0.80$ & $0.90$ \\
\hline
$\hat\beta(\tau)$ & -0.0998 & -0.0870 & -0.0668 & -0.0295 & -0.0277 & -0.0107 & 0.0033 & 0.0143 & 0.0164 \\
CI               & 0.0001  & 0.0014  & 0.0138  & 0.2467  & 0.2725  & 0.7407  & 0.9657 & 0.8561 & 0.9641 \\
CG               & 0.0000  & 0.0000  & 0.0171  & 0.2747  & 0.3864  & 0.7723  & 0.9735 & 0.8919 & 0.9724 \\
CH               & 0.0035  & 0.0125  & 0.0847  & 0.4493  & 0.4961  & 0.8447  & 0.9790 & 0.9095 & 0.9769 \\
CTW              & $\bm{0.0004}$ & $\bm{0.0039}$ & ${0.0898}$ & 0.4610 & 0.5399 & 0.8524 & 0.9812 & 0.9207 & 0.9797 \\
CTW$_{\mathrm{II}}$ & 0.0092  & 0.0320  & 0.1624  & 0.5340  & 0.5925  & 0.8711  & 0.9835 & 0.9305 & 0.9823 \\
\hline\hline
\end{tabular}\label{tab:empirical}
\end{table}

Table~\ref{tab:empirical} reports $\widehat{\beta}(\tau)$ for a grid of
quantiles together with $p$-values computed under several CRVE choices. The main evidence of a right-tail effect arises at low $\tau$. At $\tau=0.10$,
$\widehat{\beta}(0.10)=-0.0998$, and the CTW $p$-value is $0.0004$,
indicating a statistically significant negative association at the 1\%
level. At $\tau=0.20$, $\widehat{\beta}(0.20)=-0.0870$ with a CTW
$p$-value of $0.0039$, again significant at conventional levels. At
$\tau=0.30$, the point estimate remains negative
($\widehat{\beta}(0.30)=-0.0668$), but inference becomes sensitive to the
variance estimator: CI and CG reject at 5\%, whereas CH and CTW are
borderline (around the 10\% level) and CTW$_{\mathrm{II}}$ is more
conservative. For quantiles $\tau\in\{0.40,\ldots,0.90\}$, the estimates
are close to zero and none of the CRVEs yield statistically significant
effects.

Overall, emphasizing the two-way robust CTW inference, the results suggest that licensing stringency may not affect the right tail on average, consistent with \citet{larsen2020effect}, but the effect is heterogeneous across quantiles. In particular, the negative association is concentrated in the lower part of the conditional distribution of $y_{st}$ (roughly $\tau \le 0.20$). One possible interpretation is that, in these markets, a margin of high-quality candidates is more sensitive to licensing costs and therefore more likely to select into alternative occupations. For higher quantiles, we find little evidence that stringency discourages right-tail teacher quality at the 5\% significance level.

\section{Conclusion}

\label{sec:conclusion}

This paper develops a unified large-sample theory and practical inference
procedures for linear quantile regression under two-way clustering.
The key challenge is that both the non-smooth quantile score and the
two-way dependence invalidate standard arguments, and, moreover, the
effective convergence rate of the quantile regression estimator can
vary across dependence regimes. To address these issues, we work within
a separately exchangeable array framework and employ a projection-based
decomposition that isolates row, column, interaction, and idiosyncratic
components. This structure yields an asymptotic distribution theory that adapts to regime-dependent
normalizations.

Building on the limit theory, we propose a feasible two-way cluster-robust
sandwich covariance estimator. We show that both the ``bread'' component
(a kernel estimator of the conditional density at the target quantile)
and the ``meat'' component (an estimator of the covariance of the
sample score that aggregates row and column contributions) are consistent
under appropriate smoothness and moment conditions. The resulting
 procedure is asymptotically valid in the Gaussian regimes,
with a proof that explicitly tracks how regime-dependent rates and
two-way dependence alter the relative magnitude of leading terms and
remainder terms.

Moreover, we clarify the intrinsic limits of uniform inference under
two-way clustering. When the interaction component remains asymptotically
non-negligible while clustering variation along both dimensions is
bounded, the limiting distribution can be non-Gaussian, and uniform
consistency over the full model class may be unattainable without
additional restrictions. 

The simulation results further demonstrate the necessity of using a two-way cluster-robust variance estimator when two-way clustering is present. They also highlight the robustness of the two-way procedure across a range of dependence structures: it remains valid under varying levels of clustering dependence in two dimensions, and even in the absence of within-cluster dependence. In an empirical application, we find that the effect of teacher-licensing stringency on teacher quality is heterogeneous across the distribution. Specifically, tighter licensing requirements are negatively associated with the lower conditional quantiles of the right-tail teacher-quality measure. In contrast, we find little evidence that licensing stringency discourages high-quality teachers at higher quantiles.

 Overall, the paper closes a theoretical gap for quantile
regression with two-way clustered data and offers easy-to-implement inference tools that are directly
applicable in empirical settings where multi-dimensional clustering
is unavoidable. 

\clearpage{}

\appendix

\section{Proof of Theorem \ref{thm:1}}
\begin{proof}
We suppress $\tau$ to save space. To simplify notation, we first present the argument for $d=1$. The extension to fixed $d>1$ follows by applying the joint CLT, the Cramér--Wold device, together with the maintained assumption that the diagonal elements of $\sigma_{j,\Gamma}^2$ are of the same order, which ensures that all coordinates have comparable scaling and that the eigenvalues of the covariance matrices are of the same order.

Let $\mu_{GH}=\left(\frac{H\sigma_{\text{I},\Gamma}^{2}}{1+H\sigma_{\text{I},\Gamma}^{2}},\frac{G\sigma_{\text{II},\Gamma}^{2}}{1+G\sigma_{\text{II},\Gamma}^{2}},\frac{\sigma_{\text{III},\Gamma}^{2}}{1+\sigma_{\text{III},\Gamma}^{2}},\frac{\sigma_{\text{IV},\Gamma}^{2}}{1+\sigma_{\text{IV},\Gamma}^{2}}\right)$,
where the subscript $G$ and $H$ represents the dependence on the
function which can vary with $G$ and $H$, and we allow $G$ and
$H$ to grow to infinity. Observe that $\mu_{GH}\in\left[0,1\right]^{4}$,
and hence by Bolzano-Weierstrass theorem, there exists a convergent
subsequence, which implies that $\left(H\sigma_{\text{I},\Gamma}^{2},G\sigma_{\text{II},\Gamma}^{2},\sigma_{\text{III},\Gamma}^{2},\sigma_{\text{IV},\Gamma}^{2}\right)$
admits a subsequence converging in the extended reals $\left[0,\infty\right]^{4}$.
For notation simplicity, we keep writing $GH$ in place of the selected
subsequence, and hereafter.

Define the sample score
\[
\mathbb{S}(\beta)
=
\frac{1}{GH}\sum_{g=1}^{G}\sum_{h=1}^{H}
\psi_{gh}(\beta)
=
\frac{1}{GH}\sum_{g=1}^{G}\sum_{h=1}^{H}
X_{gh}\Bigl(\tau-\mathbf{1}\{y_{gh}\le X_{gh}^{\top}\beta\}\Bigr),
\]
and let
\(
\mathcal{S}(\beta)=E[\mathbb{S}(\beta)]
=
E\left[
X_{gh}\Bigl(\tau-F_{y\vert X}(X_{gh}^{\top}\beta\vert X_{gh})\Bigr)
\right].
\)
We first obtain the local rate of \(\widehat\beta\). For fixed \(t\), set
\(\beta=\beta_{0}+r_{GH}^{-1/2}t\). By Knight's identity,
\[
\rho_{\tau}(e_{gh}-r_{GH}^{-1/2}X_{gh}^{\top}t)-\rho_{\tau}(e_{gh})
=
-r_{GH}^{-1/2}t^{\top}\Psi_{gh}
+
\int_{0}^{r_{GH}^{-1/2}X_{gh}^{\top}t}
\Bigl(\mathbf{1}\{e_{gh}\le s\}-\mathbf{1}\{e_{gh}\le0\}\Bigr)\,ds .
\]
Using the smoothness of the conditional density at zero and the
stochastic-equicontinuity argument in Lemma
\ref{lemma: stochastic equi main normal}, uniformly over compact sets
of \(t\),
\[
H_{GH}(t)
=r_{GH}
\left[
Q_{GH}(\beta_{0}+r_{GH}^{-1/2}t)-Q_{GH}(\beta_{0})
\right]
=
-t^{\top}r_{GH}^{1/2}\mathbb{S}(\beta_{0})
+\frac{1}{2}t^{\top}D(\tau)t
+o_{P}(1),
\]
where
\(
Q_{GH}(\beta)=\frac{1}{GH}\sum_{g=1}^{G}\sum_{h=1}^{H}
\rho_{\tau}(y_{gh}-X_{gh}^{\top}\beta).
\)
Since \(r_{GH}^{1/2}\mathbb{S}(\beta_{0})=O_{P}(1)\) and
\(D(\tau)\) is positive definite, the local objective is eventually
strictly positive on large spheres. To see this, 
for \(\|t\|=M\),
\[
H_{GH}(t)
\ge
-M\|r_{GH}^{1/2}\mathbb{S}(\beta_{0})\|
+
\frac{1}{2}\lambda_{\min}(D(\tau))M^{2}
+
o_{P}(1).
\]
Because \(r_{GH}^{1/2}\mathbb{S}(\beta_{0})=O_{P}(1)\), The first term is \(O_{P}(M)\), whereas the second term is positive and
of order \(M^{2}\). Therefore, by choosing \(M\) sufficiently large,
\(
\inf_{\|t\|=M}H_{GH}(t)>0
\)
with probability arbitrarily close to one. Since \(H_{GH}(0)=0\) and
\(H_{GH}\) is convex, a minimizer cannot lie outside the ball
\(\{t:\|t\|\le M\}\): otherwise the line segment from \(0\) to the
minimizer would cross the boundary at some point \(t_M\) with
\(\|t_M\|=M\), and convexity would imply
\(
H_{GH}(t_M)\le H_{GH}(0)=0,
\)
contradicting the positivity of \(H_{GH}\) on the boundary. Thus
\(
r_{GH}^{1/2}(\widehat\beta-\beta_{0})=O_{P}(1).
\)

We now derive the Bahadur expansion. By Lemma \ref{lem:approx_score},
\(
r_{GH}^{1/2}\mathbb{S}(\widehat\beta)=o_{P}(1).
\)
Moreover, a Taylor expansion of \(\mathcal{S}(\beta)\) around
\(\beta_{0}\) with \(r_{GH}^{1/2}(\widehat\beta-\beta_{0})=O_{P}(1)\) gives
\(
\mathcal{S}(\widehat\beta)
=
-D(\tau)(\widehat\beta-\beta_{0})
+
o_{P}\!\left(r_{GH}^{-1/2}\right).
\)
 Therefore, combining the previous results
with
\(
\nu_{S}(\beta)
=
r_{GH}^{1/2}\{\mathcal{S}(\beta)-\mathbb{S}(\beta)\},
\)
we have
\[
o_{P}(1)
=
r_{GH}^{1/2}\mathbb{S}(\widehat\beta)
=
-D(\tau)r_{GH}^{1/2}(\widehat\beta-\beta_{0})
-\nu_{S}(\widehat\beta)
+o_{P}(1).
\]
Since \(\mathcal{S}(\beta_{0})=0\),
\(
r_{GH}^{1/2}\mathbb{S}(\beta_{0})=-\nu_{S}(\beta_{0}).
\)
By Lemma \ref{lemma: stochastic equi main normal},
\(
\nu_{S}(\widehat\beta)-\nu_{S}(\beta_{0})=o_{P}(1).
\)
Consequently,
\[
r_{GH}^{1/2}(\widehat\beta-\beta_{0})
=
D(\tau)^{-1}r_{GH}^{1/2}\mathbb{S}(\beta_{0})
+o_{P}(1).
\]

We now establish the Asymptotic normality of $r_{GH}^{1/2}\mathbb{S}(\beta_{0})$.
Recall $\mathbb{S}(\beta_{0})=\frac{1}{GH}\sum_{g=1}^{G}\sum_{h=1}^{H}\Psi_{gh}$,
with the Hoeffding-type decomposition $\Psi_{gh}=\Psi_{g}^{(\mathrm{I})}+\Psi_{h}^{(\mathrm{II})}+\Psi_{gh}^{(\mathrm{III})}+\Psi_{gh}^{(\mathrm{IV})},$
we can write 
\begin{align*}
r_{GH}^{1/2}\mathbb{S}(\beta_{0}) & =\sqrt{\frac{r_{GH}}{G}}\,\sigma_{\mathrm{I},\Gamma}\,S_{G}^{(\mathrm{I})}+\sqrt{\frac{r_{GH}}{H}}\,\sigma_{\mathrm{II},\Gamma}\,S_{H}^{(\mathrm{II})}+\sqrt{\frac{r_{GH}}{GH}}\,\sigma_{\mathrm{III},\Gamma}\,S_{GH}^{(\mathrm{III})}+\sqrt{\frac{r_{GH}}{GH}}\,\sigma_{\mathrm{IV},\Gamma}\,S_{GH}^{(\mathrm{IV})},
\end{align*}
where 
\[
S_{G}^{(\mathrm{I})}:=\frac{1}{\sqrt{G}}\sum_{g=1}^{G}\sigma_{\mathrm{I},\Gamma}^{-1}\Psi_{g}^{(\mathrm{I})},\quad S_{H}^{(\mathrm{II})}:=\frac{1}{\sqrt{H}}\sum_{h=1}^{H}\sigma_{\mathrm{II},\Gamma}^{-1}\Psi_{h}^{(\mathrm{II})},
\]
\[
S_{GH}^{(\mathrm{III})}:=\frac{1}{\sqrt{GH}}\sum_{g=1}^{G}\sum_{h=1}^{H}\sigma_{\mathrm{III},\Gamma}^{-1}\Psi_{gh}^{(\mathrm{III})},\quad S_{GH}^{(\mathrm{IV})}:=\frac{1}{\sqrt{GH}}\sum_{g=1}^{G}\sum_{h=1}^{H}\sigma_{\mathrm{IV},\Gamma}^{-1}\Psi_{gh}^{(\mathrm{IV})}.
\]

\medskip{}
\emph{Case 1: $H\sigma_{\mathrm{I},\Gamma}^{2}+G\sigma_{\mathrm{II},\Gamma}^{2}\to\infty$.}
Assume without loss of generality $H\sigma_{\mathrm{I},\Gamma}^{2}\ge G\sigma_{\mathrm{II},\Gamma}^{2}$,
so that $r_{GH}=G/\sigma_{\mathrm{I},\Gamma}^{2}$ and $\sqrt{r_{GH}/G}\,\sigma_{\mathrm{I},\Gamma}=1$.
Moreover, 
\[
\sqrt{\frac{r_{GH}}{H}}\,\sigma_{\mathrm{II},\Gamma}=\sqrt{\frac{G}{H}}\frac{\sigma_{\mathrm{II},\Gamma}}{\sigma_{\mathrm{I},\Gamma}}\to \sqrt{\lambda}, \qquad \lambda := \lim \frac{G\sigma_{\mathrm{II},\Gamma}^{2}}{H\sigma_{\mathrm{I},\Gamma}^{2}} \in [0,1],\qquad\sqrt{\frac{r_{GH}}{GH}}\to0.
\]

Since $\{\sigma_{\mathrm{I},\Gamma}^{-1}\Psi_{g}^{(\mathrm{I})}\}_{g\le G}$
are i.i.d., a Lyapunov CLT gives $S_{G}^{(\mathrm{I})}\overset{d}{\to}\mathcal{N}(0,1)$,
and similarly $S_{H}^{(\mathrm{II})}\overset{d}{\to}\mathcal{N}(0,1)$.
Because $S_{GH}^{(\mathrm{III})}=O_{P}(1)$ and $S_{GH}^{(\mathrm{IV})}=O_{P}(1)$
(see Case 2) and $\sqrt{r_{GH}/(GH)}=o(1)$, the last two terms are
$o_{P}(1)$. Moreover, provided that $S_{G}^{\mathrm{I}}$ and $S_{H}^{\mathrm{II}}$
are independent, the joint CLT yields that 
\[
r_{GH}^{1/2}\mathbb{S}(\beta_{0})\overset{d}{\to}\mathcal{N}\!\left(0,1+\lambda\right).
\]
Furthermore, with $\Omega_{GH}=\frac{1}{GH}\bigl(H\sigma_{\mathrm{I},\Gamma}^{2}+G\sigma_{\mathrm{II},\Gamma}^{2}+\sigma_{\mathrm{III},\Gamma}^{2}+\sigma_{\mathrm{IV},\Gamma}^{2}\bigr)$,
\[
r_{GH}\Omega_{GH}=\frac{G}{\sigma_{\mathrm{I},\Gamma}^{2}}\cdot\frac{1}{GH}\Bigl(H\sigma_{\mathrm{I},\Gamma}^{2}+G\sigma_{\mathrm{II},\Gamma}^{2}+\sigma_{\mathrm{III},\Gamma}^{2}+\sigma_{\mathrm{IV},\Gamma}^{2}\Bigr)=1+\lambda+o(1),
\]
so Slutsky's lemma yields $\Omega_{GH}^{-1/2}\mathbb{S}(\beta_{0})\overset{d}{\to}\mathcal{N}(0,1)$.

\medskip{}
\emph{Case 2: $r_{GH}\asymp GH$ (equivalently, $H\sigma_{\mathrm{I},\Gamma}^{2}+G\sigma_{\mathrm{II},\Gamma}^{2}=O(1)$
and $\sigma_{\mathrm{III},\Gamma}^{2}=o(1)$).} For simplicity, we use $r_{GH}= GH$. Using $E(\Psi_{gh}^{(\mathrm{IV})}\vert U_{g})=E(\Psi_{gh}^{(\mathrm{IV})}\vert V_{h})=0$,
we have for $(g,h)\neq(g',h')$ that $E(\Psi_{gh}^{(\mathrm{IV})}\Psi_{g'h'}^{(\mathrm{IV})})=0$,
hence 
\[
Var\!\bigl(S_{GH}^{(\mathrm{IV})}\bigr)=\frac{1}{GH}\sum_{g=1}^{G}\sum_{h=1}^{H}E\!\left(\sigma_{\mathrm{IV},\Gamma}^{-1}\Psi_{gh}^{(\mathrm{IV})}\Psi_{gh}^{(\mathrm{IV})}\sigma_{\mathrm{IV},\Gamma}^{-1}\right)<\infty,\qquad S_{GH}^{(\mathrm{IV})}=O_{P}(1).
\]
Let $\mathcal{F}_{GH}:=\sigma(\{U_{g}\}_{g\le G},\{V_{h}\}_{h\le H})$.
Then $E(\Psi_{gh}^{(\mathrm{IV})}\vert\mathcal{F}_{GH})=0$ and, conditional
on $\mathcal{F}_{GH}$, $\{\sigma_{\mathrm{IV},\Gamma}^{-1}\Psi_{gh}^{(\mathrm{IV})}\}_{g,h}$
are independent. Define 
\[
V_{GH}^{(\mathrm{IV})}:=Var\!\bigl(S_{GH}^{(\mathrm{IV})}\vert\mathcal{F}_{GH}\bigr)=\frac{1}{GH}\sum_{g=1}^{G}\sum_{h=1}^{H}\sigma_{\mathrm{IV},\Gamma}^{-1}E\!\left(\Psi_{gh}^{(\mathrm{IV})}\Psi_{gh}^{(\mathrm{IV})}\vert U_{g},V_{h}\right)\sigma_{\mathrm{IV},\Gamma}^{-1}.
\]
A conditional Lyapunov CLT yields $\bigl(V_{GH}^{(\mathrm{IV})}\bigr)^{-1/2}S_{GH}^{(\mathrm{IV})}\vert\mathcal{F}_{GH}\overset{d}{\to}\mathcal{N}(0,1)$,
and a LLN with the law of total expectation implies $V_{GH}^{(\mathrm{IV})}=1+o_{P}(1)$.

We have marginal CLT for different terms, and we now establish a \emph{joint}
CLT for $\bigl(S_{G}^{(\mathrm{I})},S_{H}^{(\mathrm{II})},S_{GH}^{(\mathrm{IV})}\bigr)$
via characteristic functions. Let 
\[
S_{G}^{(\mathrm{I})}:=\frac{1}{\sqrt{G}}\sum_{g=1}^{G}\sigma_{\mathrm{I},\Gamma}^{-1}\Psi_{g}^{(\mathrm{I})},\qquad S_{H}^{(\mathrm{II})}:=\frac{1}{\sqrt{H}}\sum_{h=1}^{H}\sigma_{\mathrm{II},\Gamma}^{-1}\Psi_{h}^{(\mathrm{II})},
\]
\[
S_{GH}^{(\mathrm{IV})}:=\frac{1}{\sqrt{GH}}\sum_{g=1}^{G}\sum_{h=1}^{H}\sigma_{\mathrm{IV},\Gamma}^{-1}\Psi_{gh}^{(\mathrm{IV})},\qquad\mathcal{F}_{GH}:=\sigma(\{U_{g}\}_{g\le G},\{V_{h}\}_{h\le H}).
\]
For $(u,v,w)\in\mathbb{R}^{3}$, define the characteristic function
$\phi_{GH}(u,v,w):=E\exp\!\Bigl(iuS_{G}^{(\mathrm{I})}+ivS_{H}^{(\mathrm{II})}+iwS_{GH}^{(\mathrm{IV})}\Bigr).$
By iterated expectations, 
\begin{align*}
\phi_{GH}(u,v,w) & =E\Bigl[\exp\!\bigl(iuS_{G}^{(\mathrm{I})}+ivS_{H}^{(\mathrm{II})}\bigr)\,E\!\left(\exp\!\bigl(iwS_{GH}^{(\mathrm{IV})}\bigr)\vert\mathcal{F}_{GH}\right)\Bigr].
\end{align*}
Recall that conditional on $\mathcal{F}_{GH}$, $\{\sigma_{\mathrm{IV},\Gamma}^{-1}\Psi_{gh}^{(\mathrm{IV})}\}_{g,h}$
are independent with mean zero and conditional variance $V_{GH}^{(\mathrm{IV})}:=Var\!\bigl(S_{GH}^{(\mathrm{IV})}\vert\mathcal{F}_{GH}\bigr)=1+o_{P}(1),$
so the conditional Lyapunov CLT gives, for each fixed $w$, 
\[
E\!\left(\exp\!\bigl(iwS_{GH}^{(\mathrm{IV})}\bigr)\vert\mathcal{F}_{GH}\right)\;\overset{P}{\to}\;\exp\!\left(-\tfrac{1}{2}w^{2}\right).
\]
Since $\bigl|\exp(iuS_{G}^{(\mathrm{I})}+ivS_{H}^{(\mathrm{II})})\bigr|\le1$,
dominated convergence yields 
\[
\phi_{GH}(u,v,w)\to\exp\!\left(-\tfrac{1}{2}w^{2}\right)\cdot\lim_{G,H\to\infty}E\exp\!\bigl(iuS_{G}^{(\mathrm{I})}+ivS_{H}^{(\mathrm{II})}\bigr).
\]
Finally, since $\{U_{g}\}$ and $\{V_{h}\}$ are independent and each
array is i.i.d., the (marginal) Lyapunov CLT implies $E\exp\!\bigl(iuS_{G}^{(\mathrm{I})}+ivS_{H}^{(\mathrm{II})}\bigr)=E\exp\!\bigl(iuS_{G}^{(\mathrm{I})}\bigr)\,E\exp\!\bigl(ivS_{H}^{(\mathrm{II})}\bigr)\to\exp\!\left(-\tfrac{1}{2}u^{2}-\tfrac{1}{2}v^{2}\right),$
and hence $\phi_{GH}(u,v,w)\to\exp\!\left(-\tfrac{1}{2}u^{2}-\tfrac{1}{2}v^{2}-\tfrac{1}{2}w^{2}\right).$
By Lévy's continuity theorem, 
\begin{align}
\bigl(S_{G}^{(\mathrm{I})},S_{H}^{(\mathrm{II})},S_{GH}^{(\mathrm{IV})}\bigr)\ \overset{d}{\to}\ \mathcal{N}(0,\mathbf{I}_{3}).\label{eq: joint clt}
\end{align}
Moreover, in such case we have $\sigma_{\mathrm{III},\Gamma}S_{GH}^{(\mathrm{III})}=o_{P}(1)$
and hence 
\[
\frac{\sqrt{r_{GH}}}{GH}\sum_{g=1}^{G}\sum_{h=1}^{H}\Psi_{gh}\overset{d}{\to}\mathcal{N}\!\left(0,\lim_{G,H\to\infty}\left(H\sigma_{\mathrm{I},\Gamma}^{2}+G\sigma_{\mathrm{II},\Gamma}^{2}+\sigma_{\mathrm{IV},\Gamma}^{2}\right)\right).
\]
Here, the limiting variance is positive definite by Assumption \ref{as:order of variance}.
Moreover, $\lim_{G,H\to\infty}\sigma_{\mathrm{IV},\Gamma}^{2}<\infty$
by Jensen's inequality and $E\!\left(\Psi_{gh}\Psi_{gh}^{\top}\right)\le E\!\left\Vert X_{gh}\right\Vert ^{4}<\infty.$
Finally, the application of Slutsky's lemma yields $\Omega_{GH}^{-1/2}\mathbb{S}(\beta_{0})\overset{d}{\to}\mathcal{N}(0,1)$.

Finally, because the above argument holds for any convergent subsequence,
the claimed uniformity follows from the convergent-subsequence argument
together with continuity of the limiting distribution in the parameter; see, e.g., \citet{davezies2021empirical} or Lemma IA.1 in
\citet{hounyo2025projection}.
\end{proof}

\section{Proof of Theorem \ref{thm:Jacobian}}
\begin{proof}
Fix an arbitrary deterministic matrix $B\in\mathbb{R}^{d\times d}$
and define the scalar weight $\mathfrak{X}_{gh}:=\mathrm{tr}(BX_{gh}X_{gh}^{\top})$.
Let 
\[
\widehat{D}(\beta):=\frac{1}{GH\,\ell}\sum_{g=1}^{G}\sum_{h=1}^{H}K\!\left(\frac{y_{gh}-X_{gh}^{\top}\beta}{\ell}\right)\mathfrak{X}_{gh},\qquad K(u)=\tfrac{1}{2}\mathbf{1}\{|u|\le1\}.
\]
Given that $B$ is arbitrary and by Cramer–Wold device, it suffices
to focus on 
\[
\left(r_{GH,D}\right)^{1/2}\left\{ \widehat{D}(\widehat{\beta})-E[\mathfrak{X}_{gh}f_{e|X}(0\vert X_{gh})]-\frac{\ell^{2}}{6}E[\mathfrak{X}_{gh}f_{e\vert X}^{\left(2\right)}(0\vert X_{gh})+o\left(\ell^{2}\right)\right\} 
\]
Write {\footnotesize{}
\begin{align*}
 & \widehat{D}(\widehat{\beta})-E[\mathfrak{X}_{gh}f_{e|X}(0\vert X_{gh})]-\frac{\ell^{2}}{6}E[\mathfrak{X}_{gh}f_{e\vert X}^{\left(2\right)}(0\vert X_{gh})]=\underbrace{\left(\widehat{D}(\widehat{\beta})-E[\widehat{D}(\beta)]\vert_{\beta=\widehat{\beta}}\right)-\left(\widehat{D}(\beta_{0})-E[\widehat{D}(\beta_{0})]\right)}_{(\mathrm{I})}\\
 & +\underbrace{\left(\widehat{D}(\beta_{0})-E[\widehat{D}(\beta_{0})]\right)}_{(\mathrm{II})}+\underbrace{\left(E[\widehat{D}(\beta)]\vert_{\beta=\widehat{\beta}}-E[\widehat{D}(\beta_{0})]\right)}_{(\mathrm{III})}+\underbrace{\big(E[\widehat{D}(\beta_{0})]-E[\mathfrak{X}_{gh}f_{e|X}(0\vert X_{gh})]\big)-\frac{\ell^{2}}{6}E[\mathfrak{X}_{gh}f_{e\vert X}^{\left(2\right)}(0\vert X_{gh})}_{(\mathrm{IV})}.
\end{align*}
}{\footnotesize\par}

\paragraph{Term I, stochastic term at $\widehat{\beta}$.}

Observe that $r_{GH}^{1/2}\left(\widehat{\beta}-\beta_{0}\right)=O_{P}\left(1\right)$.
Hence, applying Lemma \ref{lem: stochastic equicontinuity} yields
that $(\mathrm{I})=o_{P}(r_{GH}^{-1/2}\ell^{-1/2})$.

\paragraph{Term II, Consistency and CLT at $\beta_{0}$.}

\ 

\emph{Consistency:} Define $Z_{gh}:=\ell^{-1}K(e_{gh}/\ell)\mathfrak{X}_{gh}$
so that $\widehat{D}(\beta_{0})=(GH)^{-1}\sum_{g,h}Z_{gh}$. Under
two-way clustering, a convenient way to control $\mathrm{Var}(\widehat{D}(\beta_{0}))$
is via the two-way Hoeffding/ANOVA decomposition: write $Z_{gh}-EZ_{gh}=Z_{g\cdot}^{(\mathrm{I})}+Z_{\cdot h}^{(\mathrm{II})}+Z_{gh}^{(\mathrm{III})}$,
where 
\begin{align*}
Z_{g\cdot}^{(\mathrm{I})} & :=E[Z_{gh}\vert U_{g}]-EZ_{gh},\\
Z_{\cdot h}^{(\mathrm{II})} & :=E[Z_{gh}\vert V_{h}]-EZ_{gh},\\
Z_{gh}^{(\mathrm{III})} & :=E[Z_{gh}\vert U_{g},V_{h}]-E[Z_{gh}\vert U_{g}]-E[Z_{gh}\vert V_{h}]+EZ_{gh},\\
Z_{gh}^{(\mathrm{IV})} & :=Z_{gh}-E[Z_{gh}\vert U_{g},V_{h}].
\end{align*}
Then 
\[
\widehat{D}(\beta_{0})-E[\widehat{D}(\beta_{0})]=\frac{1}{G}\sum_{g=1}^{G}Z_{g\cdot}^{(\mathrm{I})}+\frac{1}{H}\sum_{h=1}^{H}Z_{\cdot h}^{(\mathrm{II})}+\frac{1}{GH}\sum_{g=1}^{G}\sum_{h=1}^{H}\left(Z_{gh}^{(\mathrm{III})}+Z_{gh}^{(\mathrm{IV})}\right).
\]
By orthogonality of these projections, we have 
\[
\mathrm{Var}\big(\widehat{D}(\beta_{0})\big)=\frac{1}{G}\mathrm{Var}(Z_{g\cdot}^{(\mathrm{I})})+\frac{1}{H}\mathrm{Var}(Z_{\cdot h}^{(\mathrm{II})})+\frac{1}{GH}\left(Var(Z_{gh}^{(\mathrm{III})})+Var(Z_{gh}^{(\mathrm{IV})})\right).
\]
By conditional Jensen, each second moment is bounded by $E[Z_{gh}^{2}]$
up to a constant. Since $K(u)=\tfrac{1}{2}\mathbf{1}\{|u|\le1\}$,
we have $Z_{gh}^{2}=\ell^{-2}\cdot\tfrac{1}{4}\,\mathbf{1}\{|e_{gh}|\le\ell\}\,\mathfrak{X}_{gh}^{2}$
and hence 
\[
E[Z_{gh}^{2}]=\frac{1}{4\ell^{2}}E\!\left[\mathbf{1}\{|e_{gh}|\le\ell\}\,\mathfrak{X}_{gh}^{2}\right]=\frac{1}{4\ell^{2}}E\left[\mathfrak{X}_{gh}^{2}\int_{-\ell}^{\ell}f_{e|X}(e\vert X_{gh})\,de\right]=\frac{1}{2\ell}\,E[\mathfrak{X}_{gh}^{2}f_{e|X}(0\vert X_{gh})]+o({\ell}^{-1}).
\]
Therefore 
\[
\mathrm{Var}\big(\widehat{D}(\beta_{0})\big)\lesssim\Big(\frac{1}{G}+\frac{1}{H}+\frac{1}{GH}\Big)\frac{1}{\ell}.
\]
The right-hand side  is $O(R^{-1}\ell^{-1})=o(1)$,
which implies by Chebyshev's inequality that $\widehat{D}(\beta_{0})-E[\widehat{D}(\beta_{0})]=o_{P}(1)$.

\emph{CLT result.} Now, we show the CLT result. By an analogous argument
as those for Term IV below, we have 
\begin{align*}
E[Z_{gh}\vert U_{g},V_{h}]= & E[\mathfrak{X}_{gh}f_{e|X,U,V}(0\vert X_{gh},U_{g},V_{h})\vert U_{g},V_{h}]\\
 & +\frac{\ell^{2}}{6}E[\mathfrak{X}_{gh}f_{e\vert X,U,V}^{\left(2\right)}(0\vert X_{gh},U_{g},V_{h})\vert U_{g},V_{h}]+o\left(\ell^{2}\right),\\
E[Z_{gh}^{2}\vert U_{g},V_{h}]= & \frac{1}{2\ell}E[\mathfrak{X}_{gh}^{2}f_{e|X,U,V}(0\vert X_{gh},U_{g},V_{h})\vert U_{g},V_{h}]\\
 & +\frac{\ell}{12}E[\mathfrak{X}_{gh}^{2}f_{e\vert X,U,V}^{\left(2\right)}(0\vert X_{gh},U_{g},V_{h})\vert U_{g},V_{h}]+o\left(\ell\right),
\end{align*}
Hence, $Var\left(E[Z_{gh}\vert U_{g},V_{h}]\right)=Var\left(E[\mathfrak{X}_{gh}f_{e|X,U,V}(0\vert X_{gh},U_{g},V_{h})\vert U_{g},V_{h}]\right)+o\left(1\right).$
Similarly, we have 
\begin{align*}
Var\left(Z_{g\cdot}^{(\mathrm{I})}\right) & =Var\left(E[Z_{gh}\vert U_{g}]\right)=Var\left(E[\mathfrak{X}_{gh}f_{e|X,U}(0\vert X_{gh},U_{g})\vert U_{g}]\right)+o\left(1\right):=\sigma_{\mathrm{I},Z}^{2}+o\left(1\right),\\
Var\left(Z_{\cdot h}^{(\mathrm{II})}\right) & =Var\left(E[Z_{gh}\vert V_{h}]\right)=Var\left(E[\mathfrak{X}_{gh}f_{e|X,V}(0\vert X_{gh},V_{h})\vert V_{h}]\right)+o\left(1\right):=\sigma_{\mathrm{II},Z}^{2}+o\left(1\right),\\
Var\left(Z_{gh}^{(\mathrm{IV})}\right) & =\frac{1}{2\ell}E[\mathfrak{X}_{gh}^{2}f_{e|X}(0\vert X_{gh})]+o\left(\ell^{-1}\right).
\end{align*}
 Moreover, by Assumptions \ref{as:bandwidth}(i) and (v), we have $E[\mathfrak{X}_{gh}^{2}f_{e|X}(0\vert X_{gh})]>0$,
which implies that the non-Gaussian term $Var\left(Z_{gh}^{(\mathrm{III})}\right)=O\left(1\right)$
is negligible compared to $Var\left(Z_{gh}^{(\mathrm{IV})}\right)$.

Given that $Z_{g\cdot}^{(\mathrm{I})}$ is i.i.d. over $g$. By Lapunov's
central limit theorem, we deduce that the marginal CLT result 
\[
\frac{1}{\sqrt{G}}\sum_{g=1}^{G}\sigma_{\mathrm{I},Z}^{-1}Z_{g\cdot}^{(\mathrm{I})}\overset{d}{\to}\mathcal{N}\left(0,\mathbf{I}_{d}\right).
\]
Similarly, we can deduce that $\frac{1}{\sqrt{H}}\sum_{h=1}^{H}\sigma_{\mathrm{II},Z}^{-1}Z_{\cdot h}^{(\mathrm{II})}\overset{d}{\to}\mathcal{N}\left(0,\mathbf{I}_{d}\right)$.
Applying the similar marginal CLT and joint CLT arguments as those for \eqref{eq: joint clt}, we
have 
\[
\frac{\sqrt{\ell}}{\sqrt{GH}}\sum_{g=1}^{G}\sum_{h=1}^{H}Z_{gh}^{(\mathrm{IV})}\overset{d}{\to}\mathcal{N}\left(0,\frac{1}{2}E[\mathfrak{X}_{gh}^{2}f_{e|X}(0\vert X_{gh})]\right),
\]
and 
\begin{align*}
r_{GH,D}^{1/2}\left(\widehat{D}(\beta_{0})-E[\widehat{D}(\beta_{0})]\right)= & \left(\frac{r_{GH,D}\sigma_{\mathrm{I},Z}^{2}}{G}\right)^{1/2}\frac{1}{\sqrt{G}}\sum_{g=1}^{G}\sigma_{\mathrm{I},Z}^{-1}Z_{g\cdot}^{(\mathrm{I})}+\left(\frac{r_{GH,D}\sigma_{\mathrm{II},Z}^{2}}{H}\right)^{1/2}\frac{1}{\sqrt{H}}\sum_{h=1}^{H}\sigma_{\mathrm{II},Z}^{-1}Z_{\cdot h}^{(\mathrm{II})}\\
 & +\left(\frac{r_{GH,D}}{GH\ell}\right)^{1/2}\frac{\sqrt{\ell}}{\sqrt{GH}}\sum_{g=1}^{G}\sum_{h=1}^{H}Z_{gh}^{(\mathrm{IV})}+o_{P}\left(1\right)\\
\overset{d}{\to} & \mathcal{N}\left(0,\nu_{\mathrm{I}}+\nu_{\mathrm{II}}+\frac{\nu_{\mathrm{IV}}}{2}E[\mathfrak{X}_{gh}^{2}f_{e|X}(0\vert X_{gh})]\right),
\end{align*}
where $\nu_{\mathrm{I}}=\lim_{N,T\to\infty}\frac{r_{GH,D}\sigma_{\mathrm{I},Z}^{2}}{G}$,
$\nu_{\mathrm{II}}=\lim_{N,T\to\infty}\frac{r_{GH,D}\sigma_{\mathrm{II},Z}^{2}}{H}$,
and $\nu_{\mathrm{IV}}=\lim_{N,T\to\infty}\frac{r_{GH,D}}{GH\ell}$. Here, we focus on any convergent subsequence such that $\nu_{\bullet}$ is well-defined.

\paragraph{Term III, plug-in error in expectation.}

Assume $r_{GH}^{1/2}\|\widehat{\beta}-\beta_{0}\|\le C_{0}$, where
$r_{GH}\to\infty$. On this event write $\widehat{\beta}=\beta_{0}+r_{GH}^{-1/2}t$
with $\|t\|\le C_{0}$. Then, similarly as above, 
\[
E[\widehat{D}(\beta_{0}+r_{GH}^{-1/2}t)]=E\!\left[\mathfrak{X}_{gh}\int K(v)\,f_{e|X}(\ell v+r_{GH}^{-1/2}X_{gh}^{\top}t\vert X_{gh})\,dv\right].
\]
By the mean value theorem, for each $(v,t)$ there exists an intermediate
point between $\ell v$ and $\ell v+r_{GH}^{-1/2}X_{gh}^{\top}t$
such that 
\[
\Big|f_{e|X}(\ell v+r_{GH}^{-1/2}X_{gh}^{\top}t\vert X_{gh})-f_{e|X}(\ell v\vert X_{gh})\Big|\le r_{GH}^{-1/2}\,|X_{gh}^{\top}t|\,\sup_{u}|f_{e|X}^{(1)}(u\vert X_{gh})|.
\]
Using $|X_{gh}^{\top}t|\le\|X_{gh}\|\|t\|\le C_{0}\|X_{gh}\|$ and
the assumed bound $\sup_{e,x}|f_{e|X}^{(1)}(e\vert X_{gh})|<\infty$,
we obtain 
\[
\sup_{\|t\|\le C_{0}}\Big|f_{e|X}(\ell v+r_{GH}^{-1/2}X_{gh}^{\top}t\vert X_{gh})-f_{e|X}(\ell v\vert X_{gh})\Big|\lesssim r_{GH}^{-1/2}\|X_{gh}\|.
\]
Since $\int|K(v)|\,dv=1$, Jensen's inequality implies 
\begin{align*}
 & \left|E[\widehat{D}(\beta)]\vert_{\beta=\widehat{\beta}}-E\left[\widehat{D}\left(\beta_{0}\left(\tau\right)\right)\right]\right|\\
 & \le\sup_{\left\Vert t\right\Vert \le C_{0}}E\left[\left|\mathfrak{X}_{gh}\right|\int\left|K\left(v\right)\right|\left|f_{e|X}\left(\ell v+r_{GH}^{-1/2}X_{gh}^{\top}t\vert X_{gh}\right)-f_{e|X}\left(\ell v\vert X_{gh}\right)\right|dv\right]\\
 & \lesssim\sup_{\left\Vert t\right\Vert \le C_{0}}E\left[\left|\mathfrak{X}_{gh}\right|\int\left|K\left(v\right)\right|r_{GH}^{-1/2}\left\Vert X_{gh}\right\Vert dv\right]\\
 & \le r_{GH}^{-1/2}E\left[\left|\mathfrak{X}_{gh}\right|\left\Vert X_{gh}\right\Vert \right].
\end{align*}
Thus given the fourth moment of $X_{gh}$ is bounded, one can deduce
that $E[\widehat{D}(\beta)]\vert_{\beta=\widehat{\beta}}=E[\widehat{D}(\beta_{0})]+O(r_{GH}^{-1/2})=E[\mathfrak{X}_{gh}f_{e|X}(0\vert X_{gh})]+O(r_{GH}^{-1/2})$.\medskip{}

\paragraph{Term IV, bias at $\beta_{0}$.}

Let $\beta_{0}=\beta_{0}(\tau)$ and define the regression error $e_{gh}:=y_{gh}-X_{gh}^{\top}\beta_{0}$.
By conditioning on $X_{gh}$, we have 
\[
E[\widehat{D}(\beta_{0})]=E\!\left[\mathfrak{X}_{gh}\,E\!\left[\frac{1}{\ell}K\!\left(\frac{e_{gh}}{\ell}\right)\Big|X_{gh}\right]\right]=E\!\left[\mathfrak{X}_{gh}\int K(v)\,f_{e|X}(\ell v\vert X_{gh})\,dv\right].
\]
Using a second-order Taylor expansion of $f_{e|X}(\cdot\vert X_{gh})$
at $0$, 
\[
f_{e|X}(\ell v\vert X_{gh})=f_{e|X}(0\vert X_{gh})+\ell vf_{e\vert X}^{\left(1\right)}(0\vert X_{gh})+\frac{\ell^{2}v^{2}}{2}f_{e\vert X}^{\left(2\right)}(0\vert X_{gh})+o(\ell^{2}),
\]
uniformly over $|v|\le1$ (the support of $K$). Since $\int K(v)\,dv=1$,
$\int vK(v)\,dv=0$, and $\int v^{2}K(v)\,dv=1/3<\infty$, it follows
that 
\[
\int K(v)\,f_{e|X}(\ell v\vert X_{gh})\,dv=f_{e|X}(0\vert X_{gh})+\frac{\ell^{2}}{6}f_{e\vert X}^{\left(2\right)}(0\vert X_{gh})+o(\ell^{2}).
\]
Therefore $E[\widehat{D}(\beta_{0})]=E[\mathfrak{X}_{gh}f_{e|X}(0\vert X_{gh})]+\frac{\ell^{2}}{6}E[\mathfrak{X}_{gh}f_{e\vert X}^{\left(2\right)}(0\vert X_{gh})]+o\left(\ell^{2}\right)$.

\paragraph{Conclusion.}

Together, Terms I-IV show that 
\[
\widehat{D}(\widehat{\beta})=E[\mathfrak{X}_{gh}f_{e|X}(0\vert X_{gh})]+O_{P}\left(r_{GH}^{-1/2}\ell^{-1/2}+\ell^{2}+R^{-1/2}\ell^{-1/2}\right)=E[\mathfrak{X}_{gh}f_{e|X}(0\vert X_{gh})]+o_{P}\left(1\right).
\]
Moreover, when $\frac{\sigma_{\mathrm{I},1\Gamma}^{2}}{\ell\sigma_{\mathrm{I},Z}^{2}}=O\left(1\right)$
and $\frac{\sigma_{\mathrm{II},1\Gamma}^{2}}{\ell\sigma_{\mathrm{II},Z}^{2}}=O\left(1\right)$,
or $H\sigma_{\mathrm{I},1\Gamma}^{2}+G\sigma_{\mathrm{II},1\Gamma}^{2}=O\left(1\right)$,
we have $r_{GH,D}^{1/2}r_{GH}^{-1/2}\ell^{-1/2}=O\left(1\right)$ and $r_{GH,D}^{1/2}r_{GH}^{-1/2}=o(1)$,
which implies 
\begin{align*}
&r_{GH,D}^{1/2}  \left(\widehat{D}(\widehat{\beta})-E[\mathfrak{X}_{gh}f_{e|X}(0\vert X_{gh})]-\frac{\ell^{2}}{6}E[\mathfrak{X}_{gh}f_{e\vert X}^{\left(2\right)}(0\vert X_{gh})]+o\left(\ell^{2}\right)\right)\\
&\qquad\qquad\qquad\qquad\qquad\qquad\qquad\qquad\overset{d}{\to}  \mathcal{N}\left(0,\nu_{\mathrm{I}}+\nu_{\mathrm{II}}+\frac{\nu_{\mathrm{IV}}}{2}E[\mathfrak{X}_{gh}^{2}f_{e|X}(0\vert X_{gh})]\right).
\end{align*}
Since $\mathfrak{X}_{gh}=\mathrm{tr}(BX_{gh}X_{gh}^{\top})$ and $B$
is arbitrary, the Cramer-Wold Device implies the pointwise result along any convergent subsequence. By the analogous argument as proof for Theorem \ref{thm:1},
one can extend it to show the uniformity result, which completes the proof.
\end{proof}

\section{Proof of Theorem \ref{thm:2-1}}
\begin{proof}
Define the oracle variance estimator $\widetilde{\Omega}=\widetilde{\Omega}_{\mathrm{I}}+\widetilde{\Omega}_{\mathrm{II}}+\widetilde{\Omega}_{\mathrm{III,IV}}$,
with 
\begin{align}
\widetilde{\Omega}_{\mathrm{I}} & =\frac{1}{G^{2}H^{2}}\sum_{g=1}^{G}\sum_{h=1}^{H}\sum_{h'\neq h}^{H}\Psi_{gh}\Psi_{gh'}^{\top},\label{eq: omega 1}\\
\widetilde{\Omega}_{\mathrm{II}} & =\frac{1}{G^{2}H^{2}}\sum_{g=1}^{G}\sum_{g'\ne g}^{G}\sum_{h=1}^{H}\Psi_{gh}\Psi_{g'h}^{\top},\label{eq: omega 2}\\
\widetilde{\Omega}_{\mathrm{III,IV}} & =\frac{1}{G^{2}H^{2}}\sum_{g=1}^{G}\sum_{h=1}^{H}\Psi_{gh}\Psi_{gh}^{\top}.\label{eq: omega 3}
\end{align}
We first show $r_{GH}\left(\widehat{\Omega}-\widetilde{\Omega}\right)=o_{P}\left(1\right)$,
and we decompose into three terms 
\begin{align*}
r_{GH}\left(\widehat{\Omega}-\widetilde{\Omega}\right)= & r_{GH}\left(\widehat{\Omega}_{\mathrm{I}}-\widetilde{\Omega}_{\mathrm{I}}\right)+r_{GH}\left(\widehat{\Omega}_{\text{II}}-\widetilde{\Omega}_{\mathrm{II}}\right)+r_{GH}\left(\widehat{\Omega}_{\mathrm{III,IV}}-\widetilde{\Omega}_{\mathrm{III,IV}}\right).
\end{align*}
Without loss of generality, let $d=1$ hereafter. As in proof for
Theorem \ref{thm:1}, the argument mainly focus on the convergent
subsequence.

Without loss of generality, assume $H\sigma_{\text{I},1\Gamma}^{2}\ge G\sigma_{\text{II},1\Gamma}^{2}$.
When $H\sigma_{\text{I},1\Gamma}^{2}\to\infty$, $r_{GH}=O\left(G/\sigma_{\text{I},1\Gamma}^{2}\right)$
and the intersection terms $\widehat{\Omega}_{\mathrm{III,IV}}$ is
negligible. It suffices to show 
\[
r_{GH}\left(\widehat{\Omega}_{\mathrm{I}}-\widetilde{\Omega}_{\mathrm{I}}\right)=o_{P}\left(1\right).
\]
The proof for $r_{GH}\left(\widehat{\Omega}_{\mathrm{II}}-\widetilde{\Omega}_{\mathrm{II}}\right)=o_{P}\left(1\right)$
follows similarly. By Lemma \ref{lemma: order of oracle variance},
one can always standardize through multiplying $\widehat{\Omega}_{\text{I}}$
and $\widetilde{\Omega}_{\text{I}}$ by $\sigma_{\text{I},1\Gamma}^{-2}$.
It is equivalent to show that when $\sigma_{\text{I},1\Gamma}^{2}=1$ and $r_{GH}=G$,
\begin{equation}
G\bigl[\widehat{\Omega}_{\text{I}}-\widetilde{\Omega}_{\text{I}}\bigr]=o_{P}(1).\label{eq:Omega1-plugin-goal}
\end{equation}
By the corresponding expression for $\widetilde{\Omega}_{\text{I}}$,
we can write 
\begin{align*}
G\bigl[\widehat{\Omega}_{\text{I}}-\widetilde{\Omega}_{\text{I}}\bigr]= & \frac{1}{GH^{2}}\sum_{g=1}^{G}\sum_{h=1}^{H}\sum_{h'\neq h}^{H}\Bigl\{\widehat{\Psi}_{gh}\widehat{\Psi}_{gh'}-\Psi_{gh}\Psi_{gh'}\Bigr\}\\
= & \frac{1}{GH^{2}}\sum_{g=1}^{G}\sum_{h=1}^{H}\sum_{h'\neq h}^{H}\left(\mathbf{1}\left\{ e_{gh}\le G^{-1/2}X_{gh}t,e_{gh'}\le G^{-1/2}X_{gh'}t\right\} -\mathbf{1}\left\{ e_{gh}\le0,e_{gh'}\le0\right\} \right)X_{gh}X_{gh'}\\
 & -\frac{1}{GH^{2}}\sum_{g=1}^{G}\sum_{h=1}^{H}\sum_{h'\neq h}^{H}\left(\tau\cdot\mathbf{1}\left\{ 0<e_{gh}\le G^{-1/2}X_{gh}t\right\} \right)X_{gh}X_{gh'}\\
 & -\frac{1}{GH^{2}}\sum_{g=1}^{G}\sum_{h=1}^{H}\sum_{h'\neq h}^{H}\left(\tau\cdot\mathbf{1}\left\{ 0<e_{gh'}\le G^{-1/2}X_{gh}t\right\} \right)X_{gh}X_{gh'}\\
:= & B_{1,GH}(t)+B_{2,GH}(t)+B_{3,GH}(t).
\end{align*}
We focus mainly on the first term $B_{1,GH}(t)$. Define the centered
 difference 
\[
D_{ghh'}(t):=\Big(\mathbf{1}\{e_{gh}\le G^{-1/2}X_{gh}t,e_{gh'}\le G^{-1/2}X_{gh'}t\}-\mathbf{1}\{e_{gh}\le0,e_{gh'}\le0\}\Big)-p_{ghh'}(t),
\]
where the conditional expectation is $p_{ghh'}(t):=F\!\Big(G^{-1/2}X_{gh}t,G^{-1/2}X_{gh'}t\vert X_{gh},X_{gh'},\{V_{h}\}\Big)-F\!\Big(0,0\vert X_{gh},X_{gh'},\{V_{h}\}\Big)$. Provided that $X_{gh}t$ can be negative, the two indicator sets need not be nested and the conditional expectation can be negative. However, this causes no
difficulty, since we only need a second-moment bound. Indeed,
\( | \mathbf 1\{e_{gh}\le G^{-1/2}X_{gh}t,e_{gh'}\le G^{-1/2}X_{gh'}t\} - \mathbf 1\{e_{gh}\le0,e_{gh'}\le0\} | \le \mathbf 1\{|e_{gh}|\le |G^{-1/2}X_{gh}t|\} + \mathbf 1\{|e_{gh'}|\le |G^{-1/2}X_{gh'}t|\}. \)
Thus the signs of \(X_{gh}t\) and \(X_{gh'}t\) are immaterial for the
following bound. For notational simplicity, the subsequent display is
written as if \(X_{gh}t\) and \(X_{gh'}t\) are positive; the general
case is covered by replacing them with their absolute values.
Then $E\big(D_{ghh'}(t)\vert X_{gh},X_{gh'},\{V_{h}\}\big)=0$ and
\begin{align*}
B_{1,GH}(t) & =\underbrace{\frac{1}{GH^{2}}\sum_{g,h\neq h'}D_{ghh'}(t)X_{gh}X_{gh'}}_{=:\mathcal{T}_{1,GH}(t)}+\underbrace{\frac{1}{GH^{2}}\sum_{g,h\neq h}\Big(p_{ghh'}(t)X_{gh}X_{gh'}-E[p_{ghh'}(t)X_{gh}X_{gh'}\vert\{V_{h}\}]\Big)}_{=:\mathcal{T}_{2,GH}(t)}\\
 & +\underbrace{\frac{1}{GH^{2}}\sum_{g,h\neq h}\Big(E[p_{ghh'}(t)X_{gh}X_{gh'}\vert\{V_{h}\}]-E[p_{ghh'}(t)X_{gh}X_{gh'}]\Big)}_{=:\mathcal{T}_{3,GH}(t)}+\frac{H-1}{H}\underbrace{E\big[p_{ghh'}(t)X_{gh}X_{gh'}\big]}_{=:\mathcal{T}_{4,GH}(t)}.
\end{align*}

\paragraph{Step 1. $\sup_{\left|t\right|\le C_{0}}\left|\mathcal{T}_{1,GH}(t)\right|=O_{P}\left(G^{-3/4}\log G\right).$}

We now bound these four terms one by one, with uniformity in $|t|\le C_{0}$.
We begin with $\sup_{\left|t\right|\le C_{0}}\left|\mathcal{T}_{1,GH}(t)\right|$.
Partition the parameter space of $\left\{ \left\Vert t\right\Vert \in\mathbb{R}^{d}:\left\Vert t\right\Vert \leq C_{0}\right\} $
into $N=\left(G^{1/4}\right)^{d}$ cubes $\left\{ E_{i}\right\} _{i=1}^{N}$
with the side length at most $b_{G}=G^{-1/4}$ (The dimension $d$
only matters here and hence we keep it). Let $t_{i}$ be the a corner
or smallest value in cube $E_{i}$. By construction, for any $t$,
we can find $E_{j}$ such that $t\in E_{j}$. By triangular inequality,
\begin{align*}
\max_{\left|t\right|\leq C_{0}}\left|\mathcal{T}_{1,GH}(t)\right| & \le\max_{i\le N}\left|\mathcal{T}_{1,GH}(t_{i})\right|+\sup_{t\in E_{j}}\left|\mathcal{T}_{1,GH}(t)-\mathcal{T}_{1,GH}(t_{j})\right|.
\end{align*}

For term $\max_{i\le N}\left|\mathcal{T}_{1,GH}(t_{i})\right|$, define
the centered kernel representation 
\[
\bar{A}_{g}(t)=\frac{1}{H^{2}}\sum_{h=1}^{H}\sum_{h'\neq h}^{H}A_{ghh'}(t),\qquad A_{ghh'}(t):=D_{ghh'}(t)\,X_{gh}X_{gh'},
\]

\emph{Stage I: Bound $\max_{i\le N}\left|\mathcal{T}_{1,GH}(t_{i})\right|$
by Bernstein's inequality.} Notice that $\mathcal{T}_{1,GH}(t_{i})=\frac{1}{G}\sum_{g=1}^{G}\bar{A}_{g}(t_{i})$
and given $\{V_{h}\}$, $\bar{A}_{g}(t_{i})$ is independent over
$g$. We seek to apply Bernstein's inequality which requires the bounds
of $\mathrm{Var}\bigl(\bar{A}_{g}(t_{i})\vert\{V_{h}\}\bigr)$ and
$\max_{g\le G}\left|\bar{A}_{g}(t_{i})\right|$.

\medskip{}
\emph{(a) Bounding $\mathrm{Var}\bigl(\bar{A}_{g}(t_{i})\vert\{V_{h}\}\bigr)$.}
Expand the conditional variance 
\[
\mathrm{Var}\bigl(\bar{A}_{g}(t_{i})\vert\{V_{h}\}\bigr)=\frac{1}{H^{4}}\sum_{h\neq h'}\sum_{k\neq k'}E\!\left[A_{ghh'}(t_{i})A_{gkk'}(t_{i})\vert\{V_{h}\}\right].
\]

Each term can be bounded by Cauchy--Schwarz inequality: 
\[
\left|E\!\left[A_{ghh'}(t_{i})A_{gkk'}(t_{i})\vert\{V_{h}\}\right]\right|\le E\!\left[A_{ghh'}(t_{i})^{2}\vert\{V_{h}\}\right]^{1/2}E\!\left[A_{gkk'}(t_{i})^{2}\vert\{V_{h}\}\right]^{1/2}.
\]
Hence $\mathrm{Var}\bigl(\bar{A}_{g}(t_{i})\vert\{V_{h}\}\bigr)\le\sup_{h\neq h'}E\!\left[A_{ghh'}(t_{i})^{2}\vert\{V_{h}\}\right].$
So the entire problem reduces to bounding the second moment of a single
kernel $A_{ghh'}(t_{i})$.

Since $A_{ghh'}(t_{i})=D_{ghh'}(t_{i})X_{gh}X_{gh'}$, we have $A_{ghh'}(t_{i})^{2}=X_{gh}^{2}X_{gh'}^{2}\,D_{ghh'}(t_{i})^{2}$
and thus 
\[
E\!\left[A_{ghh'}(t)^{2}\vert\{V_{h}\}\right]=E\!\left[X_{gh}^{2}X_{gh'}^{2}\,E\left(D_{ghh'}(t)^{2}\vert X_{gh},X_{gh'},\{V_{h}\}\right)\,\Bigm|\{V_{h}\}\right].
\]
Now when $X_{gh}t>0$ and $X_{gh'}t>0$, $D_{ghh'}(t_{i})$ is a centered Bernoulli difference (indicator
minus its conditional mean), and hence it satisfies $E\left(D_{ghh'}(t_{i})^{2}\vert X_{gh},X_{gh'},\{V_{h}\}\right)\le p_{ghh'}(t_{i})$.
Here, $p_{ghh'}(t_{i})$ is of order of the probability mass swept
by moving the thresholds from $(0,0)$ to $(r_{GH}^{-1/2}X_{gh}t_{i},r_{GH}^{-1/2}X_{gh'}t_{i})$.
By the mean value theorem applied to the conditional bivariate CDF,
we obtain 
\begin{align*}
p_{ghh'}\left(t_{i}\right)= & \int_{-\infty}^{G^{-1/2}X_{gh}^{\top}t_{i}}\int_{-\infty}^{G^{-1/2}X_{gh'}t_{i}}f\left(u,v\vert X_{gh},X_{gh'},\left\{ V_{h}\right\} \right)dudv-\int_{-\infty}^{0}\int_{-\infty}^{0}f\left(u,v\vert X_{gh},X_{gh'},\left\{ V_{h}\right\} \right)dudv\\
= & G^{-1/2}t_{i}\Biggl(\int_{-\infty}^{G^{-1/2}X_{gh}\bar{t}}f\left(u,r_{GH}^{-1/2}X_{gh'}^{\top}\bar{t}\vert X_{gh},X_{gh'},\left\{ V_{h}\right\} \right)du\cdot X_{gh'}\\
 & +\int_{-\infty}^{G^{-1/2}X_{gh'}\bar{t}}f\left(r_{GH}^{-1/2}X_{gh}^{\top}\bar{t},v\vert X_{gh},X_{gh'},\left\{ V_{h}\right\} \right)dv\cdot X_{gh}\Biggl),
\end{align*}
Applying Fubini's theorem, together with the uniform boundedness of the
conditional marginal densities,
the two integrals above are uniformly bounded. Hence, for some constant
\(c>0\), independent of \(g,h,h'\), \(G,H\), \(t_i\), and the conditioning
variables,
$p_{ghh'}\left(t_{i}\right)\le c G^{-1/2}\left|t_{i}\right|\left(\left|X_{gh}\right|+\left|X_{gh'}\right|\right)$.
Hence, by the property of Bernoulli random variable, we have 
\(
E\!\left(D_{ghh'}(t_{i})^{2}\vert X_{gh},X_{gh'},\{V_{h}\}\right)\lesssim G^{-1/2}C_{0}\bigl(|X_{gh}|+|X_{gh'}|\bigr).
\)
Plugging this back to the conditional variance yields 
\[
\mathrm{Var}\bigl(\bar{A}_{g}(t_{i})\vert\{V_{h}\}\bigr)\lesssim G^{-1/2}\,\sup_{h\neq h'}E\!\left[|X_{gh}|^{2}|X_{gh'}|^{3}\vert\{V_{h}\}\right].
\]
Under the maintained conditional moment assumption (Assumption \ref{as:moment-strong}(i)), we can deduce that $\sup_{h\neq h'}E(|X_{gh}|^{2}|X_{gh'}|^{3}\vert\{V_{h}\})<\infty$
a.s. (this term is identical over $g$), this becomes $\mathrm{Var}(\bar{A}_{g}(t_{i})\vert\{V_{h}\})\lesssim G^{-1/2}$
a.s.

\medskip{}
\emph{(b) Uniform maximal bound.} We now prove a bound on $\max_{g\le G}|\bar{A}_{g}(t_{i})|$
at a fixed grid point $t_{i}$. Since $|D_{ghh'}(t_{i})|\le1$, given
$\max_{g,h}|X_{gh}|\le CG^{1/8}$ a.s., we have 
\[
\left|\max_{g}\bar{A}_{g}(t_{i})\right|\le\max_{g,h,h'}|A_{ghh'}(t_{i})|\le\max_{g,h,h'}|X_{gh}X_{gh'}|\le\left(\max_{g,h}|X_{gh}|\right)^{2}\le G^{1/4}.
\]
holds a.s. This provides the required almost-sure maximal bound at
each grid point $t_{i}$.

\medskip{}
\emph{(c) Bernstein's inequality.} With these two ingredients we apply
Bernstein's inequality conditionally on $\{V_{h}\}$ with the threshold
$\varepsilon_{G}=c_{1}G^{-3/4}{\log G}$: 
\begin{align*}
P\!\left(\left|\frac{1}{G}\sum_{g=1}^{G}\bar{A}_{g}(t_{i})\right|\ge\varepsilon_{G}\Bigm|\{V_{h}\}\right) & \le\exp\!\left(-\frac{\varepsilon_{G}^{2}/2}{\sum_{g}\mathrm{Var}(\frac{1}{G}\bar{A}_{g}(t_{i})\vert\{V_{h}\})+\max_{g}\frac{1}{G}\bar{A}_{g}(t_{i})\varepsilon_{G}/3}\right).\\
 & =\exp\!\left(-\frac{c_{2}G^{-3/2}(\log G)^2}{c_{3}G^{-3/2}+c_{4}G^{-1}G^{1/4}G^{-3/4}{\log G}}\right)\\
 & =\exp\left(-c_{5}{\log G}\right)=G^{-C},
\end{align*}
Now apply the union bound over the $N$ grid points. Since $N\asymp(G^{1/4})^{d}$,
we have 
\[
P\!\left(\max_{1\le i\le N}\left|\frac{1}{G}\sum_{g=1}^{G}\bar{A}_{g}(t_{i})\right|\ge\varepsilon_{G}\Bigm|\{V_{h}\}\right)\le2N\,G^{-C}=O(G^{d/4-C}).
\]
One can set $c_{1}$ such that $C>d/4$ and
$G^{d/4-C}=o\left(1\right)$ as $G\to\infty$. By law of total probability,
\[
P\!\left(\max_{1\le i\le N}\left|\mathcal{T}_{1,GH}(t_{i})\right|\ge\varepsilon_{G}\right)=E\left[P\!\left(\max_{1\le i\le N}\left|\frac{1}{G}\sum_{g=1}^{G}\bar{A}_{g}(t_{i})\right|\ge\varepsilon_{G}\Bigm|\{V_{h}\}\right)\right]=o(1).
\]
Thus the grid term satisfies 
\[
\max_{1\le i\le N}\left|\mathcal{T}_{1,GH}(t_{i})\right|=O_{P}\!\left(G^{-3/4}{\log G}\right).
\]

\medskip{}
\emph{Stage II: within-cube oscillation bound $\sup_{t\in E_{j}}\left|\mathcal{T}_{1,GH}(t)-\mathcal{T}_{1,GH}(t_{i})\right|$.}
Recall that $t\in E_{j}$ and hence by construction $\left|t-t_{j}\right|\le b_{G}:=G^{-1/4}$.
Define the bracket increment 
\[
\Delta_{ghh'}(t,t'):=\mathbf{1}\{e_{gh}\le G^{-1/2}X_{gh}t,e_{gh'}\le G^{-1/2}X_{gh'}t\}-\mathbf{1}\{e_{gh}\le G^{-1/2}X_{gh}t',e_{gh'}\le G^{-1/2}X_{gh'}t'\},
\]
so that $\mathcal{T}_{1,GH}(t)-\mathcal{T}_{1,GH}(t_{j})=\frac{1}{G}\sum_{g=1}^{G}\frac{1}{H^{2}}\sum_{h=1}^{H}\sum_{h'\neq h}^{H}\left(\Delta_{ghh'}(t,t_{j})-\widetilde{\Delta}_{ghh'}(t,t_{j})\right)$,
where $\widetilde{\Delta}_{ghh'}(t,t')=p_{ghh'}(t)-p_{ghh'}(t')$
is the corresponding difference of conditional CDF increments. Given
that the indicator function and cdf are monotone increasing, $\left|\Delta_{ghh'}(t,t_{j})\right|\le\Delta_{ghh'}(t_{j}+b_{G},t_{j}-b_{G})$
and $\left|\widetilde{\Delta}_{ghh'}(t,t_{j})\right|\le\widetilde{\Delta}_{ghh'}(t_{j}+b_{G},t_{j}-b_{G})$.
Hence, by triangular inequality, 
\begin{align*}
\sup_{t\in E_{j}}\left|\mathcal{T}_{1,GH}(t)-\mathcal{T}_{1,GH}(t_{j})\right|\le & \left|\frac{1}{G}\sum_{g=1}^{G}\frac{1}{H^{2}}\sum_{h=1}^{H}\sum_{h'\neq h}^{H}\left(\Delta_{ghh'}(t_{j}+b_{G},t_{j}-b_{G})-\widetilde{\Delta}_{ghh'}(t_{j}+b_{G},t_{j}-b_{G})\right)\right|\\
 & +2\left|\frac{1}{G}\sum_{g=1}^{G}\frac{1}{H^{2}}\sum_{h=1}^{H}\sum_{h'\neq h}^{H}\widetilde{\Delta}_{ghh'}(t_{j}+b_{G},t_{j}-b_{G})\right|
\end{align*}
For the first term, applying conditional Bernstein with the union
bound, as did in Stage I, yields 
\[
\left|\frac{1}{G}\sum_{g=1}^{G}\frac{1}{H^{2}}\sum_{h=1}^{H}\sum_{h'\neq h}^{H}\left(\Delta_{ghh'}(t_{j}+b_{G},t_{j}-b_{G})-\widetilde{\Delta}_{ghh'}(t_{j}+b_{G},t_{j}-b_{G})\right)\right|=O_{P}(G^{-3/4}{\log G}).
\]

For the second term, by mean value theorem, 
\begin{align*}
\widetilde{\Delta}_{ghh'}(t_{j}+b_{G},t_{j}-b_{G})= & \int_{-\infty}^{G^{-1/2}X_{gh}^{\top}\left(t_{j}+b_{G}\right)}\int_{-\infty}^{G^{-1/2}X_{gh'}\left(t_{j}+b_{G}\right)}f\left(u,v\vert X_{gh},X_{gh'},\left\{ V_{h}\right\} \right)dudv\\
 & -\int_{-\infty}^{G^{-1/2}X_{gh}^{\top}\left(t_{j}-b_{G}\right)}\int_{-\infty}^{G^{-1/2}X_{gh'}\left(t_{j}-b_{G}\right)}f\left(u,v\vert X_{gh},X_{gh'},\left\{ V_{h}\right\} \right)dudv\\
= & G^{-1/2}2b_{G}\Biggl(\int_{-\infty}^{G^{-1/2}X_{gh}\bar{t}}f\left(u,r_{GH}^{-1/2}X_{gh'}^{\top}\bar{t}\vert X_{gh},X_{gh'},\left\{ V_{h}\right\} \right)du\cdot X_{gh'}\\
 & +\int_{-\infty}^{G^{-1/2}X_{gh'}\bar{t}}f\left(r_{GH}^{-1/2}X_{gh}^{\top}\bar{t},v\vert X_{gh},X_{gh'},\left\{ V_{h}\right\} \right)dv\cdot X_{gh}\Biggl).\\
\lesssim & G^{-3/4}.
\end{align*}
Putting Stages I and II together gives the uniform-in-$t$ concentration
\[
\sup_{\|t\|\le C_{0}}\left|\mathcal{T}_{1,GH}(t)\right|\le\max_{1\le i\le N}\left|\mathcal{T}_{1,GH}(t_{i})\right|+\sup_{t\in E_{j}}\left|\mathcal{T}_{1,GH}(t)-\mathcal{T}_{1,GH}(t_{j})\right|=O_{P}(G^{-3/4}{\log G}).
\]

\paragraph{Step 2. $\sup_{\left|t\right|\le C_{0}}\left|\mathcal{T}_{2,GH}(t)\right|=O_{P}\left(G^{-1}\log G\right).$}

We now return to $\mathcal{T}_{2,GH}(t)$ and fix $t$. Let $A_{2,g}\left(t\right)=\frac{1}{H^{2}}\sum_{h=1}^{H}\sum_{h'\neq h}^{H}\Big(p_{ghh'}(t)X_{gh}X_{gh'}-E[p_{ghh'}(t)X_{gh}X_{gh'}\vert\{V_{h}\}]\Big)$
and hence $\mathcal{T}_{2,GH}(t)=\frac{1}{G}\sum_{g=1}^{G}A_{2,g}\left(t\right)$.
Observe that $E\left(A_{2,g}\left(t\right)\vert\left\{ V_{h}\right\} \right)=0$,
and conditional on $\left\{ V_{h}\right\} $, $A_{2,g}\left(t\right)$
is independent over $g$. Moreover, the conditional variance is 
\[
Var\left(A_{2,g}\left(t\right)\vert\left\{ V_{h}\right\} \right)\le\sup_{h,h'}E\left(p_{ghh'}(t)^{2}X_{gh}^{2}X_{gh'}^{2}\vert\left\{ V_{h}\right\} \right)
\]
Recall that $p_{ghh'}\left(t\right)\lesssim G^{-1/2}\left|t\right|\left(\left|X_{gh}\right|+\left|X_{gh'}\right|\right)$, so given Assumption \ref{as:moment-strong}(i), 
we have $$\sup_{h\neq h'}E\left(p_{ghh'}(t)^{2}X_{gh}^{2}X_{gh'}^{2}\vert\left\{ V_{h}\right\} \right)\lesssim G^{-1}\sup_{h\neq h'}E\left(X_{gh}^{4}X_{gh'}^{2}\right)\lesssim G^{-1}.$$
Then, by the Bernstein inequality and union bound, as in Step 1, one
can extend the result to uniformly $\left|t\right|\le C_{0}$ and
obtain the desired results. Proofs are close to those in Step 1, so
will be omitted.

\paragraph{Step 3. $\sup_{\left|t\right|\le C_{0}}\left|\mathcal{T}_{3,GH}(t)\right|=O_{P}\left(\left(GH\right)^{-1/2}\log G\right).$ }

Fix $t$. Let 
\[
\mathcal{T}_{3,GH}(t)=\frac{1}{H^{2}}\sum_{h=1}^{H}\sum_{h'\neq h}^{H}A_{3,hh'}\left(t\right),
\]
where 
\begin{align*}
A_{3,hh'}\left(t\right) & =\frac{1}{G}\sum_{g=1}^{G}\Big(E[p_{ghh'}(t)X_{gh}X_{gh'}\vert\{V_{h}\}]-E[p_{ghh'}(t)X_{gh}X_{gh'}]\Big)\\
 & =\frac{1}{G}\sum_{g=1}^{G}\Big(E[p_{ghh'}(t)X_{gh}X_{gh'}\vert V_{h},V_{h'}]-E[p_{ghh'}(t)X_{gh}X_{gh'}]\Big).
\end{align*}
Notice that $E\left(A_{3,hh'}\left(t\right)\right)=0$ and $A_{3,hh'}\left(t\right)$
is a U-process based on $\left(V_{h},V_{h'}\right)$. Hence, $\mathcal{T}_{3,GH}(t)$
is a function of these $H$ coordinates: 
\[
\mathcal{T}_{3,GH}(t)=\phi\bigl(V_{1},\ldots,V_{H}).
\]
We verify the bounded difference property. Fix an index $h_{0}$ and
replace only the $h_{0}$-th coordinate $V_{h_{0}}$ by an independent
copy $V'_{h_{0}}$, leaving all other coordinates unchanged. Only
those summands $A_{3,hh'}(t)$ involving $h_{0}$ can change. These
are exactly: (i) terms with $h=h_{0}$ and $h'\neq h_{0}$ (there
are $H-1$ of them), (ii) terms with $h\neq h_{0}$ and $h'=h_{0}$
(there are $H-1$ of them), and (iii) the overlap adjustment does
not introduce any extra terms since $h'\neq h$. Therefore, at most
$2H-2$ summands change. We apply McDiarmid's inequality to \(\mathcal T_{3,GH}(t)\), viewed as a
function of \((V_1,\ldots,V_H)\). This requires a uniform bounded-difference
bound. By Assumption \ref{as:moment-strong}(i), the conditional Lipschitz
property of the distribution function, and conditional H{\"o}lder's inequality,
each summand (the conditional expectation) satisfies
\[
\begin{aligned}
&
\left|
E[p_{ghh'}(t)X_{gh}X_{gh'}\mid V_h=v,V_{h'}=v']
\right| \\
&\qquad \le
C G^{-1/2}
E\!\left[
(|X_{gh}|+|X_{gh'}|)
|X_{gh}||X_{gh'}|
\mid V_h=v,V_{h'}=v'
\right]\\
&\qquad \le \left(\sup_{u,v}E(|X_{gh}|^6\mid U_g=u,V_h=v)\right)^{1/3}\left(\sup_{u,v}E(|X_{gh}|^6\mid U_g=u,V_h=v)\right)^{1/6}<\infty.
\end{aligned}
\]
uniformly in \(g,h,h'\) and \(|t|\le C_0\). 
 The size of variation of $\mathcal{T}_{3,GH}(t)$
after substituting the value of $V_{h_{0}}$, $\Delta_{h_{0}}=O\left(G^{-1/2}\frac{2H-2}{H^{2}}\right)=O\left(G^{-1/2}H^{-1}\right)$.
McDiarmid's inequality then yields, for any $\varepsilon_{GH}=C\left(GH\right)^{-1/2}\log G$,
\[
P\!\left(|\mathcal{T}_{3,GH}(t)|\ge\varepsilon_{GH}\right)\le2\exp\!\left(-\frac{2\varepsilon_{GH}^{2}}{\sum_{h_{0}=1}^{H}\Delta_{h_{0}}^{2}}\right)\le G^{-C}.
\]
The result is then extended to the uniformity result with any $\left|t\right|\le C_{0}$,
by McDiarmid's inequality and union bound as before.

\paragraph{Step 4. Decompose $\mathcal{T}_{4,GH}(t).$ }

Recall that $$\mathcal{T}_{4,GH}(t)=E\left[\left(F\left(G^{-1/2}X_{gh}^{\top}t,G^{-1/2}X_{gh'}^{\top}t\vert X_{gh},X_{gh'}\right)-F\left(0,0\vert X_{gh},X_{gh'}\right)\right)X_{gh}X_{gh'}\right].$$
By the Taylor expansion 
\begin{align*}
 & F\left(G^{-1/2}X_{gh}^{\top}t,G^{-1/2}X_{gh'}^{\top}t\vert X_{gh},X_{gh'}\right)X_{gh}X_{gh'}-F\left(0,0\vert X_{gh},X_{gh'}\right)X_{gh}X_{gh'}\\
= & \frac{\partial}{\partial t}F\left(G^{-1/2}X_{gh}^{\top}t,G^{-1/2}X_{gh'}^{\top}t\vert X_{gh},X_{gh'}\right)\vert_{t=0}X_{gh}X_{gh'}t\\
 & +\frac{1}{2}\frac{\partial^{2}}{\partial t^{2}}F\left(G^{-1/2}X_{gh}^{\top}t,G^{-1/2}X_{gh'}^{\top}t\vert X_{gh},X_{gh'}\right)X_{gh}X_{gh'}t^{2}\\
 & +\frac{1}{6}\frac{\partial^{3}}{\partial t^{3}}F\left(G^{-1/2}X_{gh}^{\top}\bar{t},G^{-1/2}X_{gh'}^{\top}\bar{t}\vert X_{gh},X_{gh'}\right)X_{gh}X_{gh'}t^{3},
\end{align*}
for some $\bar{t}$ with $\left|\bar{t}\right|\le\left|t\right|$.
Here, by the Leibniz rule, 
\begin{align*}
\frac{\partial}{\partial t}F\left(G^{-1/2}X_{gh}^{\top}t,G^{-1/2}X_{gh'}^{\top}t\vert X_{gh},X_{gh'}\right)\vert_{t=0}X_{gh}X_{gh'}= & \int_{-\infty}^{0}f\left(e_{gh},0\vert X_{gh},X_{gh'}\right)de_{gh}\cdot G^{-1/2}X_{gh}X_{gh'}^{2}\\
 & +\int_{-\infty}^{0}f\left(0,e_{gh'}\vert X_{gh},X_{gh'}\right)de_{gh'}\cdot G^{-1/2}X_{gh}^{2}X_{gh'}\\
:= & G^{-1/2}\mathcal{I}_{1,ghh'},
\end{align*}
\begin{align*}
\frac{\partial^{2}}{\partial t^{2}}F\left(G^{-1/2}X_{gh}^{\top}t,G^{-1/2}X_{gh'}^{\top}t\vert X_{gh},X_{gh'}\right)\vert_{t=0}= & 2f\left(0,0\vert X_{gh},X_{gh'}\right)\cdot G^{-1}X_{gh}^{2}X_{gh'}^{2}\\
 & +\int_{-\infty}^{0}f^{\left(0,1\right)}\left(e_{gh},0\vert X_{gh},X_{gh'}\right)de_{gh}\cdot G^{-1}X_{gh}X_{gh'}^{3}\\
 & +\int_{-\infty}^{0}f^{\left(1,0\right)}\left(0,e_{gh'}\vert X_{gh},X_{gh'}\right)de_{gh'}\cdot G^{-1}X_{gh}^{3}X_{gh'}\\
:= & G^{-1}\mathcal{I}_{2,ggh'},
\end{align*}
and 
\begin{align*}
 & \frac{\partial^{3}}{\partial t^{3}}F\left(G^{-1/2}X_{gh}^{\top}\bar{t},G^{-1/2}X_{gh'}^{\top}\bar{t}\vert X_{gh},X_{gh'}\right)\vert_{t=0}\\
= & 3f^{\left(1,0\right)}\left(G^{-1/2}X_{gh}^{\top}\bar{t},G^{-1/2}X_{gh'}^{\top}\bar{t}\vert X_{gh},X_{gh'}\right)\cdot G^{-3/2}X_{gh}^{3}X_{gh'}^{2}\\
 & +3f^{\left(0,1\right)}\left(G^{-1/2}X_{gh}^{\top}\bar{t},G^{-1/2}X_{gh'}^{\top}\bar{t}\vert X_{gh},X_{gh'}\right)\cdot G^{-3/2}X_{gh}^{2}X_{gh'}^{3}\\
 & +\int_{-\infty}^{G^{-1/2}X_{gh}^{\top}\bar{t}}f^{\left(0,2\right)}\left(e_{gh},G^{-1/2}X_{gh'}^{\top}\bar{t}\vert X_{gh},X_{gh'}\right)de_{gh}\cdot G^{-3/2}X_{gh}X_{gh'}^{4}\\
 & +\int_{-\infty}^{G^{-1/2}X_{gh'}^{\top}\bar{t}}f^{\left(2,0\right)}\left(G^{-1/2}X_{gh}^{\top}\bar{t},e_{gh'}\vert X_{gh},X_{gh'}\right)de_{gh'}\cdot G^{-3/2}X_{gh}^{4}X_{gh'}\\
:= & G^{-3/2}\mathcal{I}_{3,ggh'}.
\end{align*}
By Fubini's theorem, we have $E\left(\mathcal{I}_{1,ghh'}\right)\lesssim E\left(X_{gh}^{2}X_{gh'}\right)<\infty$.
Likewise, one can show that $E\left(\mathcal{I}_{2,ghh'}\right)<\infty$
and $E\left(\mathcal{I}_{2,ghh'}\right)<\infty$. Collecting terms
and plugging back $G^{-1/2}t=\widehat{\beta}-\beta_{0}\left(\tau\right)=D\left(\tau\right)^{-1}\frac{1}{GH}\sum_{g=1}^{G}\sum_{h=1}^{H}\Psi_{gh}+o_{P}\left(\widehat{\beta}-\beta_{0}\left(\tau\right)\right)$
yields that 
\begin{align*}
\mathcal{T}_{4,GH}(t) & =G^{-1/2}E\left(\mathcal{I}_{1,ghh'}\right)W_{GH}+o_{P}\left(G^{-1/2}\right),
\end{align*}
where $W_{GH}=D\left(\tau\right)^{-1}\frac{G^{1/2}}{GH}\sum_{g=1}^{G}\sum_{h=1}^{H}\Psi_{gh}.$

\paragraph{Step 5. }

For $B_{2,GH}\left(t\right)$, define the centered Bernoulli difference
\[
D_{2,gh}(t):=\mathbf{1}\{e_{gh}\le G^{-1/2}X_{gh}t\}-p_{2,gh}(t),
\]
where the conditional success probability is $p_{2,gh}(t):=F\!\Big(G^{-1/2}X_{gh}t\vert X_{gh},\{V_{h}\}\Big)-F\!\Big(0\vert X_{gh},\{V_{h}\}\Big)$.
Then $E\big(D_{2,gh}(t)\vert X_{gh},\{V_{h}\}\big)=0$ and 
\begin{align*}
B_{2,GH}(t) & =\frac{1}{GH^{2}}\sum_{g,h\neq h'}D_{2,gh}(t)X_{gh}X_{gh'}+\frac{1}{GH^{2}}\sum_{g,h\neq h}\Big(p_{2,gh}(t)X_{gh}X_{gh'}-E[p_{2,gh}(t)X_{gh}X_{gh'}\vert\{V_{h}\}]\Big)\\
 & +\frac{1}{GH^{2}}\sum_{g,h\neq h}\Big(E[p_{2,gh}(t)X_{gh}X_{gh'}\vert\{V_{h}\}]-E[p_{2,gh}(t)X_{gh}X_{gh'}]\Big)+\frac{H-1}{H}E\big[p_{2,gh}(t)X_{gh}X_{gh'}\big].
\end{align*}
By an argument analogue to Steps 1-4, one can deduce that $B_{2,GH}(t)$
and $B_{3,GH}(t)$ are of order $O_{P}\!\left(G^{-3/4}{\log G}\right).$
Combining results above, we have establishes that 
\[
G\bigl[\widehat{\Omega}_{\text{I}}-\widetilde{\Omega}_{\text{I}}\bigr]=G^{-1/2}E\left(\mathcal{I}_{1,ghh'}\right)W_{GH}+R_{1,GH}+o_{P}\left(G^{-1/2}\right),
\]
where $R_{1,GH}=O_{P}\!\left(G^{-3/4}{\log G}+\left(GH\right)^{-1/2}\sqrt{\log G}\right)$
. Hence, $G\bigl[\widehat{\Omega}_{\text{I}}-\widetilde{\Omega}_{\mathrm{\text{I}}}\bigr]=o_{P}\left(1\right).$

When $H\sigma_{\text{I},1\Gamma}^{2}=O\left(1\right)$, the intersection
term $\widehat{\Omega}_{\mathrm{III,IV}}$ is no longer negligible
and $r_{GH}\asymp GH.$ The main arguments used to prove
\(r_{GH}(\widehat{\Omega}_{\mathrm{I}}-\widetilde{\Omega}_{\mathrm{I}})=o_P(1)\)
also apply to
\(\widehat{\Omega}_{\mathrm{III,IV}}\). In this case, one works with the
cell-level increment
\[
p_{gh}(t)
=
F\!\left((GH)^{-1/2}X_{gh}'t\mid X_{gh},U_g,V_h\right)
-
F\!\left(0\mid X_{gh},U_g,V_h\right).
\]
Conditional on \(\{U_g\}_{g\le G}\), \(\{V_h\}_{h\le H}\), and
\(\{X_{gh}\}_{g\le G,h\le H}\), the corresponding cell-level centered
terms are independent across \((g,h)\). Therefore the same variance and
concentration arguments yield
\(r_{GH}(\widehat{\Omega}_{\mathrm{III,IV}}-\widetilde{\Omega}_{\mathrm{III,IV}})
=o_P(1)\), and the details are omitted. Finally, applying Lemma
\ref{lemma: order of oracle variance} yields that 
\[
r_{GH}(\widetilde{\Omega}-\Omega_{GH})=o_{P}(1).
\]
The application of Slutsky's Lemma with Theorems \ref{thm:1} and
\ref{thm:Jacobian} implies 
\[
\widehat{\Sigma}^{-1/2}\bigl(\hat{\beta}-\beta_{0}(\tau)\bigr)\ \overset{d}{\to}\ \mathcal{N}\!\left(0,1\right).
\]
Given that the above result holds for any convergent subsequence,
the uniformity result then follows along such subsequences. 
\end{proof}
\clearpage{}

\section{Technical Lemmas}

\begin{lemma}[Approximate score equation]\label{lem:approx_score}
Let
\(
\mathbb{S}(\beta)
=
\frac{1}{GH}\sum_{g=1}^{G}\sum_{h=1}^{H}
X_{gh}\Bigl(\tau-\mathbf{1}\{y_{gh}\le X_{gh}^{\top}\beta\}\Bigr).
\)
Under Assumption of Theorem \ref{thm:1}, 
\(
\|\mathbb{S}(\widehat\beta)\|
=
o_{P}\!\left(r_{GH}^{-1/2}\right).
\)
\end{lemma}

\begin{proof}[Proof of Lemma \ref{lem:approx_score}]
Let \(u_{gh}(\beta)=y_{gh}-X_{gh}^{\top}\beta\). By the subgradient
characterization of the convex quantile-regression objective, there
exist numbers \(a_{gh}\in[0,1]\) such that
\[
0
=
\frac{1}{GH}\sum_{g=1}^{G}\sum_{h=1}^{H}
X_{gh}\Bigl(\tau-\mathbf{1}\{u_{gh}(\widehat\beta)<0\}
-a_{gh}\mathbf{1}\{u_{gh}(\widehat\beta)=0\}\Bigr).
\]
Therefore,
\(
\mathbb{S}(\widehat\beta)
=
\frac{1}{GH}\sum_{g=1}^{G}\sum_{h=1}^{H}
X_{gh}\bigl(a_{gh}-1\bigr)
\mathbf{1}\{u_{gh}(\widehat\beta)=0\}.
\)
Given Assumption \ref{as:hyperplane}, only $o\left((GH)^{1-1/q}r_{GH}^{-1/2}\right)$ observations lie exactly on the same
quantile-regression hyperplane with probability converging to 1. Hence
\[
\|\mathbb{S}(\widehat\beta)\|
\le
\frac{1}{GH}
\sum_{g=1}^{G}\sum_{h=1}^{H}
\|X_{gh}\|\mathbf{1}\{u_{gh}(\widehat\beta)=0\}
=
o_P\left((GH)^{-1/q}r_{GH}^{-1/2}\cdot\max_{g\le G,h\le H}\|X_{gh}\|\right).
\]
Provided that $E\|X_{gh}\|^{q}<\infty$, we have $\max_{g\le G,h\le H}\|X_{gh}\|=O_P((GH)^{1/q})$ by Markov inequality and the stated probability orders follow immediately.
\end{proof}

\begin{lemma}[Local stochastic equicontinuity of $\nu_{GH}(\beta)$]
\label{lem: stochastic equicontinuity}Suppose Assumptions of Theorem
\ref{thm:Jacobian} hold. Let $\widehat{D}(\beta):=\frac{1}{GH\ell}\sum_{g=1}^{G}\sum_{h=1}^{H}K\!\left(\frac{y_{gh}-X_{gh}^{\top}\beta}{\ell}\right)\mathfrak{X}_{gh}$,
where $K\!\left(u\right)=\frac{1}{2}\mathbf{1}\left\{ \left|u\right|\le1\right\} $
and $\mathfrak{X}_{gh}=\mathrm{tr}(BX_{gh}X_{gh}^{\top})$ is a scalar
for an arbitrary deterministic matrix $B\in\mathbb{R}^{d\times d}$,
and $\nu_{GH}\left(\beta\right)=\widehat{D}(\beta)-E\left(\widehat{D}(\beta)\right)$,
then 
\[
\nu_{GH}(\widehat{\beta})-\nu_{GH}\left(\beta_{0}\right)=o_{P}\left(r_{GH}^{-1/2}\ell^{-1/2}\right)
\]
uniformly in $\widehat{\beta}$ satisfying $\left\Vert r_{GH}^{1/2}\left(\widehat{\beta}-\beta_{0}\right)\right\Vert \le C_{0}<\infty.$
\end{lemma} 
\begin{proof}
Rearranging terms, it suffices to show 
\[
\sup_{\|t\|\le C_{0}}\left|\big(\widehat{D}(\beta_{0}+r_{GH}^{-1/2}t)-\widehat{D}(\beta_{0})\big)-E\big[\widehat{D}(\beta_{0}+r_{GH}^{-1/2}t)-\widehat{D}(\beta_{0})\big]\right|=o_{P}(r_{GH}^{-1/2}\ell^{-1/2}).
\]
Since $K$ is the uniform kernel, the difference of kernels becomes
a finite signed sum of indicators. Explicitly, multiplied by $\ell$,
one can write 
\begin{align*}
\ell\Big(\widehat{D}(\beta_{0}+r_{GH}^{-1/2}t)-\widehat{D}(\beta_{0})\Big) & =\frac{1}{2GH}\sum_{g=1}^{G}\sum_{h=1}^{H}\mathfrak{X}_{gh}\left\{ \mathbf{1}\Big(\left|e_{gh}-r_{GH}^{-1/2}X_{gh}^{\top}t\right|\le\ell\Big)-\mathbf{1}\Big(\left|e_{gh}\right|\le\ell\Big)\right\} \\
 & =D_{1,GH}+D_{2,GH}+D_{3,GH}+D_{4,GH},
\end{align*}
where
\begin{align*}
D_{1,GH}\left(t\right) & :=\frac{1}{2GH}\sum_{g,h}\mathfrak{X}_{gh}\mathbf{1}\left\{ \ell<e_{gh}\le\ell+r_{GH}^{-1/2}X_{gh}^{\top}t\right\} ,\\
D_{2,GH}\left(t\right) & :=-\frac{1}{2GH}\sum_{g,h}\mathfrak{X}_{gh}\mathbf{1}\left\{ \ell+r_{GH}^{-1/2}X_{gh}^{\top}t<e_{gh}\leq\ell\right\} ,\\
D_{3,GH}\left(t\right) & :=-\frac{1}{2GH}\sum_{g,h}\mathfrak{X}_{gh}\mathbf{1}\left\{ -\ell\leq e_{gh}<-\ell+r_{GH}^{-1/2}X_{gh}^{\top}t\right\} ,\\
D_{4,GH}\left(t\right) & :=\frac{1}{2GH}\sum_{g,h}\mathfrak{X}_{gh}\mathbf{1}\left\{ -\ell+r_{GH}^{-1/2}X_{gh}^{\top}t\leq e_{gh}<-\ell\right\},
\end{align*}
by treating $X_{gh}^\top t>0$ for simplicity. We now show that each term $D_{\bullet,GH}\left(t\right)-E\left[D_{\bullet,GH}\left(t\right)\right]=o_{P}(r_{GH}^{-1/2}\ell^{1/2})$
uniformly in $\|t\|\le C_{0}$; we treat $D_{1,GH}$, and the others
follow identically. Define $I_{gh}(t):=\tfrac{1}{2}\mathfrak{X}_{gh}\mathbf{1}\{\ell<e_{gh}\le\ell+r_{GH}^{-1/2}X_{gh}^{\top}t\}$.
Then $D_{1,GH}=(GH)^{-1}\sum_{g,h}I_{gh}(t)$.

Apply the two-way Hoeffding/ANOVA decomposition: 
\[
I_{gh}(t)=I_{g\cdot}^{(\mathrm{I})}(t)+I_{\cdot h}^{(\mathrm{II})}(t)+I_{gh}^{(\mathrm{III,IV})}(t)+E\left[I_{gh}(t)\right],
\]
where 
\begin{align*}
I_{g\cdot}^{(\mathrm{I})}(t) & :=E[I_{gh}(t)\vert U_{g}]-E\left[I_{gh}(t)\right],\\
I_{\cdot h}^{(\mathrm{II})}(t) & :=E[I_{gh}(t)\vert V_{h}]-E\left[I_{gh}(t)\right],\\
I_{gh}^{(\mathrm{III,IV})}(t) & :=I_{gh}(t)-E[I_{gh}(t)\vert U_{g}]-E[I_{gh}(t)\vert V_{h}]+E\left[I_{gh}(t)\right].
\end{align*}
Hence, we have 
\begin{equation}
D_{1,GH}\left(t\right)-E\left[D_{1,GH}\left(t\right)\right]=\frac{1}{G}\sum_{g}I_{g\cdot}^{(\mathrm{I})}(t)+\frac{1}{H}\sum_{h}I_{\cdot h}^{(\mathrm{II})}(t)+\frac{1}{GH}\sum_{g,h}I_{gh}^{(\mathrm{III,IV})}(t).\label{eq: Agh}
\end{equation}
Taking variances and using the orthogonality of the projections yields
\[
\mathrm{Var}\!\left(\frac{1}{GH}\sum_{g,h}(I_{gh}(t)-E\left[I_{gh}(t)\right])\right)=\frac{1}{G}Var\left(I_{g\cdot}^{(\mathrm{I})}(t)\right)+\frac{1}{H}Var\left(I_{\cdot h}^{(\mathrm{II})}(t)\right)+\frac{1}{GH}Var\left(I_{gh}^{(\mathrm{III,IV})}(t)\right).
\]
Here, we apply the fact that, by conditioning on $(V_{h},V_{h'})$
for $h\neq h'$, one has 
\begin{align*}
E\left(I_{gh}^{(\mathrm{III,IV})}(t)I_{gh'}^{(\mathrm{III,IV})}(t)\right)=E\left(E\left(I_{gh}^{(\mathrm{III,IV})}(t)I_{gh'}^{(\mathrm{III,IV})}(t)\vert V_{h},V_{h'}\right)\right)\\
=E\left(E\left(I_{gh}^{(\mathrm{III,IV})}(t)\vert V_{h}\right)E\left(I_{gh'}^{(\mathrm{III,IV})}(t)\vert V_{h'}\right)\right)=0,
\end{align*}
and similarly $E\left(I_{gh}^{(\mathrm{III,IV})}(t)I_{g'h}^{(\mathrm{III,IV})}(t)\right)=0$
for $g\neq g'$. Thus 
\begin{equation}
\mathrm{Var}\!\left(\frac{1}{GH}\sum_{g,h}I_{gh}^{(\mathrm{III,IV})}(t)\right)=\frac{1}{(GH)^{2}}\sum_{g,h}E[(I_{gh}^{(\mathrm{III,IV})}(t))^{2}]=\frac{1}{GH}E[(I_{gh}^{(\mathrm{III,IV})}(t))^{2}].\label{eq:Igh variance}
\end{equation}

Next we bound the second moments uniformly over $\|t\|\le C_{0}$.
Fix any $\|t\|\le C_{0}$, 
\begin{align*}
Var\left(I_{g\cdot}^{(\mathrm{I})}(t)\right)= & E\left(E\left[\tfrac{1}{2}\mathfrak{X}_{gh}\left\{ F_{e\vert X,U}\left(\ell+r_{GH}^{-1/2}X_{gh}^{\top}t\right)-F_{e\vert X,U}\left(\ell\right)\right\} \vert U_{g}\right]^{2}\right)\\
 & -E\left(E\left[\tfrac{1}{2}\mathfrak{X}_{gh}\left\{ F_{e\vert X,U}\left(\ell+r_{GH}^{-1/2}X_{gh}^{\top}t\right)-F_{e\vert X,U}\left(\ell\right)\right\} \vert U_{g}\right]\right)^{2}\\
= & E\left(E\left[\tfrac{1}{2}\mathfrak{X}_{gh}\left\{ r_{GH}^{-1/2}X_{gh}^{\top}tf_{e\vert X,U}\left(\ell+r_{GH}^{-1/2}X_{gh}^{\top}\bar{t}\right)\right\} \vert U_{g}\right]^{2}\right)\\
 & -E\left(E\left[\tfrac{1}{2}\mathfrak{X}_{gh}\left\{ r_{GH}^{-1/2}X_{gh}^{\top}tf_{e\vert X,U}\left(\ell+r_{GH}^{-1/2}X_{gh}^{\top}\bar{t}\right)\right\} \vert U_{g}\right]\right)^{2}\\
\lesssim & r_{GH}^{-1}E\left(E\left[\mathfrak{X}_{gh}\left\Vert X_{gh}\right\Vert^2 \vert U_{g}\right]^{2}\right).
\end{align*}
Here, the second equality holds by the mean value theorem and the last 
inequality holds by the uniform bound of $f_{e|X,U}(e\vert X_{gh},U_{g})$
near $e=0$. Given Assumption \eqref{as:bandwidth}(ii) and the right hand side does not depend on $t$, we have 
\[
\sup_tVar\left(I_{g\cdot}^{(\mathrm{I})}(t)\right)\lesssim r_{GH}^{-1},\qquad \sup_tVar\left(I_{\cdot h}^{(\mathrm{II})}(t)\right)\lesssim r_{GH}^{-1}.
\]

For $\sup_t Var\left(I_{gh}^{(\mathrm{III,IV})}(t)\right)$, given $I_{gh}(t)^{2}\le\tfrac{1}{4}\mathfrak{X}_{gh}^{2}\mathbf{1}\{\ell<e_{gh}\le\ell+r_{GH}^{-1/2}\left\Vert X_{gh}\right\Vert C_{0}\}$,
we have

\begin{align}
\sup_{\|t\|\le C_{0}}E\left(I_{gh}(t)^{2}\right)=\frac{1}{4}E\!\left[\mathfrak{X}_{gh}^{2}\int_{\ell}^{\ell+C_{0}r_{GH}^{-1/2}\|X_{gh}\|}f_{e|X}(e\vert X_{gh})\,de\right]\le C\,r_{GH}^{-1/2}\,E\!\left[\mathfrak{X}_{gh}^{2}\,\|X_{gh}\|\right].\label{eq: Agh 2}
\end{align}
By conditional Jensen, the same bound (up to constants) holds for
$\sup_tVar\left(I_{gh}^{(\mathrm{III,IV})}(t)\right)$. Consequently, we have {\small{}
\begin{align}
\sup_{\|t\|\le C_{0}}\mathrm{Var}\!\left(\frac{1}{GH}\sum_{g,h}I_{gh}(t)\right)\le & \sup_{\|t\|\le C_{0}}\frac{1}{G}\mathrm{Var}\!\left(I_{g\cdot}^{(\mathrm{I})}(t)\right)+\sup_{\|t\|\le C_{0}}\frac{1}{H}\mathrm{Var}\!\left(I_{\cdot h}^{(\mathrm{II})}(t)\right)+\sup_{\|t\|\le C_{0}}\frac{1}{GH}\mathrm{Var}\!\left(\sum_{g,h}I_{gh}^{(\mathrm{III,IV})}(t)\right)\nonumber \\
\lesssim & \left(\frac{1}{G}+\frac{1}{H}\right)r_{GH}^{-1}+\frac{1}{GH}r_{GH}^{-1/2}.\label{eq: bound of Var I_gh}
\end{align}
}{\small\par}

To convert this variance control into a uniform stochastic bound,
we use symmetrization. Fix $\varepsilon>0$ and let $\{\eta_{gh}\}$
be Rademacher variables independent of the data and i.i.d. over $g$
and $h$. A standard symmetrization argument (Lemma 2.3.7 of van der
\citet{van1996weak}) yields 
\begin{align}
\theta_{GH}P\!\left(\sup_{\|t\|\le C_{0}}\left|\sum_{g,h}\frac{1}{GH}(I_{gh}(t)-E\left[I_{gh}(t)\right])\right|>r_{GH}^{-1/2}\ell^{1/2}\varepsilon\right)\nonumber \\
\le2\,P\!\Biggl(\sup_{\|t\|\le C_{0}} & \left|\frac{1}{GH}\sum_{g,h}\eta_{gh}I_{gh}(t)\right|>\frac{r_{GH}^{-1/2}\ell^{1/2}\varepsilon}{4}\Biggl),\label{eq:Igh 1}
\end{align}
where $\theta_{GH}:=1-\sup_{\|t\|\le C_{0}}P\left(\left|\sum_{g,h}\frac{1}{GH}(I_{gh}(t)-E\left[I_{gh}(t)\right])\right|>\frac{r_{GH}^{-1/2}\ell^{1/2}\varepsilon}{2}\right)$.
Applying Chebyshev's inequality with the bound of $\sup_{\|t\|\le C_{0}}\mathrm{Var}\!\left(\frac{1}{GH}\sum_{g,h}I_{gh}(t)\right)$
in \eqref{eq: bound of Var I_gh} yields that 
\begin{align*}
\sup_{\|t\|\le C_{0}}P\left(\left|\sum_{g,h}\frac{1}{GH}(I_{gh}(t)-E\left[I_{gh}(t)\right])\right|>\frac{r_{GH}^{-1/2}\ell^{1/2}\varepsilon}{2}\right)\le & \sup_{\|t\|\le C_{0}}\frac{4Var\left(\sum_{g,h}\frac{1}{GH}I_{gh}(t)\right)}{r_{GH}^{-1}\ell\varepsilon^{2}}\\
\lesssim & \left(\frac{1}{G}+\frac{1}{H}\right)\frac{1}{\ell}+\frac{r_{GH}^{1/2}}{GH}\frac{1}{\ell}=o\left(1\right),
\end{align*}
where the last equality holds given $R\ell\to\infty$ and $r_{GH}^{1/2}=O\left(\sqrt{GH}\right)$.
Therefore, $\theta_{GH}>1/2$ as $R\to\infty$.

Now condition on the data $\mathcal{G}_{GH}:=\{(X_{gh},e_{gh}):1\le g\le G,\,1\le h\le H\}$
and fix $G$ and $H$. At most finite elements are contained in the
functional set $\left\{ \left\{ \eta_{gh}\right\} \mapsto\frac{1}{GH}\sum_{g,h}\eta_{gh}I_{gh}(t):\left\Vert t\right\Vert \le C_{0}\right\} $,
since every element is of the form $\left\{ \eta_{gh}\right\} \mapsto\frac{1}{GH}\sum_{\left(g,h\right)\in S\left(t\right)}\frac{1}{2}\mathfrak{X}_{gh}$,
where $S\left(t\right)$ is a subset of $\left\{ 1,\ldots,G\right\} \times\left\{ 1,\ldots,H\right\} $.
Let $J_{GH}$ be the cardinality of this set. Then the conditional
supremum is a maximum over $J_{GH}$ elements, so by union bound 
\begin{align}
P\!\left(\sup_{\|t\|\le C_{0}}\left|\frac{1}{GH}\sum_{g,h}\eta_{gh}I_{gh}(t)\right|>\frac{r_{GH}^{-1/2}\ell^{1/2}\varepsilon}{4}\,\Big|\,\mathcal{G}_{GH}\right)\nonumber \\
\le\sum_{j=1}^{J_{GH}}P\!\Biggl( & \left|\frac{1}{GH}\sum_{g,h}\eta_{gh}I_{gh}(t_{j})\right|>\frac{r_{GH}^{-1/2}\ell^{1/2}\varepsilon}{4}\,\Big|\,\mathcal{G}_{GH}\Biggl),\label{eq:Igh 2}
\end{align}
for some representatives $\{t_{j}\}_{j=1}^{J_{GH}}$.

For each fixed $t$, conditional on $\mathcal{G}_{GH}$, the variables
$\eta_{gh}I_{gh}(t)$ are independent over $g$ and $h$ and bounded.
Thus Hoeffding's inequality gives 
\[
P\!\left(\left|\sum_{g,h}\frac{1}{GH}\eta_{gh}I_{gh}(t)\right|>\frac{r_{GH}^{-1/2}\ell^{1/2}\varepsilon}{4}\,\Big|\,\mathcal{G}_{GH}\right)\le2\exp\!\left(-\frac{GHr_{GH}^{-1}\ell\,\varepsilon^{2}}{8\,\nu_{GH}}\right).
\]
where $\nu_{GH}:=\frac{1}{GH}\sum_{g,h}\mathfrak{X}_{gh}^{2}\mathbf{1}\{\ell<e_{gh}\le\ell+r_{GH}^{-1/2}\left\Vert X_{gh}\right\Vert C_{0}\}.$
Next we bound $J_{GH}$ via VC theory as $G$ and $H$ grow. The collection
$\{(x,e):\ell<e\le\ell+r_{GH}^{-1/2}x^{\top}t,\|t\|\le C_{0}\}$ is
a VC class of sets with some finite dimension $V_{\mathcal{J}}\in(0,\infty)$
by Lemma 2.6.15 of \citet{van1996weak}. Hence Sauer's lemma yields
$J_{GH}\le C_{1}(GH)^{V_{\mathcal{J}}-1}$.

Combining the above, 
\begin{equation}
P\!\left(\sup_{\|t\|\le C_{0}}\left|\frac{1}{GH}\sum_{g,h}\eta_{gh}I_{gh}(t)\right|>\frac{r_{GH}^{-1/2}\ell^{1/2}\varepsilon}{4}\,\Big|\,\mathcal{G}_{GH}\right)\le2C_{1}(GH)^{V_{\mathcal{J}}-1}\exp\!\left(-\frac{GHr_{GH}^{-1}\ell\,\varepsilon^{2}}{8\,\nu_{GH}}\right).\label{eq: Agh3}
\end{equation}
Finally, define the event $E_{GH}:=\Big\{\nu_{GH}>\frac{GHr_{GH}^{-1}\varepsilon^{2}\ell}{8V_{\mathcal{J}}\log(GH)}\Big\}$.
By the law of total probability we split the unconditional probability
into the contributions from $E_{GH}$ and $E_{GH}^{c}$: 
\begin{align}
 & P\left(\sup_{\|t\|\le C_{0}}\left|\frac{1}{GH}\sum_{g=1}^{G}\sum_{h=1}^{H}\eta_{gh}I_{gh}\left(t\right)\right|>\frac{r_{GH}^{-1/2}\ell^{1/2}\varepsilon}{4}\right)\nonumber \\
 & =E\left(P\left(\sup_{\|t\|\le C_{0}}\left|\frac{1}{GH}\sum_{g=1}^{G}\sum_{h=1}^{H}\eta_{gh}I_{gh}\left(t\right)\right|>\frac{r_{GH}^{-1/2}\ell^{1/2}\varepsilon}{4}\vert\mathcal{G}_{GH}\right)\mathbf{1}\left(E_{GH}\right)\right)\nonumber \\
 & \quad+E\left(P\left(\sup_{\|t\|\le C_{0}}\left|\frac{1}{GH}\sum_{g=1}^{G}\sum_{h=1}^{H}\eta_{gh}I_{gh}\left(t\right)\right|>\frac{r_{GH}^{-1/2}\ell^{1/2}\varepsilon}{4}\vert\mathcal{G}_{GH}\right)\mathbf{1}\left(E_{GH}^{c}\right)\right).\label{eq:Igh 3}
\end{align}
On $E_{GH}^{c}$ we have $V_{\mathcal{J}}\log(GH)\le\frac{GHr_{GH}^{-1}\varepsilon^{2}\ell}{8\nu_{GH}}$,
hence $(GH)^{-V_{\mathcal{J}}}\ge\exp\left(-\frac{GHr_{GH}^{-1}\varepsilon^{2}\ell}{8\nu_{GH}}\right)$.
Therefore, together with \eqref{eq: Agh3}, we have 
\begin{align}
 & E\left(P\left(\sup_{\|t\|\le C_{0}}\left|\frac{1}{GH}\sum_{g=1}^{G}\sum_{h=1}^{H}\eta_{gh}I_{gh}\left(t\right)\right|>\frac{r_{GH}^{-1/2}\ell^{1/2}\varepsilon}{4}\vert\mathcal{G}_{GH}\right)\mathbf{1}\left(E_{GH}^{c}\right)\right)\nonumber \\
 & \leq E\!\left[2C_{1}(GH)^{V_{\mathcal{J}}-1}\exp\!\left(-\frac{GHr_{GH}^{-1}\varepsilon^{2}\ell}{8\nu_{GH}}\right)\mathbf{1}(E_{GH}^{c})\right]\le2C_{1}(GH)^{-1}\to0.\label{eq:Igh 4}
\end{align}
On $E_{GH}$, for some $\delta>0$, Markov's inequality, the expectation
$E\left(\nu_{GH}\right)\le C_{0}\,r_{GH}^{-1/2}\,E\![\mathfrak{X}_{gh}^{2}\|X_{gh}\|G_{0}(X_{gh})]\lesssim r_{GH}^{-1/2}$,
 $\frac{GH\ell^{2}}{(\log (GH))^{2}}>\frac{R\ell^{2}}{(\log R)^{2}}\to\infty$, and $r_{GH}=O(GH)$ together
imply that 
\begin{align}
 & E\left(P\left(\sup_{\|t\|\le C_{0}}\left|\frac{1}{GH}\sum_{g=1}^{G}\sum_{h=1}^{H}\eta_{gh}I_{gh}\left(t\right)\right|>\frac{r_{GH}^{-1/2}\ell^{1/2}\varepsilon}{4}\vert\mathcal{G}_{GH}\right)\mathbf{1}\left(E_{GH}\right)\right)\nonumber \\
\le & P(E_{GH})=P\!\left(\nu_{GH}\log(GH)\ell^{-1}\frac{r_{GH}}{GH}>\frac{\varepsilon^{2}}{8V_{\mathcal{J}}}\right)\lesssim\ell^{-1}\log(GH)\,r_{GH}^{-1/2}\frac{r_{GH}}{GH}=o(1).\label{eq:Igh 5}
\end{align}

Collecting terms \eqref{eq:Igh 1}, \eqref{eq:Igh 3}-\eqref{eq:Igh 5} and $\theta_{GH}>1/2$ yields that $\sup_{\|t\|\le C_{0}}|\frac{1}{GH}\sum_{g,h}(I_{gh}(t)-E\left[I_{gh}(t)\right])|=o_{P}\left(r_{GH}^{-1/2}\ell^{1/2}\right).$
Therefore $D_{1,GH}\left(t\right)-E\left[D_{1,GH}\left(t\right)\right]=o_{P}(r_{GH}^{-1/2}\ell^{1/2})$ uniformly over
$\|t\|\le C_{0}$. The same argument applies to $D_{2,GH}$,$D_{3,GH}$,
and $D_{4,GH}$ which further implies the desirable result. 
\end{proof}
\begin{lemma}[Stochastic equicontinuity of $\nu_{S}(\beta)$] \label{lemma: stochastic equi main normal}Under
Assumptions of Theorem \ref{thm:1}, let $\nu_{S}\left(\beta\right)=r_{GH}^{1/2}\left(\mathcal{S}\left(\beta\right)-\mathbb{S}\left(\beta\right)\right)$,
where 
\[
\mathbb{S}(\beta)=\frac{1}{GH}\sum_{g=1}^{G}\sum_{h=1}^{H}\psi_{gh}(\beta)=\frac{1}{GH}\sum_{g=1}^{G}\sum_{h=1}^{H}X_{gh}\Bigl(\tau-\mathbf{1}\{y_{gh}\le X_{gh}^{\top}\beta\}\Bigr)
\]
and $\mathcal{S}\left(\beta\right)=E\left[\mathbb{S}(\beta)\right]$ (w.l.o.g, let $d=1$).
Then, for all $\eta>0$ and $\varepsilon>0$, there is some $\delta>0$
such that 
\[
\limsup_{G,H\to\infty}P\left[\sup_{\left\vert \beta_{1}-\beta_{2}\right\vert \le\delta}\left\vert \nu_{S}\left(\beta_{1}\right)-\nu_{S}\left(\beta_{2}\right)\right\vert >\eta\right]\le\varepsilon.
\]
\end{lemma} 
\begin{proof}
It is standard to show $\widehat{\beta}\overset{P}{\to}\beta_{0}$,
so we omit the proof. By a Hoeffding-type decomposition, 
\begin{align*}
&\nu_{S}(\beta)  =\nu_{S}^{(\mathrm{I})}(\beta)+\nu_{S}^{(\mathrm{II})}(\beta)+\nu_{S}^{(\mathrm{III})}(\beta)+\nu_{S}^{(\mathrm{IV})}(\beta)\\
 & :=\sqrt{\frac{r_{GH}}{G}}\frac{1}{\sqrt{G}}\sum_{g=1}^{G}\psi_{g}^{(\mathrm{I})}(\beta)+\sqrt{\frac{r_{GH}}{H}}\frac{1}{\sqrt{H}}\sum_{h=1}^{H}\psi_{h}^{(\mathrm{II})}(\beta)+\sqrt{\frac{r_{GH}}{GH}}\frac{1}{\sqrt{GH}}\sum_{g=1}^{G}\sum_{h=1}^{H}\Bigl(\psi_{gh}^{(\mathrm{III})}(\beta)+\psi_{gh}^{(\mathrm{IV})}(\beta)\Bigr),
\end{align*}
where $\psi_{gh}(\beta):=X_{gh}\Bigl(\tau-\mathbf{1}\{y_{gh}\le X_{gh}^{\top}\beta\}\Bigr)$
and 
\begin{align*}
\psi_{g}^{(\mathrm{I})}(\beta) & =E\!\bigl(\psi_{gh}(\beta)\vert U_{g}\bigr)-E\!\bigl(\psi_{gh}(\beta)\bigr),\\
\psi_{h}^{(\mathrm{II})}(\beta) & =E\!\bigl(\psi_{gh}(\beta)\vert V_{h}\bigr)-E\!\bigl(\psi_{gh}(\beta)\bigr),\\
\psi_{gh}^{(\mathrm{III})}(\beta) & =E\!\bigl(\psi_{gh}(\beta)\vert U_{g},V_{h}\bigr)-E\!\bigl(\psi_{gh}(\beta)\vert U_{g}\bigr)-E\!\bigl(\psi_{gh}(\beta)\vert V_{h}\bigr)+E\!\bigl(\psi_{gh}(\beta)\bigr),\\
\psi_{gh}^{(\mathrm{IV})}(\beta) & =\psi_{gh}(\beta)-E\!\bigl(\psi_{gh}(\beta)\vert U_{g},V_{h}\bigr).
\end{align*}

\emph{Case 1: $H\sigma_{\mathrm{I},1\Gamma}^{2}+G\sigma_{\mathrm{II},1\Gamma}^{2}\to\infty$.}
Without loss of generality assume $H\sigma_{\mathrm{I},1\Gamma}^{2}\ge G\sigma_{\mathrm{II},1\Gamma}^{2}$,
so that $r_{GH}=G/\sigma_{\mathrm{I},1\Gamma}^{2}$ and $\sqrt{r_{GH}/G}\,\sigma_{\mathrm{I},1\Gamma}=1$.
Moreover, 
\[
\sqrt{\frac{r_{GH}}{H}}\,\sigma_{\mathrm{II},\Gamma}=\sqrt{\frac{G}{H}}\frac{\sigma_{\mathrm{II},\Gamma}}{\sigma_{\mathrm{I},\Gamma}}\in[0,1],\qquad\sqrt{\frac{r_{GH}}{GH}}\to0.
\]
Then the usual i.i.d.-in-$g$ equicontinuity argument (see, e.g.,
Corollary 3.1 of \citet{newey1991uniform}) yields, for some $\delta>0$,
\[
\limsup_{G,H\to\infty}P\!\left(\sup_{\vert\beta_{1}-\beta_{2}\vert\le\delta}\left\vert \frac{1}{\sqrt{G}}\sum_{g=1}^{G}\sigma_{\mathrm{I},\Gamma}^{-1}\psi_{g}^{(\mathrm{I})}(\beta_{1})-\frac{1}{\sqrt{G}}\sum_{g=1}^{G}\sigma_{\mathrm{I},\Gamma}^{-1}\psi_{g}^{(\mathrm{I})}(\beta_{2})\right\vert >\eta\right)\le\varepsilon.
\]
Similarly, provided that $\sqrt{\frac{G}{H}}\frac{\sigma_{\mathrm{II},\Gamma}}{\sigma_{\mathrm{I},\Gamma}}\le1$,
one can deduce 
\begin{align*}
 & \limsup_{G,H\to\infty}P\!\left(\sup_{\vert\beta_{1}-\beta_{2}\vert\le\delta}\left\vert \sqrt{\frac{r_{GH}}{H}}\sigma_{\mathrm{II},\Gamma}\left(\frac{1}{\sqrt{H}}\sum_{h=1}^{H}\sigma_{\mathrm{II},\Gamma}^{-1}\psi_{h}^{(\mathrm{II})}(\beta_{1})-\frac{1}{\sqrt{H}}\sum_{h=1}^{H}\sigma_{\mathrm{II},\Gamma}^{-1}\psi_{h}^{(\mathrm{II})}(\beta_{2})\right)\right\vert >\eta\right)\\
\le & \limsup_{G,H\to\infty}P\!\left(\sup_{\vert\beta_{1}-\beta_{2}\vert\le\delta}\left\vert \frac{1}{\sqrt{H}}\sum_{h=1}^{H}\sigma_{\mathrm{II},\Gamma}^{-1}\psi_{h}^{(\mathrm{II})}(\beta_{1})-\frac{1}{\sqrt{H}}\sum_{h=1}^{H}\sigma_{\mathrm{II},\Gamma}^{-1}\psi_{h}^{(\mathrm{II})}(\beta_{2})\right\vert >\eta\right)\le\varepsilon.
\end{align*}
The rest terms are negligible since $\sqrt{\frac{r_{GH}}{GH}}\to0$.

\emph{Case 2: $r_{GH}\asymp GH$.} In such case, we have $H\sigma_{\mathrm{I},1\Gamma}^{2}+G\sigma_{\mathrm{II},1\Gamma}^{2}=O(1)$
and $\sigma_{\mathrm{III},1\Gamma}^{2}=o(1)$. The result for the
first two terms $\nu_{S}^{\mathrm{I}}(\beta)$ and $\nu_{S}^{\mathrm{II}}(\beta)$
follows Case 1. The third term $\nu_{S}^{\mathrm{III}}(\beta)$ is
also negligible given $\sigma_{\mathrm{III},1\Gamma}^{2}=o(1)$. It
suffices to show $\nu_{S}^{(\mathrm{IV})}(\beta)$ is stochastically
equicontinuous, i.e., for any $\varepsilon,\eta>0$ there exists $\delta>0$
such that 
\begin{equation}
\limsup_{G,H\to\infty}P\!\left(\sup_{\vert\beta_{1}-\beta_{2}\vert\le\delta}\left\vert \frac{1}{\sqrt{GH}}\sum_{g=1}^{G}\sum_{h=1}^{H}\psi_{gh}^{(\mathrm{IV})}(\beta_{1})-\frac{1}{\sqrt{GH}}\sum_{g=1}^{G}\sum_{h=1}^{H}\psi_{gh}^{(\mathrm{IV})}(\beta_{2})\right\vert >\eta\right)\le\varepsilon.\label{eq:equic1}
\end{equation}
Conditional on $\{(U_{g},V_{h})\}$, $\{\psi_{gh}^{(\mathrm{IV})}(\beta)\}_{g,h}$
are independent across $(g,h)$ and satisfy
the uniform conditional second-moment bound 
\[
\sup_{u,v}E\!\bigl(\vert\psi_{gh}^{(\mathrm{IV})}(\beta)\vert^{2}\vert U_{g}=u,V_{h}=v\bigr)\;\le\;\sup_{u,v}E\!\bigl(\vert X_{gh}\vert^{2}\vert U_{g}=u,V_{h}=v\bigr)\;<\infty.
\]
Therefore, a standard VC-type argument yields that conditional on $\{(U_{g},V_{h})\}$, 
\[
P\!\left(\sup_{\vert\beta_{1}-\beta_{2}\vert\le\delta}\left\vert \frac{1}{\sqrt{GH}}\sum_{g=1}^{G}\sum_{h=1}^{H}\psi_{gh}^{(\mathrm{IV})}(\beta_{1})-\frac{1}{\sqrt{GH}}\sum_{g=1}^{G}\sum_{h=1}^{H}\psi_{gh}^{(\mathrm{IV})}(\beta_{2})\right\vert >\eta\ \Bigm|\ \{(U_{g},V_{h})\}\right)\le\varepsilon,
\]
and \eqref{eq:equic1} follows by the law of total probability. 
\end{proof}
\begin{lemma}[Order of oracle variance] \label{lemma: order of oracle variance}Under
Assumptions of Theorem \ref{thm:2-1} holds, we have 
\begin{align*}
\widetilde{\Omega}_{\mathrm{I}}-\frac{\sigma_{\mathrm{I},\Gamma}^{2}}{G} & =O_{P}\left(G^{-3/2}\sigma_{\mathrm{I},\Gamma}^{2}+G^{-1}H^{-1}\sigma_{\mathrm{III},\Gamma}^{2}+G^{-3/2}H^{-1}\sigma_{\mathrm{IV},\Gamma}^{2}\right)+o_P(r_{GH}^{-1}),\\
\widetilde{\Omega}_{\mathrm{II}}-\frac{\sigma_{\mathrm{II},\Gamma}^{2}}{H} & =O_{P}\left(H^{-3/2}\sigma_{\mathrm{II},\Gamma}^{2}+G^{-1}H^{-1}\sigma_{\mathrm{III},\Gamma}^{2}+G^{-1}H^{-3/2}\sigma_{\mathrm{IV},\Gamma}^{2}\right)+o_P(r_{GH}^{-1}),\\
\widetilde{\Omega}_{\mathrm{III,IV}}-\frac{1}{GH}\left(\sigma_{\mathrm{III},\Gamma}^{2}+\sigma_{\mathrm{IV},\Gamma}^{2}\right) & =O_{P}\left(\left(GH\right)^{-1}\left(G^{-1/2}+H^{-1/2}\right)\sigma_{\mathrm{III},\Gamma}^{2}+G^{-3/2}H^{-3/2}\sigma_{\mathrm{IV},\Gamma}^{2}\right)+o_P(r_{GH}^{-1}),
\end{align*}
\end{lemma} where $\widetilde{\Omega}_{\mathrm{I}}$, $\widetilde{\Omega}_{\mathrm{II}}$,
and $\widetilde{\Omega}_{\mathrm{III,IV}}$ are defined in \eqref{eq: omega 1},
\eqref{eq: omega 2}, and \eqref{eq: omega 3}, respectively. 
\begin{proof}
Rearranging terms, we can write 
\begin{align*}
\widetilde{\Omega}_{\mathrm{I}}=\frac{1}{G^{2}H^{2}}\sum_{g=1}^{G}\sum_{h=1}^{H}\sum_{h\neq h'}^{H}\Psi_{gh}\Psi_{gh'}^{\top}= & \frac{1}{G^{2}}\sum_{g=1}^{G}\widetilde{\Psi}_{g}^{(\mathrm{I})}\widetilde{\Psi}_{g}^{(\mathrm{I})\top}-\frac{1}{G^{2}H^{2}}\sum_{g=1}^{G}\sum_{h=1}^{H}\widetilde{\Psi}_{gh}^{(\mathrm{III})}\widetilde{\Psi}_{gh}^{(\mathrm{III})\top}+\widetilde{\mathfrak{R}}_{\mathrm{I},GH},
\end{align*}
where 
\begin{align*}
\widetilde{\Psi}_{g}^{(\mathrm{I})}= & \frac{1}{H}\sum_{h=1}^{H}\Psi_{gh}-\frac{1}{GH}\sum_{g=1}^{G}\sum_{h=1}^{H}\Psi_{gh},\\
\widetilde{\Psi}_{h}^{(\mathrm{II})}= & \frac{1}{G}\sum_{g=1}^{G}\Psi_{gh}-\frac{1}{GH}\sum_{g=1}^{G}\sum_{h=1}^{H}\Psi_{gh},\\
\widetilde{\Psi}_{gh}^{(\mathrm{III})}= & \Psi_{gh}-\frac{1}{G}\sum_{g=1}^{G}\Psi_{gh}-\frac{1}{H}\sum_{h=1}^{H}\Psi_{gh}+\frac{1}{GH}\sum_{g=1}^{G}\sum_{h=1}^{H}\Psi_{gh},\\
\widetilde{\mathfrak{R}}_{\mathrm{I},GH}= & -\frac{1}{G^{2}H}\sum_{g=1}^{G}\widetilde{\Psi}_{g}^{(\mathrm{I})}\widetilde{\Psi}_{g}^{(\mathrm{I})\top}-\frac{1}{GH^{2}}\sum_{h=1}^{H}\widetilde{\Psi}_{h}^{(\mathrm{II})}\widetilde{\Psi}_{h}^{(\mathrm{II})\top}+\frac{H-1}{GH}\bar{\Psi}\bar{\Psi}^{\top}=o_{P}\left(r_{GH}^{-1}\right),\\
\bar{\Psi}= & \frac{1}{GH}\sum_{g=1}^{G}\sum_{h=1}^{H}\widetilde{\Psi}_{gh}.
\end{align*}
Similarly, we have 
\begin{align*}
\widetilde{\Omega}_{\mathrm{II}}=\frac{1}{G^{2}H^{2}}\sum_{g=1}^{G}\sum_{g'\neq g}^{G}\sum_{h=1}^{H}\Psi_{gh}\Psi_{g'h}^{\top}= & \frac{1}{H^{2}}\sum_{h=1}^{H}\widetilde{\Psi}_{h}^{(\mathrm{II})}\widetilde{\Psi}_{h}^{(\mathrm{II})\top}-\frac{1}{G^{2}H^{2}}\sum_{g=1}^{G}\sum_{h=1}^{H}\widetilde{\Psi}_{gh}^{(\mathrm{III})}\widetilde{\Psi}_{gh}^{(\mathrm{III})\top}+o_{P}\left(r_{GH}^{-1}\right),
\end{align*}
and
\begin{align*}
\widetilde{\Omega}_{\mathrm{III,IV}}=\frac{1}{G^{2}H^{2}}\sum_{g=1}^{G}\sum_{h=1}^{H}\Psi_{gh}\Psi_{gh}^{\top}= & \ \frac{1}{G^{2}H^{2}}\sum_{g=1}^{G}\sum_{h=1}^{H}\widetilde{\Psi}_{gh}^{(\mathrm{III})}\widetilde{\Psi}_{gh}^{(\mathrm{III})\top}\\
 & \quad+\frac{1}{G^{2}H}\sum_{g=1}^{G}\widetilde{\Psi}_{g}^{(\mathrm{I})}\widetilde{\Psi}_{g}^{(\mathrm{I})\top}+\frac{1}{GH^{2}}\sum_{h=1}^{H}\widetilde{\Psi}_{h}^{(\mathrm{II})}\widetilde{\Psi}_{h}^{(\mathrm{II})\top}+\frac{1}{GH}\bar{\Psi}\,\bar{\Psi}^{\top}\\
= & \frac{1}{G^{2}H^{2}}\sum_{g=1}^{G}\sum_{h=1}^{H}\widetilde{\Psi}_{gh}^{(\mathrm{III})}\widetilde{\Psi}_{gh}^{(\mathrm{III})\top}+o_{P}(r_{GH}^{-1}).
\end{align*}
The proof then follows Lemma A.1 in \citet{menzel2021bootstrap} and
hence is omitted. 
\end{proof}
\clearpage{}

\section{Proof of Proposition \ref{prop:impossibility_uniform}}

\paragraph{Step 1: Data generating process with local parameter.}

Consider the scalar median regression model \((\tau=1/2)\)
\begin{equation}
y_{gh}=X_{gh}\beta_{0}+e_{gh},\qquad g=1,\dots,G,\ \ h=1,\dots,H.
\label{eq:qr_dgp_model}
\end{equation}
Let the regressor have the two-way factor structure
\begin{equation}
X_{gh}=U_g^xV_h^x,\qquad U_g^x\stackrel{i.i.d.}{\sim}U(-1,1),\quad V_h^x\stackrel{i.i.d.}{\sim}U(0,2),
\label{eq:qr_dgp_x}
\end{equation}
where the row and column sequences are mutually independent. Let \(\{U_g^e\}_{g=1}^G\) be i.i.d.\ Rademacher signs with \(\Pr(U_g^e=1)=\Pr(U_g^e=-1)=1/2\). Let \(\{V_h^e\}_{h=1}^H\) be i.i.d.\ signs satisfying
\begin{equation}
\Pr(V_h^e=1)=\frac12+\frac{c}{2\sqrt H},\qquad
\Pr(V_h^e=-1)=\frac12-\frac{c}{2\sqrt H},
\label{eq:qr_dgp_r}
\end{equation}
where \(c\ge0\) is fixed. For \(c=1\) or \(c=2\), these probabilities are well-defined for all large \(H\). Finally, let \(\{\varepsilon_{gh}\}_{g\le G,h\le H}\) be i.i.d.\ \(N(0,1)\), independent of all row and column factors. For a fixed constant \(a>0\), define
\begin{equation}
e_{gh}=aU_g^eV_h^e+\varepsilon_{gh}.
\label{eq:qr_dgp_u_smooth}
\end{equation}

This DGP satisfies the AHK representation by collecting \((U_g^x,U_g^e)\) into the row latent variable and \((V_h^x,V_h^e)\) into the column latent variable. The conditional median restriction also holds. Indeed, since \(U_g^e\) is symmetric and independent of \(X_{gh}\), the sign \(S_{gh}:=U_g^eV_h^e\) satisfies \(\Pr(S_{gh}=1\mid X_{gh})=\Pr(S_{gh}=-1\mid X_{gh})=1/2\). Therefore
\[
P(e_{gh}\le0\mid X_{gh})
=\frac12\Phi(-a)+\frac12\Phi(a)
=\frac12,
\]
so \(Q_{e_{gh}}(1/2\mid X_{gh})=0\).

The density conditions are also satisfied. Conditional on \(U_g^e,V_h^e\), the density of \(e_{gh}\) is \(e\mapsto\phi(e-aU_g^eV_h^e)\), which is smooth with bounded derivatives of all orders. Conditional on \(X_{gh}\), the density is the symmetric normal mixture \(f_{e\mid X}(e\mid X_{gh})=\frac12\phi(e-a)+\frac12\phi(e+a)\). Hence \(f_{e\mid X}(0\mid X_{gh})=\phi(a)>0\), uniformly in \(X_{gh}\), and \(f_{e\mid X}\), \(f_{e\mid X}^{(1)}\), and \(f_{e\mid X}^{(2)}\) are uniformly bounded. The conditional joint densities required in Assumption \(\ref{as:moment-strong}\) are finite mixtures of products of shifted normal densities, so their first and second partial derivatives are bounded by integrable normal-polynomial envelopes. Since \(|X_{gh}|\le2\), all required moment conditions hold. Moreover, \(E[X_{gh}^{2}]=E[(U_g^x)^2]E[(V_h^x)^2]=4/9>0\), so the nonsingularity conditions hold. Finally, the variance lower-bound condition in the definition of \(\mathcal B_3\) holds because the interaction component and the idiosyncratic component below have nonzero variances. Thus the DGP belongs to \(\mathcal B_3\).

\paragraph{Step 2: Limit distribution.}

Define the median score at \(\beta\) by
\(\psi_{gh}(\beta):=X_{gh}\bigl(1/2-\mathbf{1}\{y_{gh}\le X_{gh}\beta\}\bigr)\). At the truth \(\beta=\beta_0\), write \(S_{gh}:=U_g^eV_h^e\) and define \(\kappa:=\Phi(a)-1/2>0\). Since
\[
E\bigl[1/2-\mathbf{1}\{aS_{gh}+\varepsilon_{gh}\le0\}\mid S_{gh}\bigr]
=\Phi(aS_{gh})-\frac12
=\kappa S_{gh},
\]
we can write
\[
\frac12-\mathbf{1}\{e_{gh}\le0\}
=
\kappa U_g^eV_h^e+\xi_{gh},
\qquad
E[\xi_{gh}\mid U_g^e,V_h^e,U_g^x,V_h^x]=0.
\]
Therefore
\begin{equation}
\psi_{gh}(\beta_0)
=
\kappa(U_g^xU_g^e)(V_h^xV_h^e)+U_g^xV_h^x\xi_{gh}.
\label{eq:qr_score_smooth_decomp}
\end{equation}

The Jacobian is
\[
D:=E[f_{e\mid X}(0\mid X_{gh})X_{gh}^{2}]
=
\phi(a)E[(U_g^x)^2]E[(V_h^x)^2]
=
\frac49\phi(a).
\]
By the standard Bahadur expansion for median regression,
\begin{equation}
\sqrt{GH}(\widehat\beta-\beta_0)
=
D^{-1}\frac1{\sqrt{GH}}\sum_{g=1}^{G}\sum_{h=1}^{H}\psi_{gh}(\beta_0)+o_p(1).
\label{eq:qr_bahadur_smooth}
\end{equation}

Using \eqref{eq:qr_score_smooth_decomp},
\[
\frac1{\sqrt{GH}}\sum_{g=1}^{G}\sum_{h=1}^{H}\psi_{gh}(\beta_0)
=
\kappa\left(\frac1{\sqrt G}\sum_{g=1}^{G}U_g^xU_g^e\right)
\left(\frac1{\sqrt H}\sum_{h=1}^{H}V_h^xV_h^e\right)
+
\frac1{\sqrt{GH}}\sum_{g=1}^{G}\sum_{h=1}^{H}U_g^xV_h^x\xi_{gh}.
\]
For the first factor, since \(E[U_g^xU_g^e]=0\) and \(E[(U_g^xU_g^e)^2]=1/3\), the CLT gives
\[
\frac1{\sqrt G}\sum_{g=1}^{G}U_g^xU_g^e\overset{d}{\to}Z_U,
\qquad Z_U\sim N(0,1/3).
\]
For the second factor, \(E[V_h^xV_h^e]=E[V_h^x]E[V_h^e]=c/\sqrt H\). Hence
\[
\frac1{\sqrt H}\sum_{h=1}^{H}V_h^xV_h^e
=
\frac1{\sqrt H}\sum_{h=1}^{H}\{V_h^xV_h^e-E[V_h^xV_h^e]\}
+
c.
\]
Since \(E[(V_h^xV_h^e)^2]\to E[(V_h^x)^2]=4/3\), the CLT yields
\[
\frac1{\sqrt H}\sum_{h=1}^{H}V_h^xV_h^e\overset{d}{\to}Z_V+c,
\qquad Z_V\sim N(0,4/3).
\]
The two limits \(Z_U\) and \(Z_V\) are independent because they are functions of independent row and column factors.

It remains to handle the last term in \eqref{eq:qr_score_smooth_decomp}. Conditional on the row and column factors, \(\{U_g^xV_h^x\xi_{gh}\}_{g,h}\) are independent, mean zero, and uniformly bounded in moments. Moreover,
\[
Var(\xi_{gh}\mid U_g^e,V_h^e,U_g^x,V_h^x)
=
\frac14-\kappa^2,
\]
which does not depend on \(g,h\). Therefore, by the conditional Lindeberg CLT,
\[
\frac1{\sqrt{GH}}\sum_{g=1}^{G}\sum_{h=1}^{H}U_g^xV_h^x\xi_{gh}
\overset{d}{\to}Z_0,
\qquad
Z_0\sim N\left(0,\frac49\left(\frac14-\kappa^2\right)\right).
\]
The convergence is stable with respect to the row and column factors, so \(Z_0\) is independent of \((Z_U,Z_V)\). Combining the preceding displays gives
\begin{equation}
\sqrt{GH}(\widehat\beta-\beta_0)
\overset{d}{\to}
\frac{9}{4\phi(a)}
\left\{
\kappa Z_U(Z_V+c)+Z_0
\right\}.
\label{eq:qr_limit_smooth}
\end{equation}

The limit distribution depends on \(c\). In particular, its variance equals
\[
\left(\frac{9}{4\phi(a)}\right)^2
\left[
\kappa^2\frac13\left(\frac43+c^2\right)
+
\frac49\left(\frac14-\kappa^2\right)
\right],
\]
which is different for \(c=1\) and \(c=2\). Hence the two limiting distribution functions, denoted by \(F_1\) and \(F_2\), are distinct. Since both are continuous, \(d_K(F_1,F_2):=\sup_t|F_1(t)-F_2(t)|>0\).

\paragraph{Step 3: Impossibility of consistent testing.}

Let \(\Gamma_c\) denote the DGP above with local parameter \(c\), and compare \(c=1\) with \(c=2\). The parameter \(c\) enters the DGP only through the distribution of the column signs \(\{V_h^e\}_{h=1}^{H}\). Conditional on these signs and on all other latent variables, the distribution of the observed sample is independent of \(c\). Thus the observed experiment is a garbling of the sign experiment \(\{V_h^e\}_{h=1}^{H}\).

Let \(P_c^V\) denote the law of \(\{V_h^e\}_{h=1}^{H}\) under \(\Gamma_c\). The Hellinger affinity between \(P_1^V\) and \(P_2^V\) is
\[
\rho(P_1^V,P_2^V)
=
\left[
\sqrt{\left(\frac12+\frac{1}{2\sqrt H}\right)\left(\frac12+\frac{2}{2\sqrt H}\right)}
+
\sqrt{\left(\frac12-\frac{1}{2\sqrt H}\right)\left(\frac12-\frac{2}{2\sqrt H}\right)}
\right]^H.
\]
A Taylor expansion gives \(\rho(P_1^V,P_2^V)\to \exp(-1/8)>0\). Since Hellinger affinity cannot decrease under a Markov kernel, the Hellinger affinity between the observed laws \(P_{\Gamma_1}\) and \(P_{\Gamma_2}\) is also bounded away from zero. Consequently, \(\|P_{\Gamma_1}-P_{\Gamma_2}\|_{TV}\) is bounded away from one. Hence no test based on the observed sample can consistently distinguish \(\Gamma_1\) from \(\Gamma_2\).

Now suppose, toward a contradiction, that there exists a measurable distribution estimator \(\widehat E\) that is uniformly consistent over \(\mathcal B_3\). Since the limiting distribution under \(\Gamma_1\) is \(F_1\) and the limiting distribution under \(\Gamma_2\) is \(F_2\), and since \(d_K(F_1,F_2)>0\), such an estimator would yield a consistent test: choose \(\Gamma_1\) if \(\sup_t|\widehat E(t)-F_1(t)|\le \sup_t|\widehat E(t)-F_2(t)|\), and choose \(\Gamma_2\) otherwise. Uniform consistency would make the probability of error tend to zero under both \(\Gamma_1\) and \(\Gamma_2\), contradicting the impossibility of consistent testing established above. Therefore no uniformly consistent distribution estimator exists over \(\mathcal B_3\).

 \bibliographystyle{chicago}
\bibliography{multiway_clustering}

\end{document}